\documentclass{lmcs}
\pdfoutput=1
\usepackage[utf8]{inputenc}

\usepackage{lastpage}
\lmcsdoi{20}{2}{13}
\lmcsheading{}{\pageref{LastPage}}{}{}%
{Sep.~17,~2021}{Jun.~13,~2024}{}


\ACMCCS{[{\bf Theory of computation}]:~Logic and verification; Modal and temporal logics;} %[{\bf Mathematics of computing}]: CCC---DDD}
\usepackage{hyperref}

\theoremstyle{defC}
\newtheorem{exaC}[thm]{Example}

\theoremstyle{plain}

\usepackage{stmaryrd,amssymb}
\usepackage[curve,cmtip]{xypic}

\usepackage{scalerel}
\usepackage{stackengine,wasysym}
\newcommand\reallywidetilde[1]{\ThisStyle{%
  \setbox0=\hbox{$\SavedStyle#1$}%
  \stackengine{-.1\LMpt}{$\SavedStyle#1$}{%
    \stretchto{\scaleto{\SavedStyle\mkern.2mu\AC}{.5150\wd0}}{.6\ht0}%
  }{O}{c}{F}{T}{S}%
}}

\newcommand{\llsem}{\llparenthesis}
\newcommand{\rrsem}{\rrparenthesis}
\newcommand{\Cyl}{{\mathsf{Cyl}}}
\newcommand{\M}{{\mathcal{M}}}
\newcommand{\A}{{\mathcal{A}}}
\newcommand{\C}{{\mathcal{C}}}
\renewcommand{\H}{{\mathcal{H}}}
\renewcommand{\S}{{\mathcal{S}}}
\newcommand{\R}{{\mathcal{R}}}
\newcommand{\Paths}{{\mathsf{Paths}}}
\newcommand{\pref}{{\mathsf{pref}}}
\newcommand{\depth}{{\mathsf{depth}}}
\newcommand{\fnd}{{\mathsf{fnd}}}
\newcommand{\tr}{{\mathsf{tr}}}
\newcommand{\ftr}{{\mathsf{ftr}}}
\newcommand{\ptr}{{\mathsf{ptr}}}
\newcommand{\ran}{{\mathsf{ran}}}

\newcommand{\At}{{\mathsf{At}}}
\newcommand{\extent}{{\mathsf{e}}}
\newcommand{\arity}{{\mathsf{ar}}}
\newcommand{\op}{{\mathsf{op}}}
\newcommand{\V}{\mathcal{V}}
\newcommand{\Id}{{\mathsf{Id}}}
\newcommand{\Op}{{\mathsf{O}}}
\newcommand{\LOp}{{\mathsf{o}}}
\newcommand{\id}{{\mathsf{id}}}
\renewcommand{\phi}{\varphi}
\newcommand{\Cl}{{\mathsf{Cl}}}  
\renewcommand{\L}{{\mathsf L}}            
\newcommand{\E}{{\mathsf E}}
\newcommand{\LL}{{\mathcal L}}
\newcommand{\T}{{\mathsf T}}
\newcommand{\Set}{{\mathsf{Set}}}
\newcommand{\Rel}{{\mathsf{Rel}}}
\newcommand{\Pred}{{\mathsf{Pred}}}
\newcommand{\lsem}{\llbracket}
\newcommand{\rsem}{\rrbracket}
\newcommand{\st}{{\mathsf{st}}}
\newcommand{\dst}{{\mathsf{dst}}}
\newcommand{\tw}{{\mathsf{tw}}}
\newcommand{\Pow}{{\mathcal{P}}}
\newcommand{\FPow}{{\mathcal{P}_\omega}}
\newcommand{\SPred}{{\mathsf{Pred}}}
\newcommand{\supp}{{\mathsf{supp}}}
\newcommand{\Eq}{{\mathsf{Eq}}}
\renewcommand{\d}{{\mathsf{d}}}
\renewcommand{\next}{\bigcirc}

\allowdisplaybreaks

\begin{document}           
\bibliographystyle{alphaurl}

\title{Linear-time logics -- a coalgebraic perspective}
\author[C.~C\^{\i}rstea]{Corina C\^{\i}rstea\lmcsorcid{0000-0003-3165-5678}}	
\address{University of Southampton, UK}	
\email{cc2@ecs.soton.ac.uk}  
\thanks{Part of this work was funded by the Leverhulme Trust Research Project Grant RPG-2020-232.}	

\begin{abstract}
\noindent We describe a general approach to deriving linear-time logics for a wide variety of state-based, quantitative systems, by modelling the latter as coalgebras whose type incorporates both branching and linear behaviour. Concretely, we define logics whose syntax is determined by the type of linear behaviour, and whose domain of truth values is determined by the type of branching behaviour, and we provide two semantics for them: a step-wise semantics akin to that of standard coalgebraic logics, and a path-based semantics akin to that of standard linear-time logics. The former semantics is useful for model checking, whereas the latter is the more natural semantics, as it measures the extent with which qualitative properties hold along computation paths from a given state. Our main result is the equivalence of the two semantics. We also provide a semantic characterisation of a notion of logical distance induced by these logics. Instances of our logics support reasoning about the possibility, likelihood or minimal cost of exhibiting a given linear-time property.
\end{abstract}

\maketitle

\section{Introduction}

Linear-time temporal logics such as LTL (see e.g.~\cite{Vardi95}) or the linear-time $\mu$-calculus (see e.g.~\cite{Vardi88,Dam1992}), originally interpreted over non-deterministic models, have been successfully adapted to quantitative transition system models \cite{Vardi99,FAELLA200861,Droste2012}. These logics specify qualitative properties of paths in a (quantitative) transition system, and depending on the type of branching present in the models, have either a qualitative semantics (in the case of non-deterministic branching) or a quantitative one (in the case of probabilistic or weighted branching). Despite commonalities, which extend to the associated automata-based verification techniques, a general and uniform account of linear-time logics and their use in formal verification is still missing. Such an account would ultimately support the development of verification tools that are applicable to a wider class of models, thereby extending the scope of automated verification to include complex systems with a variety of quantitative features and correctness/optimality concerns. Extensions include (i) resource-aware systems, which associate \emph{costs} to the different actions a system can take, and where the goal is to achieve the desired behaviour with minimal cost, and (ii) systems whose executions are \emph{tree-shaped}, that is, an action can result in several successor states; the latter is needed e.g.~to model systems whose structure evolves dynamically (see (\ref{dyn}) of Example~\ref{ex:lin-beh}).

The present paper makes some steps towards filling this gap, by studying linear-time, quantitative logics for a variety of models within a unifying framework. Specifically, we model systems as \emph{coalgebras} whose type incorporates both \emph{branching} and \emph{linear} behaviour. These models come with a natural notion of \emph{path}, which formalises an individual system execution. We study linear-time, fixpoint logics for such models, that are naturally induced by the model type. This builds on a recent coalgebraic account of maximal traces in systems with branching \cite{Cirstea17}. The branching type of a model is determined by a choice of quantitative domain for the transition weights, in the form of a partial semiring; this provides the domain of truth values for the logics, and dictates how the semantics quantifies over the branching. The linear type of a model describes the structure of individual transitions, and induces a \emph{qualitative} notion of observable behaviour (of an individual execution), together with associated (both qualitative and quantitative) linear-time modal operators. The quantitative logics thus obtained have a natural step-wise semantics, as is standard in coalgebraic logics. However, this semantics does not directly capture the intuition that the (quantitative) interpretation of a formula measures the set of computation paths from a given state which satisfy the formula in the qualitative sense -- the semantics makes no reference to computation paths. Formalising this intuition requires an extension of standard measure-theoretic concepts and results to measures valued into partial semirings. Once this is done, a more natural, path-based semantics can be defined, and proved equivalent to the step-wise semantics.

For this equivalence result to hold, it is necessary to omit standard propositional operators (conjunctions and disjunctions) from the logics. To see why this is necessary, consider the LTL formula $\next p \wedge \next q$, where $\next$ is the LTL \emph{next} operator. Its interpretation in existential LTL is "there exists a path along which $p$ is true in the next step and $q$ is true in the next step". Under a step-wise semantics, the interpretation of this formula would be a function of the interpretations of $\next p$ and $\next q$. However, such a function cannot be defined: depending on the transition system, a path satisfying both $p$ and $q$ in the next step may or may not exist, when (potentially different) paths satisfying $p$ and respectively $q$ in the next step exist. The argument for the absence of arbitrary disjunctions is similar: in probabilistic LTL, there is no choice of interpretation for disjunctions of formulas which can compute the likelihood of $\next p \vee \next q$ holding along paths from a given state, from the likelihoods of $\next p$, respectively $\next q$ holding on these paths. While the absence of conjunctions and disjunctions from the syntax of the logics appears rather restrictive at first, this has relatively limited impact on the expressiveness of the logics: when interpreted qualitatively, our logics have the expressive power of \emph{deterministic} parity automata; these, in turn, are as expressive as non-deterministic parity automata in the case of \emph{word-like} linear behaviour, but less expressive than non-deterministic parity automata in the case of \emph{tree-like} linear behaviour \cite{TW68} (see also Remark~\ref{rem-disj}). As a result, as far as existential LTL and probabilistic LTL are concerned, our logics are equally expressive; in fact they are slightly \emph{more} expressive, since LTL does not have the full expressive power of deterministic parity automata over infinite words.

 Apart from matching the expressive power of existing linear-time logics for non-deterministic and probabilistic systems (see Example~\ref{ex:ltl}), the logics we consider support reasoning quantitatively about the linear behaviour of resource-aware systems (see Example~\ref{example-semirings}), and of branching systems whose executions are tree-shaped (see Example~\ref{ex:lin-beh}).

A more detailed outline of our approach and results is given below:
\begin{itemize}
\item We model systems as coalgebras of endofunctors obtained by composing a  \emph{branching monad} $\T_S : \Set \to \Set$ with a \emph{polynomial endofunctor $F : \Set \to \Set$}. The branching monad arises from a partial commutative semiring $S = (S,+,0,\bullet,1)$, whose carrier provides the domain of truth values for a quantitative logic for $\T_S \circ F$-coalgebras. Specifically, we take $\T_S X = \{ f : X \to S \mid f \text{ has finite support} \}$. Our logics contain both a (hidden) branching modality, used to quantify over the branching structure, and linear-time modalities, used to express properties of the linear behaviour. The semantics of the logics therefore requires \emph{quantitative predicate liftings} to interpret both the linear-time modalities and the hidden branching modality. The forgetful logics of \cite{KlinRot16} are similar in their use of a hidden branching modality, although they are much more permissive in the choice of such a modality -- they also cover logics which are expressive for bisimulation. The propositional operators employed by our logics are more restrictive than those found in linear-time temporal logics such as LTL \cite{Vardi95} or probabilistic LTL \cite{Vardi99}; nonetheless, our logics match the expressiveness of LTL, when instantiated to non-deterministic and probabilistic transition systems (see Example~\ref{ex:ltl}), with the only impact on expressiveness occurring when linear-time behaviours are tree-shaped (see Example~\ref{exa:express}).
\item In spite of their linear-time flavour, a key feature of our logics is their step-wise semantics: the interpretation of a formula is defined by successively unfolding the coalgebra structure, as required by the structure of the formula. This is different from logics such as (probabilistic) LTL or the (probabilistic) linear-time $\mu$-calculus, where the interpretation of a formula on a state is defined in terms of its interpretation on the (infinite) computation paths from that state. Moreover, unlike other coalgebraic logics (e.g.~the forgetful logics of \cite{KlinRot16}), the canonical choices made in the semantics of our logics also allow for a (measure-theoretic) path-based semantics to be defined, and proved equivalent to the step-wise semantics. The path-based semantics relies on a coalgebraic notion of \emph{path} through a coalgebra with branching, and involves defining a $\sigma$-algebra structure on the set of paths. Key to the equivalence result between the two semantics is the close relationship between the semiring structure and the associated branching modality.
\item In order to define our path-based semantics, a generalisation of standard measure extension results to semiring-valued measures is required. Specifically, we use a generalisation of the notion of \emph{outer measure} (see e.g.~\cite{Ash,Royden}) to extend a semiring-valued measure on a field of sets to a similar measure on the induced $\sigma$-algebra. Our focus is on the \emph{existence} of such extensions; uniqueness is not of interest to us, since the measure extensions obtained via the use of outer measures fit our goal of providing an equivalent path-based semantics for our logics.
\item Our quantitative logics naturally give rise to a logical distance between states of coalgebras with branching, and the existence of a path-based semantics helps to provide a semantic characterisation of this logical distance. Specifically, we introduce the semantic notion of \emph{linear-time distance} between states of two coalgebras, and show that it coincides with the logical distance.
\end{itemize}

We now comment on the relationship between our logics and existing quantitative, linear-time logics. We begin with the logic LTL, interpreted over either non-deterministic \cite{Vardi95} or probabilistic transition systems \cite{Vardi99}. In the case of non-deterministic systems, we consider the existential variant of LTL, whose semantics requires the existence of infinite computation paths satisfying certain linear-time properties. On the other hand, the probabilistic semantics of LTL uses measure theory to assign, to each linear-time property, the likelihood of it being satisfied on the infinite computation paths from a given state. A suitable choice of endofunctor $F : \Set \to \Set$ allows us to recover the expressiveness of both existential LTL and probabilistic LTL, albeit using a slightly different syntax (see Example~\ref{ex:ltl}). Our logics can therefore encode existing linear-time logics used in verification.

Other existing quantitative, LTL-like logics are either \emph{weighted ones}, with weights taken from a strong bimonoid and with propositional operators interpreted using the bimonoid sum, respectively multiplication (see e.g.~\cite{Droste2012}), or \emph{lattice-based ones} \cite{Kupferman2007}, with truth values taken from a lattice and with propositional operators interpreted using the lattice operations. The former approach immediately rules out partial semirings and does not appear to admit a path-based semantics, whereas the latter is much more restrictive than a semiring-based approach, given that the lattice idempotence laws do not generally hold in a semiring.

The remainder of the paper is structured as follows. Section~\ref{prelim} introduces partial semiring monads and semiring-valued relations and predicates, summarises previous work on a coalgebraic account of finite/maximal traces in systems with branching \cite{Cirstea17}, and recalls the basics of (qualitative) coalgebraic logics. Section~\ref{fixpoint-logic} describes linear-time fixpoint logics for coalgebras with branching, along with their step-wise semantics. Section~\ref{semiring-measure-gen} introduces \emph{semiring-valued measures} on $\sigma$-algebras, and studies the existence of measure extensions (from fields of sets to $\sigma$-algebras) in this setting. Building on this, Section~\ref{path-based-sem} provides an alternative, path-based semantics for the linear-time fixpoint logics of Section~\ref{fixpoint-logic}, based on notions of \emph{maximal path} and \emph{path fragment} in a coalgebra with branching, whereas Section~\ref{equiv-sem} proves the equivalence of the two semantics. Finally, Section~\ref{sem-char} shows that the logical distance induced by our logics coincides with a natural semantic distance, which itself has a linear-time flavour, while Section~\ref{conclusions} summarises the results presented and briefly outlines future work.

\subsection*{Related Work} The logics considered here were originally introduced in \cite{Cirstea14} and further studied in \cite{Cirstea15}, where a path-based semantics for their \emph{fixpoint-free} fragment was also described; the latter, however, did not employ any measure-theoretic machinery. The connection between the logics in \cite{Cirstea14} and parity automata was studied in \cite{CirsteaSH17}.

\subsection*{Acknowledgements} {Thanks are due to Clemens Kupke for prompting towards a study of logical \emph{distance} as opposed to logical \emph{equivalence}, the result of which is Section~\ref{sem-char}; to Clemens Kupke and Toby Wilkinson for related discussions on logical and semantic distances in this context; and to Alexandre Goy for interesting discussions regarding the results of this paper. Many thanks also to the anonymous reviewers, whose comments and suggestions led to considerable improvements of both the results presented and their exposition.}

\section{Preliminaries}
\label{prelim}

\subsection{Monads and Partial Semirings}
\label{monads-semirings}

In what follows, we use \emph{monads} on the category $\Set$ of sets and functions to capture branching in coalgebraic types.

\begin{defi}
A \emph{monad} on a category $\C$ is a tuple $(\T,\eta,\mu)$  where $\T : \C \to \C$ is a functor and $\eta : \Id \Rightarrow \T$ and $\mu : \T \circ \T \Rightarrow \T$ are natural transformations, called the \emph{unit} and \emph{multiplication} of $\T$, subject to the following laws:
\begin{align*}
\UseComputerModernTips\xymatrix@+0.5pc{
\T \ar@{=>}[r]^-{\T \eta} \ar@{=>}[dr]_-{\id_\T} & \T^2 \ar@{=>}[d]^-{\mu} & \T \ar@{=>}[l]_-{\eta_\T} \ar@{=>}[dl]^-{\id_\T} & & \T^3 \ar@{=>}[d]_-{\T \mu} \ar@{=>}[r]^-{\mu_\T} & \T^2 \ar@{=>}[d]^-{\mu} \\
& \T & & & \T^2 \ar@{=>}[r]_-{\mu} & \T
}
\end{align*}
\end{defi}

The specific nature of the monads we consider makes them \emph{strong} and \emph{commutative}. A \emph{strong monad} comes equipped with a \emph{strength map} $\st_{X,Y} : X \times \T Y \to \T(X \times Y)$, natural in $X$ and $Y$ and subject to coherence conditions w.r.t.~$\eta$ and $\mu$ (see e.g.~\cite[Chapter 5]{JacobsBook} for details). For such a monad, one can also define a \emph{swapped strength map} $\st'_{X,Y} : \T X \times Y \to \T(X \times Y)$ by:
\[\UseComputerModernTips\xymatrix@+0.5pc{\T X \times Y \ar[r]^-{\tw_{\T X,Y}} & Y \times \T X \ar[r]^-{\st_{Y,X}} & \T(Y \times X) \ar[r]^-{\T \tw_{Y,X}} & \T(X \times Y)}\]
where $\tw_{X,Y} : X \times Y \to Y \times X$ is the \emph{twist map} taking $(x,y) \in X \times Y$ to $(y,x)$. \emph{Commutative monads} are strong monads where the maps $\mu_{X,Y} \circ \T \st'_{X,Y} \circ \st_{\T X,Y} : \T X \times \T Y \to \T(X \times Y)$ and $\mu_{X,Y} \circ \T \st_{X,Y} \circ \st'_{X,\T Y} : \T X \times \T Y \to \T(X \times Y)$ coincide, yielding a \emph{double strength map} $\dst_{X,Y} : \T X \times \T Y \to \T(X \times Y)$ for each choice of sets $X,Y$.

The monads considered in this paper arise from partial commutative semirings.

\begin{defi}
A \emph{partial commutative monoid} $(S,+,0)$ is given by a set $S$ together with a partial operation $+ : S \times S \to S$ and an element $0 \in S$, such that:
\begin{itemize}
\item $s + 0$ is defined for all $s \in S$ and moreover, $s+0 = s$,
\item $(s+t) +u$ is defined if and only if $s + (t + u)$ is defined, and in that case $(s+t) +u = s+( t+u)$,
\item whenever $s + t$ is defined, so is $t + s$ and moreover, $s + t = t + s$.
\end{itemize}
A \emph{partial commutative semiring} is a tuple $S := (S,+,0,\bullet,1)$ with $(S,+,0)$ a partial commutative monoid and $(S,\bullet,1)$ a commutative monoid, with $\bullet$ distributing over existing sums; that is, for all $s,t,u \in S$, $s \bullet 0 = 0$, and whenever $t + u$ is defined, then so is $s \bullet t + s \bullet u$ and moreover, $s \bullet t + s \bullet u = s \bullet (t + u)$.
\end{defi}

The addition operation of any partial commutative semiring induces a pre-order relation $\sqsubseteq \,\subseteq S \times S$, given by
\begin{equation}
\label{preorder}
x \sqsubseteq y ~~~\text{if and only if}~~~ \text{there exists } z \in S \text{ such that }x + z = y\end{equation}
for $x,y \in S$. It follows immediately from the axioms of a partial commutative semiring that $\sqsubseteq$ has $0 \in S$ as \emph{bottom element}, that is, $0 \sqsubseteq s$ for all $s \in S$, and that $\sqsubseteq$ is preserved by $+$ and $\bullet$ in each argument  (see \cite{Cirstea17} for details).

\begin{asm}
\label{ass:cpo}
We assume that $(S,\sqsubseteq)$ is a complete lattice, with the unit $1$ of $\bullet$ as top element. Moreover, we assume that both $+$ and $\bullet$ preserve joins of increasing countable chains, respectively meets of decreasing countable chains, in each argument.
\end{asm}

\begin{rem}
\label{rem-semiring}
Below we list some basic properties of a partial commutative semiring $(S,+,0,\bullet,1)$ satisfying Assumption~\ref{ass:cpo}:
\begin{enumerate}
\item \label{c1} Each $s \in S$ has a (not necessarily unique) "1-complement" w.r.t.~$+$ (that is, an element $s' \in S$ such that $s + s' = 1$). This follows directly from the definition of $\sqsubseteq$, using the fact that $1$ is top for $\sqsubseteq$.
\item \label{c2} Whenever $\sum\limits_{i \in I} a_i$ is defined (with $I$ a finite set) and $J \subseteq I$, $\sum\limits_{i \in J} a_i$ is also defined. This follows from $(S,+,0)$ being a partial commutative monoid.
\item \label{c3} Whenever $\sum\limits_{i \in I} a_i$ is defined (with $I$ a finite set), then so is $\sum\limits_{i \in I} a_i \bullet b_i$. To see this, let $c_i \in S$ be such that $b_i + c_i = 1$ for $i \in I$ (by using (\ref{c1}) above). Then $\sum\limits_{i \in I} a_i = \sum\limits_{i \in I} a_i \bullet (b_i + c_i) = \sum\limits_{i \in I} (a_i \bullet b_i + a_i \bullet c_i) = \sum\limits_{i \in I} (a_i \bullet b_i)  + \sum\limits_{i \in I} (a_i \bullet c_i)$ and therefore by (\ref{c2}), $\sum\limits_{i \in I} a_i \bullet b_i$ is defined.
\item \label{c4} Let $a_0 \sqsubseteq a_1 \sqsubseteq \ldots$ and $b_0 \sqsubseteq b_1 \sqsubseteq \ldots$ be increasing chains in $S$ with joins $a$ and $b$, respectively. Then, $a+b$ is the join of $a_0 + b_0 \sqsubseteq a_1 + b_1 \sqsubseteq \ldots$. That $a+b$ is an upper bound for this increasing chain is immediate from the preservation of $\sqsubseteq$ by $+$ in each argument. That it is the least upper bound also follows easily: if $c$ is an upper bound for $a_0 + b_0 \sqsubseteq a_1 + b_1 \sqsubseteq \ldots$ then, for $i \in \omega$, $a_0 + b_i \sqsubseteq a_1+b_i \sqsubseteq \ldots \sqsubseteq c$ together with the preservation of joins of increasing countable chains by $+$ in the first argument gives $a + b_i = \sup_{j \in \omega}a_j + b_i = \sup_{j \in \omega} (a_j + b_i) \sqsubseteq c$; we then obtain $a + b = a + \sup\limits_{i \in \omega} b_i = \sup\limits_{i \in \omega} (a + b_i) \sqsubseteq c$.
\item \label{c5} Similarly, if $a_0 \sqsupseteq a_1 \sqsupseteq \ldots$ and $b_0 \sqsupseteq b_1 \sqsupseteq \ldots$ are decreasing chains in $S$ with meets $a$ and $b$, respectively, then, $a+b$ is the meet of $a_0 + b_0 \sqsupseteq a_1 + b_1 \sqsupseteq \ldots$. This is proved similarly to (\ref{c4}) above.
\end{enumerate}
\end{rem}

\begin{exa}
\label{example-semirings}
In what follows, we consider the \emph{boolean semiring} $(\{0,1\},\vee,0,\wedge,1)$, the (partial) \emph{probabilistic semiring} $([0,1],+,0,*,1)$, the \emph{tropical semiring} ${\mathbb N} = (\mathbb N^\infty,\min,\infty,+,0)$ (with $\mathbb N^\infty = \mathbb N \cup \{\infty\}$) and its bounded variants $S_B = ([0,B] \cup \{\infty\},\min,\infty,+_B,0)$ with $B \in \mathbb N$, where for $m,n \in [0,B]\cup \{\infty\}$ we have
\begin{eqnarray*}
m +_B n = \begin{cases} m + n, \text{ if } m + n \le B\\
\infty, \text{ otherwise} \end{cases} .
\end{eqnarray*} The associated pre-orders are $\le$ on $\{0,1\}$ and $[0,1]$, and $\ge$ on $\mathbb N^\infty$ and $[0,B] \cup \{\infty\}$. All these pre-orders satisfy Assumption~\ref{ass:cpo}.
\end{exa}

A partial commutative semiring $S$ satisfying Assumption~\ref{ass:cpo} induces a \emph{semiring monad} $(\T_S,\eta,\mu)$ with
\begin{align*}
\T_S X &= \textstyle\{\, \varphi : X \to S \mid \supp(\varphi) \text{ is finite}\,, \sum\limits_{x \in \supp(\phi)} \phi(x) \text{ is defined}\,\}\,,\\
(\T_S f)(\sum_{i \in I} c_i x_i) &= \sum_{i \in I} c_i f(x_i) \text{ for } f : X \to Y,\\
\eta_X(x)(y) &= \begin{cases}1 & \text{if } y = x \\
0 & \text{otherwise} \end{cases}, \\
\mu_X(\Phi)(x) &= \textstyle\sum\limits_{\varphi \in \supp(\Phi)}\Phi(\varphi) \bullet \varphi(x) \text{ for } \Phi \in \T_S^2 X,
\end{align*}
Here, $\supp(\varphi) = \{ x \in X \mid \varphi(x) \ne 0\}$ is the \emph{support} of $\varphi$, and we use the formal sum notation $\sum_{i \in I} c_i x_i$, with $I$ finite, to refer to the element of $\T_S X$ mapping 
 $x \in X$ to $\sum_{j \in J_x} c_j$ with $J_x = \{ i \in I \mid x_i = x\}$. To see that the sum $\textstyle\sum\limits_{\varphi \in \supp(\Phi)}\Phi(\varphi) \bullet \varphi(x)$ used in the definition of the monad multiplication is always defined, note that since both $\textstyle\sum\limits_{\varphi \in \supp(\Phi)}\Phi(\varphi)$ and $\sum\limits_{x \in \supp(\phi)} \phi(x)$ with $\phi \in \supp(\Phi)$ are defined, then by (\ref{c3}) of Remark~\ref{rem-semiring}, so is $\textstyle\sum\limits_{\varphi \in \supp(\Phi)}\Phi(\varphi) \bullet (\sum\limits_{x \in \supp(\phi)} \phi(x)) =\textstyle\sum\limits_{\varphi \in \supp(\Phi)} \sum\limits_{x \in \supp(\phi)} \Phi(\varphi) \bullet \phi(x)$. It then follows by (\ref{c2}) of Remark~\ref{rem-semiring} that for $x \in X$, $\textstyle\sum\limits_{\varphi \in \supp(\Phi)}\Phi(\varphi) \bullet \phi(x)$ is itself defined.

The monad $\T_S$ above is \emph{strong} (as in fact all monads on $\Set$ have this property) and \emph{commutative} (see e.g.~\cite{CoumansJ2011}), with \emph{strength} $\st_{X,Y} : X \times \T_S Y \to \T_S(X \times Y)$ and \emph{double strength} $\dst_{X,Y} : \T_S X \times \T_S Y \to \T_S(X \times Y)$ given by
\begin{eqnarray*}
\st_{X,Y}(x,\psi)(z,y) & = &\begin{cases} \psi(y) & \text{if } z = x \\
0 & \text{otherwise}
\end{cases}, \qquad 
\dst_{X,Y}(\varphi,\psi)(z,y) = \varphi(z) \bullet \psi(y)\,.
\end{eqnarray*}

We write $1 = \{*\}$ for a final object in the category $\Set$; this should not be confused with the unit of the semiring multiplication. We immediately note that $\T_S 1 = S$. In what follows, we will use $\T_S 1$ and $S$ interchangeably, in particular $\T_S 1$ will be used in contexts where its (free) $\T_S$-algebra structure (given by $\mu_1 : \T_S^2 1 \to \T_S 1$) is relevant.

The relationship between monads and partial semirings was thoroughly studied in \cite{CoumansJ2011,Cirstea17}. We use semiring monads to model branching, with the semirings in Example~\ref{example-semirings} accounting for finite non-deterministic, probabilistic and weighted branching. In the latter case, we think of the weights as costs associated to single computation steps, with the bounded variants of the tropical semiring imposing an upper limit on the maximum allowable costs.

\begin{rem}
\label{partially-additive-rem}
Our earlier work \cite{Cirstea14,Cirstea15} was parameterised by a so-called \emph{partially additive monad}. The connection with the partial semiring monads used here is as follows: any commutative, partially additive monad which is, in addition, finitary, is isomorphic to a partial semiring monad (see \cite[Remark~2.4]{Cirstea15}). While our earlier work also covers the unbounded powerset monad (as an additive monad), we are not aware of any other examples of non-finitary monads to which the results in \cite{Cirstea14,Cirstea15} apply.
\end{rem}

\subsection{Algebras and Coalgebras}

\begin{defi}
Given an endofunctor $F : \Set \to \Set$, an \emph{$F$-algebra} is a pair $(A,\alpha)$ with $A$ a set and $\alpha : F A \to A$ a function, while an \emph{$F$-algebra homomorphism} between $F$-algebras $(A,\alpha)$ and $(B,\beta)$ is given by a function $f : A \to B$ such that $f \circ \alpha = \beta \circ F f$. Also, an \emph{$F$-coalgebra} is a pair $(C,\gamma)$ with $C$ a set (of states) and $\gamma : C \to F C$ a function called \emph{transition map}, while an \emph{$F$-coalgebra homomorphism} between $F$-coalgebras $(C,\gamma)$ and $(D,\delta)$ is given by a function $g : C \to D$ such that $F g \circ \gamma = \delta \circ g$. 
\end{defi}

\begin{exa}
\label{ex:lin-beh}~
\begin{enumerate}
\item When $F = A \times \Id$, $F$-coalgebras are in one-to-one correspondence with deterministic, labelled transition systems. Each state in such a system has a unique, labelled transition to another state. 
\item When $F = \{*\} + A \times \Id$, $F$-coalgebras are in one-to-one correspondence with deterministic, labelled transition systems with termination. Each state in such a system has either a terminating transition or a unique, labelled transition to another state.
\item When $F = A \times \Id \times \Id$, each state in an $F$-coalgebra has a unique, labelled transition resulting in \emph{two} successor states.
\item \label{dyn} When $F =  \{*\} + A \times \Id + B \times \Id \times \Id$, $F$-coalgebras model (deterministic) systems whose execution in a given state can result in either termination, or a \emph{single} successor state, or \emph{two} successor states. The latter can be used to model systems whose structure can evolve dynamically: one can think of the second successor state as modelling the creation of a new process.
\end{enumerate}
\end{exa}

The coalgebras of interest to us are $\T_S \circ F$-coalgebras, where $\T_S : \Set \to \Set$ is a semiring-valued monad (see Section~\ref{monads-semirings}) and $F : \Set \to \Set$ is a \emph{polynomial} functor (that is, $F$ is constructed from identity and constant functors using \emph{finite} products and set-indexed coproducts). We assume, w.l.o.g., that $F = \coprod_{i \in I} \Id^{j_i}$, with $j_i \in \omega$ for $i \in I$. We note that any polynomial endofunctor on $\Set$ is naturally isomorphic to a coproduct of finite (including empty) products of identity functors.

We write $(Z,\zeta)$ for the \emph{final} $F$-coalgebra (that is, a final object in the category of $F$-coalgebras and $F$-coalgebra homomorphisms). Its existence, under the assumption that $F$ is a polynomial functor, follows from \cite[Theorem~10.1]{Rutten00}. We also write $(I,\iota)$ for the \emph{initial} $F$-algebra (that is, an initial object in the category of $F$-algebras and $F$-algebra homomorphisms); its existence, for polynomial functors $F$, was proved in \cite{Lambek68}. Both $\zeta : Z \to F Z$ and $\iota : F I \to I$ are isomorphisms. We will refer to the elements of $Z$ ($I$) as \emph{maximal} (resp.~\emph{finite}) \emph{traces}. Both finite and maximal traces are completed traces, with maximal traces including, in addition to the finite traces, also any infinite ones; this is illustrated in the next example.

\begin{exa}
When $F = A \times \Id$, there are \emph{no} finite traces (the initial $F$-algebra has an empty carrier), whereas maximal traces are given by infinite sequences of elements of $A$. Changing $F$ to $\{*\} + A \times \Id$ results in finite traces being given by finite sequences of elements of $A$ (these capture terminating behaviours ending in $\iota_1(*)$), whereas maximal traces are given by either finite or infinite sequences of elements of $A$ (the latter capture non-terminating behaviours).
\end{exa}

A notion of \emph{(maximal) path} from a state in a $\T_S \circ F$-coalgebra (similar to that in \cite{Cirstea11}), as well as a notion of \emph{path fragment} from a state, are defined below.

\begin{defi}[Path, path fragment]
\label{path-fragm-def}
Let $F = \coprod_{i \in I} \Id^{j_i}$ with $j_i \in \omega$ for $i \in I$, and let $C$ be a set (of states). A \emph{path} is an element of the final $C \times F$-coalgebra $(Z_C,\zeta_C)$, while a \emph{path fragment} is an element of the initial $C  \times (\{*\} + F)$-algebra $(I_C,\iota_C)$. The \emph{depth} of a path fragment $q \in I_C$ is defined inductively by:
\begin{itemize}
\item if $\pi_2(\iota^{-1}_C(q)) = \iota_1(*)$, then $\depth(q) = 0$;
\item if $\pi_2(\iota^{-1}_C(q)) = \iota_2(\iota_i(*))$ with $j_i = 0$, then $\depth(q) = 1$;
\item if $\pi_2(\iota^{-1}_C(q)) = \iota_2(\iota_i (q_1,\ldots,q_{j_i}))$ for some $i \in I$ with $j_i > 0$, then $\depth(q) = 1 + \max(\depth(q_1),\ldots,\depth(q_{j_i}))$.
\end{itemize}
\end{defi}
A path thus corresponds to a possibly infinite tree, with each node annotated by a pair $(c,i)$ with $c \in C$ and $i \in I$, and having $j_i$ immediate sub-trees. The leaf nodes in such a tree are of the form $(c,i)$ with $j_i = 0$.  As $\iota_C : C \times (\{*\} + F I_C) \to I_C$ is an isomorphism, a path fragment is given by a \emph{finite} tree, with non-leaf nodes similar to those in a path but with an additional type of leaf node, namely one annotated by elements of $C$ only.

\begin{defi}
\label{defi:pref}
A path fragment $q$ is said to be a \emph{prefix} of a path $p$ if
\begin{enumerate}
\item $\pi_1(\iota^{-1}_C(q)) = \pi_1(\zeta_C(p))$,
\item either $\pi_2(\iota^{-1}_C(q)) = \iota_1(*)$, or $\pi_2(\iota^{-1}_C(q)) = \iota_2(\iota_i (q_1,\ldots,q_{j_i}))$ for some $i \in I$ and moreover, $\pi_2(\zeta_C(p)) = \iota_i (p_1,\ldots,p_{j_i})$ and $q_k$ is a prefix of $p_k$ for $k \in \{1,\ldots,j_i\}$.
\end{enumerate}
We write $\pref(p)$ for the set $\{q \mid q \text{ is a prefix of } p \}$. 
\end{defi}
Thus, $q$ is a prefix of $p$ if the finite tree associated to $q$ is a subtree of the possibly infinite tree associated to $p$. One can define when a path fragment $q$ is a prefix of another path \emph{fragment} $p$ in a similar way, namely by replacing the path $p$ with a path fragment, and the use of $\zeta_C$ by $\iota_C^{-1}$ in Definition~\ref{defi:pref}. Two path fragments are then said to be \emph{compatible} if they are both prefixes of some other path fragment, and are said to be \emph{incompatible} otherwise.

\begin{defi}[Path, path fragment in $(C,\gamma)$]
\label{defi-path}
Let $(C,\gamma)$ be a $\T_S \circ F$-coalgebra. A \emph{path from $c \in C$ in $(C,\gamma)$} is a path $p \in Z_C$ such that
\begin{itemize}
\item $\pi_1(\zeta_C(p)) = c$,
\item if $\pi_2(\zeta_C(p)) = \iota_i(p_1,\ldots,p_{j_i})$ and $\pi_1 ( \zeta_C(p_k)) = c_k$ for $k \in \{1,\ldots,j_i\}$, then\linebreak$\gamma(c)(\iota_i(c_1,\ldots,c_{j_i})) \ne 0$ and moreover, $p_k$ is a path from $c_k$ in $(C,\gamma)$, for $k \in \{1,\ldots,j_i\}$.
\end{itemize}
A \emph{path fragment from $c \in C$ in $(C,\gamma)$} is an element $q$ of the initial $C  \times (\{*\} + F)$-algebra $(I_C,\iota_C)$ such that
\begin{itemize}
\setcounter{enumi}{2}
\item $\pi_1(\iota_C^{-1}(q)) = c$,
\item either $\pi_2(\iota_C^{-1}(q))\!= \iota_1(*)$, or $\pi_2(\iota_C^{-1}(q))\!= \iota_2(\iota_i(q_1,\ldots,q_{j_i}))$ for some $i \in I$ with $\pi_1 ( \iota_C^{-1}(q_k))\!= c_k$ for $k \in \{1,\ldots,j_i\}$ and $\gamma(c)(\iota_i(c_1,\ldots,c_{j_i})) \ne 0$ and $q_k$ a path fragment from $c_k$ in $(C,\gamma)$, for $k \in \{1,\ldots,j_i\}$.
\end{itemize}
The set of all paths from $c \in C$ in $(C,\gamma)$ is denoted $\Paths_c$.  We also write $(\Paths_C,\zeta_C)$ for the $C \times F$-subcoalgebra of $(Z_C,\zeta_C)$ whose elements are paths from some state of $(C,\gamma)$. (That this is a subcoalgebra follows immediately from the definition of a path from $c \in C$.) For simplicity of notation, the dependency on $\gamma$ is left implicit in both cases.
\end{defi}
The difference between Definitions~\ref{path-fragm-def} and \ref{defi-path} is that the latter requires \emph{non-zero} weights for all transitions belonging to paths/path fragments in $(C,\gamma)$. One consequence of this, which will be exploited later, is that there are only \emph{finitely many} path fragments of a given depth from any state of $(C,\gamma)$.

In what follows, we will write $(c,i)(c_1,\ldots,c_{j_i})$ for the path fragment with root $(c,i)$ and immediate leaves given by $(c_1,\iota_1(*)),\ldots,(c_{j_i},\iota_1(*))$, and $(c,i)(c_1,\ldots,c_{j_i})[A_1/c_1,\ldots,A_{j_i}/c_{j_i}]$ for the set of paths obtained by replacing each leaf node $(c_i,\iota_1(*))$ of $(c,i)(c_1,\ldots,c_{j_i})$ by one of the paths in $A_i$, with $i \in \{1,\ldots,j_i\}$. As a result, if any of the sets $A_i$ is empty, then so is $(c,i)(c_1,\ldots,c_{j_i})[A_1/c_1,\ldots,A_{j_i}/c_{j_i}]$.

\subsection{Semiring-Valued Relations and Relation Lifting}
\label{rel-lifting}

Throughout this section we fix a partial commutative semiring $(S,+,0,\bullet,1)$ satisfying Assumption~\ref{ass:cpo}.

We let $\Rel$ denote the category\footnote{To keep notation simple, the dependency on $S$ is left implicit.} with objects given by triples $(X,Y,R)$, where $R : X \times Y \to S$ is an \emph{$S$-valued relation}, and with arrows from $(X,Y,R)$ to $(X',Y',R')$ given by pairs of functions $(f,g)$ as below, such that $R \sqsubseteq R' \circ (f \times g)$:
\[\UseComputerModernTips\xymatrix{X \times Y \ar@{}[dr]|-{\sqsubseteq}\ar[r]^-{f \times g} \ar[d]_-{R} & X' \times Y' \ar[d]^-{R'}\\ S \ar@{=}[r] & S}\]
Here, the order $\sqsubseteq$ on $S$ has been extended pointwise to $S$-valued relations with the same carrier. We write $q : \Rel \to \Set \times \Set$ for the functor taking $(X,Y,R)$ to $(X,Y)$ and $(f,g)$ to itself. It follows easily that $q$ is a fibration, with reindexing functors $(f,g)^* : \Rel_{X',Y'} \to \Rel_{X,Y}$ taking $R' : X' \times Y' \to S$ to $R' \circ (f \times g) : X \times Y \to S$. We also write $\Rel_{X,Y}$ for the \emph{fibre over $(X,Y)$}, i.e.~the subcategory of $\Rel$ with objects given by $S$-valued relations over $X \times Y$ and arrows given by $(\id_X,\id_Y)$.

\begin{defi}[Relation lifting]
An \emph{$S$-valued relation lifting} for a functor $F : \Set \to \Set$ is a \emph{fibred} functor\footnote{That is, a functor which preserves reindexings.} $L : \Rel \to \Rel$ making the following diagram commute:
\[\UseComputerModernTips\xymatrix{
\Rel \ar[d]_-{q} \ar[r]^-{L} & \Rel \ar[d]^-{q} \\
\Set \times \Set \ar[r]_-{F \times F} & \Set \times \Set}\]
\end{defi}

\cite{Cirstea17} shows how to canonically lift polynomial endofunctors to the category of $S$-valued relations. The definition of the lifting makes use of the partial semiring structure on $S$.

\begin{defi}
\label{rel-lift-pol}
For a polynomial endofunctor $F : \Set \to \Set$, the \emph{relation lifting} $\Rel(F) : \Rel \to \Rel$ is defined by structural induction on $F$:
\begin{itemize}
\item If $F = \Id$, $\Rel(F)$ takes an $S$-valued relation to itself.
\item If $F = C$, $\Rel(F)$ takes an $S$-valued relation to the \emph{equality relation} $\Eq(C) : C \times C \to S$ given by
\[\Eq_C(c,c') ~=~ \begin{cases} 1, \text{ if }c = c' \\
0, \text{ otherwise} \end{cases} .\]
\item If $F = F_1 \times F_2$,  $\Rel(F)$ takes an $S$-valued relation $R : X \times Y \to S$ to:
\[\!\!{\small \UseComputerModernTips\xymatrix@-1.2pc{(F_1 X \times F_2 X)\!\times\!(F_1 Y \times F_2 Y) \ar[rrr]^-{\langle \pi_1 \times \pi_1,\pi_2 \times \pi_2 \rangle} & & & (F_1 X \times F_1 Y) \times (F_2 X \times F_2 Y) \ar[rrrrr]^-{\Rel(F_1)(R) \times \Rel(F_2)(R)} & & & & & S \times S \ar[r]^-{\bullet} & S} .}\]
The functoriality of this definition follows from the preservation of $\sqsubseteq$ by $\bullet$.
\item if $F = \coprod_{i \in I}F_i$, $\Rel(F)(R) : (\coprod_{i \in I}F_i X) \times (\coprod_{i \in I}F_i Y) \to S$ is defined by case analysis:
\begin{align*}
\Rel(F)(R)(\iota_i(u),\iota_j(v)) & ~=~ \begin{cases} \Rel(F_i)(R)(u,v), & \text{ if } i = j\\
0, & \text{ otherwise} \end{cases}
\end{align*}
for $i,j \in I$, $u \in F_i X$ and $v \in F_j Y$. 
\end{itemize}
\end{defi}
It follows immediately from the above definition that $q \circ \Rel(F) = (F \times F) \circ q$. Moreover, an easy inductive proof shows that:
\begin{enumerate}
\item $\Rel(F)$ is a fibred functor,
\item $\Rel(F)$ preserves joins of increasing countable chains and meets of decreasing countable chains in each fibre of $q$.
\end{enumerate}

\begin{rem}
When $S = (\{0,1\},\vee,0,\wedge,1)$, $S$-valued relations $R : X \times Y \to S$ coincide with standard ones $R \subseteq X \times Y$. In this case, the notion of relation lifting of Definition~\ref{rel-lift-pol} also coincides with the standard one, as described e.g.~in \cite{JacobsBook}[Chapter 3]. 
\end{rem}

A special relation lifting called \emph{extension lifting}  is defined in \cite{Cirstea17} for any commutative, partially additive monad $\T$. The extension lifting $\E_\T : \Rel \to \Rel$ lifts the endofunctor $\T \times \Id$ to $\Rel$
\[\UseComputerModernTips\xymatrix{
\Rel \ar[d]_-{q} \ar[r]^-{\E_\T} & \Rel \ar[d]^-{q} \\
\Set \times \Set \ar[r]_-{\T \times \Id} & \Set \times \Set}\]
in a canonical way. In the special case when $\T$ is the partial semiring monad $\T_S$, the extension lifting takes $R : X \times Y \to S$ to the relation $\E_{\T_S}(R) : \T_S (X) \times Y \to S$ given by
\begin{equation}
\label{extension-relation-lifting}
\E_{\T_S}(R)(\sum_{i \in I} c_i x_i,y) = \sum_{i \in I} c_i \bullet R(x_i,y)
\end{equation}
with $c_i \in S$ and $x_i \in X$ for $i \in I$, and $y \in Y$. (Note here that the definedness of $\sum_i c_i \bullet R(x_i,y)$ follows from the definedness of $\sum_{i \in I} c_i$ by (\ref{c3}) of Remark~\ref{rem-semiring}.) Fibredness of $\E_{\T_S}$ also follows directly from its definition.

\begin{rem}
An alternative definition of the extension relation lifting, applicable to any strong monad $\T$, maps $R : X \times Y \to \T 1$ to
\[\UseComputerModernTips\xymatrix@+0.5pc{
\T X \times Y \ar[r]^-{\st'_{X,Y}} & \T(X \times Y) \ar[r]^-{\T (R)} & \T^2 1 \ar[r]^-{\mu_1} & \T 1
}\]
This general definition is canonical in the sense that the relation lifting of $R$ is its unique extension to a left-linear map -- note that both $\T X$ and $\T 1$ are the carriers of free $\T$-algebras. It is straightforward to check that, for partial semiring monads, the two definitions coincide.
\end{rem}

\subsection{Semiring-Valued Predicates and Predicate Lifting}
\label{pred-lifting}

The standard approach to defining the semantics of modal and fixpoint logics involves interpreting formulas as predicates over the state space of the system of interest. In the coalgebraic approach to modal logic, individual modal operators are interpreted using so called \emph{predicate liftings} \cite{Pattinson03}. In order to follow the same approach for \emph{quantitative} logics, we work with predicates valued in the partial commutative semiring $(S,+,0,\bullet,1)$ used to model branching. A similar approach is taken in \cite{Schroder:2011:DLF}, where \emph{fuzzy predicate liftings}, valued in the unit interval, are used to provide a semantics to fuzzy description logics. The more general notion of predicate lifting considered here is also implicit in some of the earlier work on coalgebraic logic, e.g.~\cite{Klin2005,Schroeder2008}.

We let $\SPred$ denote the category with objects given by pairs $(X,P)$ with $P : X \to S$ an \emph{$S$-valued predicate}, and arrows from $(X,P)$ to $(X',P')$ given by functions $f : X \to X'$ such that $P \sqsubseteq P' \circ f$:
\[\UseComputerModernTips\xymatrix{X \ar@{}[dr]|-{\sqsubseteq}\ar[r]^-{f} \ar[d]_-{P} & X' \ar[d]^-{P'}\\ S \ar@{=}[r] & S}\]
As with $S$-valued relations, we obtain a fibration $p : \Pred \to \Set$, with $p$ taking $(X,P)$ to $X$ and $f$ to itself. The fibre over $X$ is denoted $\Pred_X$, and the reindexing functor $f^* : \Pred_{X'} \to \Pred_{X}$ takes $P' : X' \to S$ to $P' \circ f : X \to S$.

The next definition generalises \emph{monotone} predicate liftings, as used in the semantics of coalgebraic modal logics \cite{Pattinson03}, to a quantitative setting.

\begin{defi}[Predicate lifting]
\label{def-pred-lifting}
An \emph{($S$-valued) predicate lifting} of arity $n \in \omega$ for an endofunctor $F : \Set \to \Set$  is a fibred functor $L : \Pred^n \to \Pred$ making the following diagram commute:
\[\UseComputerModernTips\xymatrix{
\Pred^n \ar[d]_-{p} \ar[r]^-{L} & \Pred \ar[d]^-{p} \\
\Set \ar[r]_-{F} & \Set}\]
where the category $\Pred^n$ has objects given by tuples $(X,P_1,\ldots,P_n)$ with $P_i : X \to S$ for $i \in \{1,\ldots,n\}$, and arrows from $(X,P_1,\ldots,P_n)$ to $(X',P_1',\ldots,P_n')$ given by functions $f : X \to X'$ such that $P_i \sqsubseteq P'_i \circ f$ for all $i \in \{1,\ldots,n\}$.
\end{defi}

\begin{rem}
It is not difficult to see that predicate liftings as defined above coincide with \emph{monotone} predicate liftings in the sense of \cite{Pattinson03}, suitably generalised to an $S$-valued setting, that is, with monotone natural transformations $l : (S^{^{\_}})^n \Longrightarrow S^{^{\_}} \circ F$, where $S^{^{\_}} : \Set \to \Set$ is the contravariant functor mapping $X$ to the set of $S$-valued functions on $X$. Here, a natural transformation $l$ as above is \emph{monotone} if whenever $f_i \sqsubseteq g_i$ with $f_i, g_i : X \to S$ and $i \in \{1,\ldots,n\}$, then also $l(f_1,\ldots,f_n) \sqsubseteq l(g_1,\ldots,g_n)$. The naturality of standard predicate liftings is here captured by the fibredness requirement on $\L$, whereas monotonicity corresponds to the functoriality of $\L$.
\end{rem}

We now restrict attention to polynomial functors $F : \Set \to \Set$, and show how to define a canonical \emph{set} of predicate liftings for $F$ by induction on its structure. The next definition exploits the observation that any polynomial endofunctor is naturally isomorphic to a coproduct of finite (including empty) products of identity functors. 

\begin{defi}
\label{canonical-pred-lift}
Let $F = \coprod_{i \in I} \Id^{j_i}$, with $j_i \in \omega$ for $i \in I$. The set of predicate liftings $\Lambda = \{ L_i \mid i \in I\}$ has elements $L_i : \Pred^{j_i} \to \Pred$ with $i \in I$ given by:
\[(L_i)_X(P_1,\ldots,P_{j_i})(f) ~=~ \begin{cases} P_1(x_1) \bullet \ldots \bullet P_{j_i}(x_{j_i}),  & \text{ if } f = (x_1,\ldots,x_{j_i}) \in \iota_{i}(\Id^{j_i}) \\ 0 & \text{ otherwise} \end{cases}. \]
\end{defi}
The functoriality of this definition follows from the preservation of $\sqsubseteq$ by $\bullet$. The fact that each $L_i$ is a fibred functor follows directly from its definition. As a result of Assumption~\ref{ass:cpo}, all these predicate liftings preserve joins of increasing countable chains and meets of decreasing countable chains in each argument.

\begin{rem}
\label{lambda-cont-cocont}
The predicate liftings $L_i$ of Definition~\ref{canonical-pred-lift} preserve joins of increasing countable chains as well as meets of decreasing countable chains in each fibre of $p$, in each argument. This is immediate from the definition of $L_i$ and the preservation of such joins and meets by $\bullet$ in each argument.
\end{rem}

\begin{rem}
\label{rem:nabla}
When $S = (\{0,1\},\vee,0,\wedge,1)$, the predicate liftings of Definition~\ref{canonical-pred-lift} are essentially the same as the \emph{Nabla modality} of \cite{Moss1999}. 
\end{rem}
\begin{exa}
\label{trees}
For $F = \{*\} + A \times \Id \times \Id \,\simeq\, \{*\} + \coprod_{a \in A} \Id \times \Id$, states in $F$-coalgebras unfold to (potentially infinite) binary trees with internal nodes labelled by elements of $A$ and leaves not carrying any label. Definition~\ref{canonical-pred-lift} yields a nullary predicate lifting $L_0$ and an $A$-indexed set of binary predicate liftings $(L_a)_{a \in A}$:
\begin{align*}
(L_0)_X(f) & ~=~ \begin{cases} 1, & \text{ if } f = \iota_1(*) \\
0, & \text{ otherwise} \end{cases},\\
(L_a)_X(P_1,P_2)(f) & ~=~ \begin{cases} P_1(x_1) \bullet P_2(x_2), & \text{ if } f = \iota_a(x_1,x_2) \\ 0, & \text{ otherwise} \end{cases}.
\end{align*}
\end{exa}
As we are interested in \emph{linear-time} logics, a special \emph{extension lifting}, akin to the extension relation lifting of Section~\ref{rel-lifting}, will be used to abstract away branching.
\begin{defi}
\label{def-ext-pred-lifting}
Let $(S,+,0,\bullet,1)$ be a partial commutative semiring with associated monad $\T_S$. The \emph{extension predicate lifting} $\E_{\T_S} : \Pred \to \Pred$ is the lifting of $\T_S : \Set \to \Set$ to $\Pred$
\[\UseComputerModernTips\xymatrix{
\Pred \ar[d]_-{p} \ar[r]^-{\E_{\T_S}} & \Pred \ar[d]^-{p} \\
\Set \ar[r]_-{\T_S} & \Set}\]
which takes $P : X \to S$ to the predicate $\E_{\T_S}(P) : \T_S X \to S$ given by
\[\sum_{i \in I}c_i x_i ~\mapsto~ \sum_{i \in I} c_i \bullet P(x_i)\]
with $c_i \in S$ for $i \in I$ being such that $\sum_{i \in I} c_i$ is defined, and with $x_i \in X$ for $i \in I$.
\end{defi}
It then follows directly from the definition that $\E_{\T_S}$ is a fibred functor. It also follows from Assumption~\ref{ass:cpo} together with (\ref{c4}) and (\ref{c5}) of Remark~\ref{rem-semiring} that $\E_{\T_S}$ preserves the joins of increasing chains $P_0 \sqsubseteq P_1 \sqsubseteq \ldots$ and the meets of decreasing chains $P_0 \sqsupseteq P_1 \sqsupseteq \ldots$ in each fibre of $p$.

\begin{exa}
For $S = (\{0,1\},\vee,0,\wedge,1)$, the predicate lifting $\E_{\T_S}$ takes a standard predicate $P \subseteq X$ to the predicate $\{ Y \subseteq X \mid Y \text{ finite}\,,\,  Y \cap P \ne \emptyset\} \subseteq \FPow X$, where $\FPow : \Set \to \Set$ is the finite powerset functor (naturally isomorphic to $\T_S$). This corresponds to the standard $\Diamond$ modality. For $S = ([0,1], +, 0, *, 1)$, the predicate lifting $\E_{\T_S}$ takes a predicate $P : X \to [0,1]$ to the predicate $\E_{\T_S}(P) : \T_S X \to [0,1]$ given by $\E_{\T_S}(P)(\sum_{i \in I} c_i x_i) = \sum\limits_{i \in I} (c_i * P(x_i))$. For $S = (\mathbb N^\infty, \min, \infty, +, 0)$, the predicate lifting $\E_{\T_S}$ takes a predicate $P : X \to \mathbb N^\infty$ to the predicate $\E_{\T_S}(P) : \T_S X \to \mathbb N^\infty$ given by $\E_{\T_S}(P)(\sum_{i \in I} c_i x_i) = \min\limits_{i \in I} (c_i + P(x_i))$.
\end{exa}

\begin{rem}
\label{rem-ext}
$\E_{\T_S}(P)$ can alternatively be defined as $\mu_1 \circ \T_S P$, for $P : X \to S$.
\end{rem}

\subsection{Finite and Maximal Traces via Relation Lifting}

We now summarise the definitions of \emph{finite trace behaviour} and \emph{maximal trace behaviour} of a state in a coalgebra with branching, as defined in \cite{Cirstea17}. The approach in loc.\,cit.~applies to coalgebras of functors obtained as compositions of a single partially additive monad and a finite number of polynomial endofunctors on $\Set$. For simplicity, here we restrict attention to compositions of type $\T_S \circ F$, with $S$ a partial commutative semiring and $F : \Set \to \Set$ a polynomial functor.

The notion of \emph{coalgebraic bisimulation} provides a canonical and uniform observational equivalence relation between states of $F$-coalgebras. One of the many (and largely equivalent) definitions of bisimulation involves lifting the functor $F$ to the category of two-valued relations (obtained in our setting by taking $S = (\{0,1\},\vee,0,\wedge,1)$). Then, \emph{$F$-bisimilarity}, defined as the largest bisimulation between the states of two $F$-coalgebras, can be characterised as the greatest fixpoint of a monotone operator on the complete lattice of relations between the underlying state spaces \cite{JacobsBook}. A similar approach is taken in \cite{Cirstea17} to define the \emph{extent} to which a state in a coalgebra with branching can exhibit a given maximal/finite trace. The definition in loc.\,cit.~differs from the above characterisation of bisimilarity in two ways: (i) $S$-valued relations are used in place of two-valued relations, and (ii) the relation lifting employed also involves the extension relation lifting $\E_{\T_S}$ defined earlier, as the goal is to relate \emph{branching-time} behaviours with \emph{linear-time} ones, as opposed to relating behaviours of the same coalgebraic type.

\begin{defiC}[\cite{Cirstea17}]
\label{def-max-trace-beh}
Let $(Z,\zeta)$ denote the final $F$-coalgebra. The \emph{maximal trace behaviour} of a state in a $\T_S \circ F$-coalgebra $(C,\gamma)$ is the greatest fixpoint $\tr_\gamma : C \times Z \to S$ of the monotone operator $\Op : \Rel_{C,Z} \to \Rel_{C,Z}$ given by the composition
\begin{equation}
\UseComputerModernTips\xymatrix@+0.5pc{
\Rel_{C,Z} \ar[r]^-{\Rel(F)} & \Rel_{F C,F Z} \ar[r]^-{\E_{\T_S}} & \Rel_{\T_S F C, F Z}  \ar[r]^-{(\gamma \times \zeta)^*} & \Rel_{C,Z}}
\end{equation}
Two states $c,d \in C$ are said to be \emph{maximal trace equivalent}, written $c \simeq_\tr d$, iff $\tr_\gamma(c,z) = \tr_\gamma(d,z)$ for all $z \in Z$.
\end{defiC}
We spell out the definition of the operator $\Op$ in Definition~\ref{def-max-trace-beh}, based on the definitions of $\Rel(F)$ (Definition~\ref{rel-lift-pol}) and $\E_{\T_S}$ (see (\ref{extension-relation-lifting})). For an $S$-valued relation $R : C \times Z \to S$ and $(c,z) \in C \times Z$ with $\gamma(c) = \sum\limits_{i \in I} s_i \iota_{\lambda_i}(c_1^i,\ldots,c_{j_i}^i)$ and $\zeta(z) = \iota_{\lambda}(z_1,\ldots,z_{\arity(\lambda)})$, we have:
\begin{align*}
\Op(R)(c,z) = \sum\limits_{i \in I \,\text{s.t.}\, \lambda_i = \lambda}s_i \bullet R(c^i_1,z_1) \bullet \ldots \bullet R(c^i_{\arity(\lambda)},z_{\arity(\lambda)}) \,.
\end{align*}
Thus, if the trace $z$ prescribes $\lambda$ as the next transition and then continues with traces $z_1,\ldots,z_{\arity(\lambda)}$, then only $\lambda$-transitions from $c$ are taken into account when defining $\Op(R)(c,z)$. This value is obtained by summing, across all $\lambda$-transitions from $c$ ($\lambda_i = \lambda$), the values $R(c^i_1,z_1) \bullet \ldots \bullet R(c^i_{\arity(\lambda)},z_{\arity(\lambda)})$ weighted with the corresponding transition weights ($s_i$).

The existence of a greatest fixpoint for $\Op$ follows from the Knaster-Tarski theorem \cite{Tarski55}. This result will be used repeatedly in what follows.

\begin{thmC}[{\cite{Tarski55}}]
\label{thm:tarski}
Let $\Op : L \to L$ be a monotone operator on a complete lattice $(L,\sqsubseteq)$. Then, the set of fixpoints of $\Op$ forms a complete lattice, also under $\sqsubseteq$.
\end{thmC}
Under the assumptions of Theorem~\ref{thm:tarski}, the following result provides a way to obtain the least and greatest fixpoint of $\Op$.

\begin{thmC}[{\cite{Cousot79}}]
\label{thm:cousot}
Let $\Op : L \to L$ be a monotone operator on a complete lattice $(L,\sqsubseteq)$. Consider the ascending chain $(\Op^\alpha(\bot))_\alpha$, with $\alpha$ ranging over the ordinals, defined by $\Op^0(\bot) = \bot$, $\Op^{\alpha+1}(\bot) = \Op(\Op^\alpha(\bot))$ for any ordinal $\alpha$, and $\Op^{\alpha}(\bot) = \bigsqcup_{\beta < \alpha} \Op^\beta(\bot)$ for $\alpha$ a limit ordinal. Then, the least fixpoint of $\Op$ is given by $\Op^\gamma(\bot)$ for some ordinal $\gamma$. The greatest fixpoint of $\Op$ is obtained dually, via a descending chain $(\Op^\alpha(\top))_\alpha$.
\end{thmC}

The fact that the operator $\Op$ of Definition~\ref{def-max-trace-beh} is monotone follows from the functoriality of each of $\Rel(F)$, $\E_{\T_S}$ and $(\gamma \times \zeta)^*$.

The definition of \emph{finite} trace behaviour simply replaces maximal traces by finite ones.

\begin{defiC}[\cite{Cirstea17}]
\label{def-finite-trace-beh}
Let $(I,\iota)$ denote the initial $F$-algebra. The \emph{finite trace behaviour} of a state in a $\T_S \circ F$-coalgebra $(C,\gamma)$ is the greatest fixpoint of the operator $\Op' : \Rel_{C,I} \to \Rel_{C,I}$ given by the composition
\begin{equation}
\UseComputerModernTips\xymatrix@+0.5pc{
\Rel_{C,I} \ar[r]^-{\Rel(F)} & \Rel_{F C,F I} \ar[r]^-{\E_{\T_S}} & \Rel_{\T_S F C, F I}  \ar[r]^-{(\gamma \times \iota^{-1})^*} & \Rel_{C,I}}
\end{equation}
\end{defiC}

The operator $\Op'$ is thus similar to the operator $\Op$, except that the final $F$-coalgebra is replaced by the initial $F$-algebra. We note that taking the \emph{least} fixpoint of $\Op'$ would yield an equivalent definition, since in this case the fixpoint is unique. Informally, this is because the elements of $I$ are \emph{finite} traces, and because the $n$th approximations $(\Op')^n(\bot)$ and $(\Op')^n(\top)$ of the least, respectively greatest fixpoint of $\Op'$ will coincide on pairs $(c,i) \in C \times I$ with $i$ a finite trace of depth at most $n$, for all $n \in \omega$ (and thus so will the least and greatest fixpoints of $\Op'$).

\begin{exa}~
\begin{enumerate}
\item For $S = (\{0,1\}, \vee, 0, \wedge, 1)$, the greatest fixpoint of $\Op$/$\Op'$ relates a state $c$ in a $\T_S \circ F$-coalgebra $(C,\gamma)$ with a maximal/finite trace $t$ iff there exists a sequence of choices in the unfolding of $\gamma$ starting from $c$ that results in an $F$-behaviour bisimilar to $t$.
\item For $S = ([0,1], +, 0, *, 1)$, the greatest fixpoint of $\Op$/$\Op'$ yields, for each state in a $\T_S \circ F$-coalgebra and each maximal/finite trace, the accumulated probability (across all branches) of this trace being exhibited. Here we note that, for \emph{infinite} maximal traces, the associated probability is often $0$. Arguably, this has limited usefulness, and a measure-theoretic definition that takes into account the accumulated probabilities of exhibiting \emph{finite prefixes} of infinite traces would in this case be more useful. The logics defined later in this paper do not suffer from this issue, as they allow expressing the probability of exhibiting certain \emph{sets} of traces, including traces with a given finite prefix.
\item For $S=(\mathbb N^\infty,\min,\infty,+,0)$, the greatest fixpoint of $\Op$/$\Op'$ maps a pair $(c,t)$, with $c$ a state in a $\T_S \circ F$-coalgebra and $t$ a maximal/finite trace, to the minimal cost of exhibiting $t$ from $c$. Intuitively, this is computed by adding the weights of individual transitions along the same branch, and minimising this sum across all the branches. Again, such minimal costs are often infinite for \emph{infinite} maximal traces, but in this case this simply reflects the fact that infinitely running systems have infinite costs. While resource gain could be modelled by considering coalgebras of a different type (namely $S \times (\T_S \circ F)$, with $F$ as before and with the first compoment being used to associate resource gains to coalgebra \emph{states}), we leave a detailed study of this more general case to future work.
\end{enumerate}
\end{exa}

\subsection{\texorpdfstring{$\mu$}{μ}- and \texorpdfstring{$\nu$}{ν}-Extents via Predicate Lifting}

We now define the notions of \emph{$\nu$-extent} and \emph{$\mu$-extent} of a coalgebra with branching, which generalise the non-emptiness of the set of maximal, respectively finite traces in a system with non-deterministic branching, to systems with quantitative branching. These notions were introduced in \cite{CirsteaSH17} in order to provide automata-based model checking techniques for the logics in \cite{Cirstea14,Cirstea15}. Here, the notion of $\nu$-extent will be key to defining a notion of linear-time behaviour of a state in a system with branching (Section~\ref{sem-char}), as well as a path-based semantics for quantitative, linear-time fixpoint logics interpreted over such systems (Section~\ref{path-based-sem}).

\begin{defiC}[\cite{CirsteaSH17}]
\label{extent-coalgebra}
The \emph{$\nu$-extent} (\emph{$\mu$-extent}) of a $\T_S \circ F$-coalgebra $(C,\gamma)$ is the greatest fixpoint (resp.~least fixpoint) of the operator on $S^C$ taking $p : C \to S$ to the composition
\begin{align*}
\UseComputerModernTips\xymatrix@-0.5pc{C \ar[r]^-{\gamma} & \T_S F C \ar[rr]^-{\T_S F p} & & \T_S F S \ar[rr]^-{\T_S(\bullet_F)} & & \T_S S = \T_S^2 1 \ar[r]^-{\mu_1} & \T_S 1 = S}
\end{align*}
where $\bullet_F : F S \to S$ is given by $\bullet_F(\iota_i(s_1,\ldots,s_{j_i})) = s_1 \bullet \ldots \bullet s_{j_i}$ for $i \in I$. We write $\extent^\nu_\gamma : C \to S$ for the $\nu$-extent of $(C,\gamma)$.
\end{defiC}
The above operator uses a one-step unfolding of the coalgebra structure to compute (a finer approximation of) the extent of a state based on (current approximations of) the extents of its immediate successors. As the generality of $F$ allows for immediate successors which are \emph{tuples} of states, the semiring multiplication also needs to be used (namely in $\bullet_F$). On the other hand, the monad multiplication is used to accumulate the values from different branches. The composition in Definition~\ref{extent-coalgebra} takes $c \in C$ with $\gamma(c) = \sum_i s_i (c_i^1,\ldots,c_i^{j_i})$ to $\mu_1(\sum_i s_i (p(c_i^1) \bullet \ldots \bullet p(c_i^{j_i})))$. As a result, the extent in state $c$ is a sum, across all transitions from $c$, of the extents of the successors of $c$, scaled by the weights of the corresponding transitions. In case of transitions with \emph{several} successor states, the extents of these states are first multiplied and then scaled by the transition weight.

\begin{exa}~
\label{exa:extent}
\begin{enumerate}
\item For $S = (\{0,1\},\vee,0,\wedge,1)$, the $\nu$- (resp.~$\mu$-)extent has a value of $1$ on a state iff there exists a maximal (resp.~finite) trace from that state, arising from a sequence of choices in the branching behaviour. (Such a trace will not exist e.g.~from a state which offers no choices for proceeding.)  
\item For $S = ([0,1],+,0,*,1)$, the $\nu$- (resp.~$\mu$-)extent on a state gives the probability of not deadlocking (resp.~not deadlocking and not executing forever) from that state. 
\item For $S = (\mathbb N^\infty,\min,\infty, +,0)$, the $\nu$- (resp.~$\mu$-)extent on a given state gives the minimal cost that can be achieved along a maximal (resp.~finite) trace from that state.
\end{enumerate}
\end{exa}

\begin{exaC}[\cite{CirsteaSH17}]
\label{mu-nu-extent-example}
Consider the $\T_S \circ F$-coalgebras below, with $F = \{*\} + A \times \Id$ and $S = ([0,1],+,0,*,1)$ (resp.~$S=(\mathbb N^\infty,\min,\infty,+,0)$), and note that in both cases, all finite traces must end with a transition from $y$ to $\iota_1(*)$. The $\nu$-extent maps $x$ to $\frac{2}{5}$, $y$ to $\frac{3}{5}$ and $z$ to $\frac{1}{5}$ (resp.~$x$ and $y$ to $1$ and $z$ to $0$), whereas the $\mu$-extent again maps $x$ to $\frac{2}{5}$, $y$ to $\frac{3}{5}$ and $z$ to $\frac{1}{5}$ (resp.~$x$ and $z$ to $4$ and $y$ to $2$). Intuitively, the reason for the $\mu$- and $\nu$-extents being the same in the probabilistic case is that the probability of never reaching $y$ from either $x$ or $z$ is $0$.
\begin{align*}
\UseComputerModernTips\xymatrix@-1pc{
& & & & x \ar@/_0.5pc/[dll]_-{\frac{1}{2},a} \ar@/^0.5pc/[drr]^-{\frac{1}{2},b} \\
& & y \ar[ll]_-{\frac{1}{2},*} \ar@/_1ex/[urr]_-{\frac{1}{4},c} & & & & z \ar@(u,r)[]^-{\frac{1}{2},c} \ar@/^1ex/[ull]^-{\frac{1}{4},c} &
} \qquad \UseComputerModernTips\xymatrix@-1pc{
& & & & x \ar@/_0.5pc/[dll]_-{2,a} \ar@/^0.5pc/[drr]^-{1,b} \\
& & y \ar[ll]_-{2,*} \ar@/_1ex/[urr]_-{0,c} & & & & z \ar@(u,r)[]^-{0,c} \ar@/^1ex/[ull]^-{0,c} &
}
\end{align*}
To see why the $\nu$-extents are as claimed in the probabilistic case, note that they are the greatest (and in this case unique) solution of the following system of equations:
\begin{align*}
\qquad 
\begin{bmatrix}
x & = & \frac{1}{2} * y + \frac{1}{2} * z\\
y & = & \frac{1}{2} + \frac{1}{4} * x\\
z & = & \frac{1}{4} * x + \frac{1}{2} * z
\end{bmatrix}
\end{align*}
\end{exaC}
It turns out that the notions of $\nu$- and $\mu$-extent can be used to give alternative characterisations of the notions of maximal trace behaviour, and respectively of finite trace behaviour. Specifically, these can be shown to coincide with the $\nu$-, respectively $\mu$-extents of certain "products" of the coalgebra in question with the final, respectively the initial $F$-algebra, where the latter two are now viewed as $\T_S \circ F$-coalgebras.

\begin{defiC}[\cite{CirsteaSH17}]
\label{simple-prod-aut}
The \emph{product} of $\T_S \circ F$-coalgebras $(C,\gamma)$ and $(D,\delta)$ is the $\T_S \circ F$-coalgebra with carrier $C \times D$ and transition function $\gamma \otimes \delta$ given by
\begin{align*}
\UseComputerModernTips\xymatrix{C \times D \ar[r]^-{\gamma \times \delta} & \T_S F C \times \T_S F D \ar[rr]^-{\dst_{F C,F D}} & & \T_S (F C \times F D) \ar[rr]^-{\langle F \pi_1,F\pi_2 \rangle^*} & & \T_S F (C \times D) }
\end{align*}
where $\dst$ is the double strength of the monad $\T_S$ (see Section~\ref{monads-semirings}) and $\langle F \pi_1,F\pi_2 \rangle^*$ is pre-composition with $\langle F \pi_1,F\pi_2 \rangle : F(C \times D) \to F C \times F D$.
\end{defiC}
The effect of pre-composing with $\langle F \pi_1,F\pi_2 \rangle$ is that pairs of non-matching one-step behaviours are discarded from the resulting coalgebra. Thus, the product coalgebra collects the common $F$-behaviour of the two coalgebras, suitably weighted according to the weights of the two coalgebras.

The following result, while new, resembles \cite[Theorem~14]{CirsteaSH17} in both its statement and its proof.

\begin{thm}
Let $(C,\gamma)$ be a $\T_S \circ F$-coalgebra, let $(Z,\zeta)$ denote the final $F$-coalgebra, and let $(I,\iota)$ denote the initial $F$-coalgebra. The following hold:
\begin{enumerate}
\item Both the $\mu$-extent and the $\nu$-extent of the product coalgebra $(C \times I,\gamma \otimes \iota)$ coincide with the finite trace behaviour $\ftr_\gamma : C \times I \to S$.
\item The $\nu$-extent of the product coalgebra $(C \times Z,\gamma \otimes \zeta)$ coincides with the maximal trace behaviour $\tr_\gamma : C \times Z \to S$.
\end{enumerate}
\end{thm}
\begin{proof}[{\bfseries Proof (sketch).}]
The proof is similar to that of \cite[Theorem~14]{CirsteaSH17} and involves showing that the operators used in the definition of the maximal/finite trace behaviour on the one hand, and of the $\mu$-/$\nu$-extent of the product automaton on the other, coincide.
\end{proof}

\subsection{Coalgebraic Fixpoint Logics}
\label{coalg-fixpoint-logics}

We now recall how a \emph{two-valued} fixpoint logic for $F$-coalgebras can be defined from a set of monotone predicate liftings for $F$ \cite{CKP2011}. A fragment of this logic, which does not contain conjunctions and arbitrary disjunctions, will later (Section~\ref{fixpoint-logic}) also be given a quantitative interpretation over $\T_S \circ F$-coalgebras, with $F : \Set \to \Set$ a polynomial functor and $\T_S : \Set \to \Set$ as before. Given a functor $F : \Set \to \Set$, a ($\{0,1\}$-valued) fixpoint logic for $F$-coalgebras is parameterised by a set of $\{0,1\}$-valued predicate liftings for $F$ (see Definition~\ref{def-pred-lifting}). Concretely, given a set $\Lambda$ of modal operators with associated $\{0,1\}$-valued, monotone predicate liftings $(\lsem \lambda \rsem)_{\lambda \in \Lambda}$ for $F$, the associated fixpoint logic has syntax:
\begin{eqnarray*}\textstyle
\!\!\!\!\!\!\!\mu\LL_\Lambda^\V \ni \varphi ::= \bot \,\mid\, \top \,\mid\, x \,\mid\, \langle \lambda  \rangle(\varphi_1,\ldots,\varphi_{\arity(\lambda)}) \,\mid\, \phi \vee \psi \,\mid\, \phi \wedge \psi \,\mid\, \mu x .\varphi \,\mid\, \nu x.\varphi
\end{eqnarray*}
with $x \in \V$ for some set $\V$ of variables. Here, $\arity(\lambda)$ is the arity of the modal operator $\langle \lambda \rangle$. Then, for an $F$-coalgebra $(C,\gamma)$ and a \emph{valuation} $V : \V \to \{0,1\}^C$, the semantics $\lsem \phi \rsem_\gamma^V \in \{0,1\}^C$ of formulas $\phi$ as above is given by:
\begin{itemize}
\item $\lsem \bot \rsem_\gamma^V (c) = 0$ for $c \in C$,
\item $\lsem \top \rsem_\gamma^V (c) = 1$ for $c \in C$,
\item $\lsem x \rsem_\gamma^V = V(x)$,
\item $\lsem \phi \vee \psi \rsem_\gamma^V(c) = \lsem \phi \rsem_\gamma^V(c) \vee \lsem \psi \rsem_\gamma^V(c)$ for $c \in C$,
\item $\lsem \phi \wedge \psi \rsem_\gamma^V(c) = \lsem \phi \rsem_\gamma^V(c) \wedge \lsem \psi \rsem_\gamma^V(c)$ for $c \in C$,
\item $\lsem \langle \lambda \rangle(\varphi_1,\ldots,\varphi_{\arity(\lambda)}) \rsem_\gamma^V = \gamma^* (\lsem \lambda \rsem_C (\lsem \varphi_1 \rsem_\gamma^V,\ldots,\lsem \varphi_{\arity(\lambda)} \rsem_\gamma^V))$, where $\gamma^* : \{0,1\}^{F C} \to \{0,1\}^{C}$ performs reindexing along $\gamma : C \to F C$.
\item $\lsem \mu x.\varphi \rsem_\gamma^{V \setminus \{x\}}$ ($\lsem \nu x.\varphi \rsem_\gamma^{V \setminus \{x\}}$) is the least (resp.~greatest) fixpoint of the operator $\Op_\phi$ on $\{0,1\}^C$ taking $p : C \to \{0,1\}$ to $\lsem \phi \rsem_\gamma^{V[p/x]}$, where the valuation $V[p/x] : \V \to \{0,1\}^C$ takes $x$ to $p$ and $y \in \V \setminus \{x\}$ to $V(y)$.
\end{itemize}
The existence of the fixpoints required to interpret fixpoint formulas $\mu x.\phi$ and $\nu x.\phi$ is guaranteed by Theorem~\ref{thm:tarski}, given the complete lattice structure of $\{0,1\}^C$ (inherited from $\{0,1\}$) and the monotonicity of the operators involved. The latter is an immediate consequence of the monotonicity of the predicate liftings.

\begin{rem}
\label{rem:poly}
One can take $F : \Set \to \Set$ to be a polynomial functor and the predicate liftings $\lsem\lambda\rsem$ to be as in Definition~\ref{canonical-pred-lift} (with $S = (\{0,1\},\vee,0,\wedge,1)$). As these predicate liftings preserve joins of increasing countable chains and meets of decreasing countable chains in $\Pred_C$ (see Remark~\ref{lambda-cont-cocont}), the operator $\Op_\phi$ used in defining the semantics of fixpoint formulas is both continuous and co-continuous. It then follows by Kleene's fixpoint theorem (included below for reference) that the least (respectively greatest) fixpoint of $\Op_\phi$ can be constructed as the join of the increasing chain $0 \sqsubseteq \Op_\phi(0) \sqsubseteq \ldots$ (respectively the meet of the decreasing chain $1 \sqsupseteq \Op_\phi(1) \sqsupseteq \ldots$) in $\Pred_C$.
\end{rem}

\begin{thmC}[{\cite{Kleene52}}]
Let $\Op : L \to L$ be a continuous operator on a complete lattice $(L,\sqsubseteq)$. Then, the least fixpoint of $\Op$ is given by $\bigsqcup_{\alpha < \omega} \Op^\alpha(\bot)$.
\end{thmC}

\section{Quantitative Linear-Time Fixpoint Logics}
\label{fixpoint-logic}

We are now ready to define \emph{linear-time} fixpoint logics for coalgebras of type $\T_S \circ F$, where $\T_S : \Set \to \Set$  is the monad induced by a partial semiring $(S, +,0,\bullet,1)$ whose associated order $\sqsubseteq$ satisfies Assumption~\ref{ass:cpo}, and $F : \Set \to \Set$ is a polynomial functor. Similarly to Definition~\ref{canonical-pred-lift}, we assume that $F$ is given by $\coprod_{\lambda \in \Lambda} \Id^{\arity(\lambda)}$, with $\Lambda$ a set and each $\lambda \in \Lambda$ being assigned a finite arity $\arity(\lambda)$. Our logics will be valued into the semiring carrier $S$, and will use modalities from $\Lambda$ with associated \emph{$S$-valued} predicate liftings $(\lsem \lambda \rsem)_{\lambda \in \Lambda}$ for $F$ as in Definition~\ref{canonical-pred-lift}. A step-wise semantics for these logics is given herewith, while an alternative, path-based semantics is described in Section~\ref{path-based-sem} and proved equivalent to the step-wise semantics in Section~\ref{equiv-sem}.

\begin{defi}[Linear-time fixpoint logic \cite{Cirstea14}]
\label{syntax}
Let $\V$ be a set (of variables). The logic $\mu\LL_\Lambda^\V$ has syntax given by
\begin{eqnarray*}\textstyle
\!\!\!\!\!\!\!\mu\LL_\Lambda^\V \ni \varphi ::=~ \bot \,\mid\,\top \,\mid\, x \,\mid\, \langle \lambda \rangle (\varphi_1,\ldots,\varphi_{\arity(\lambda)}) \,\mid\, \mu x .\varphi \,\mid\, \nu x.\varphi
\end{eqnarray*}
with $x \in \V$ and $\lambda \in \Lambda$. We write $\mu\LL_\Lambda$ for the set of \emph{closed} formulas (containing no free variables).
\end{defi}

While here we are concerned with a \emph{quantitative} interpretation of $\mu\LL_\Lambda$ over $\T_S \circ F$-coalgebras, we immediately note that $\mu\LL_\Lambda$ also has a \emph{qualitative} interpretation over $F$-coalgebras, described in Section~\ref{coalg-fixpoint-logics}. For this reason, we can view $\mu\LL_\Lambda$-formulas as describing \emph{linear-time} properties, that is, \emph{qualitative} properties of $F$-structures.

A step-wise semantics for $\mu\LL_\Lambda$ is defined below, by induction on the structure of formulas.

\begin{defi}[Step-wise semantics for $\mu\LL_\Lambda$]
\label{step-wise-semantics}
For a $\T_S \circ F$-coalgebra $(C,\gamma)$ and a valuation $V : \V \to S^C$ (interpreting the variables in $\V$ as $S$-valued predicates over $C$), the denotation $\lsem \phi \rsem_\gamma^V \in S^C$ of a formula $\phi \in \mu\LL_\Lambda^\V$ is defined inductively on the structure of $\phi$ by
\begin{itemize}
\item $\lsem \bot \rsem_\gamma^V$ is given by $\lsem \bot \rsem_\gamma^V(c) = 0$ for all $c \in C$,
\item $\lsem \top \rsem_\gamma^V$ is given by the $\nu$-extent $\extent^\nu_\gamma : C \to S$,
\item $\lsem x \rsem_\gamma^V = V(x)$,
\item $\lsem \langle \lambda \rangle(\varphi_1,\ldots,\varphi_{\arity(\lambda)}) \rsem_\gamma^V = \gamma^* (\E_{\T_S} (\lsem \lambda \rsem_C (\lsem \varphi_1 \rsem_\gamma^V,\ldots,\lsem \varphi_{\arity(\lambda)} \rsem_\gamma^V)))$, where $\gamma^* : S^{\T_S F C} \to S^C$ denotes reindexing along $\gamma : C \to \T_S F C$.
\item $\lsem \mu x.\varphi \rsem_\gamma^{V \setminus \{x\}}$ ($\lsem \nu x.\varphi \rsem_\gamma^{V \setminus \{x\}}$) is the least (resp.~greatest) fixpoint of the operator on $S^C$ taking $p \in S^C$ to $\lsem \phi \rsem_\gamma^{V[p/x]}$, where the valuation $V[p/x] : \V \to S^C$ is defined as in Section~\ref{coalg-fixpoint-logics}.
\end{itemize}
(In the third clause, both the formal sum notation and the action of the monad multiplication have been extended pointwise to functions on $C$.) 
\end{defi} 

The fact that the operator used to interpret fixpoint formulas is monotone follows from the functoriality of both $\E_{\T_S}$ and the predicate liftings $\lsem \lambda \rsem$ with $\lambda \in \Lambda$. The existence of the required least and greatest fixpoints now follows from Theorem~\ref{thm:tarski}.

The semantics of $\mu\LL_\Lambda$ thus resembles that of coalgebraic fixpoint logics (see Section~\ref{coalg-fixpoint-logics}), with the following differences: (i) the interpretation of a formula is an \emph{$S$-valued} predicate over the state space, (ii) the \emph{extension} predicate lifting is used alongside the predicate liftings $\lsem \lambda \rsem$ to give semantics to modal formulas, and (iii) the constant $\top$ has a \emph{coinductive} interpretation. In particular, the use of the extension predicate lifting to abstract away branching behaviour is what makes $\mu\LL_\Lambda$ a \emph{linear-time logic}.

We argue that a coinductive interpretation of truth is a natural generalisation of the  semantics of (the existential variant of) the linear-time $\mu$-calculus, wherein a state satisfies a fixpoint formula if and only if there exists an infinite path from that state which satisfies the formula. In particular, a state in a $\T_S \circ F$-coalgebra from which \emph{no} maximal path exists interprets $\top$ as $0$, whereas the interpretation of $\top$ in a state which admits one or more maximal paths "measures" the set of \emph{all} such maximal paths, as a value in $S$ which is computed coinductively. It is precisely this coinductive interpretation of $\top$ that allows the existence of an equivalent path-based semantics.

\begin{rem}
\label{rem-disj}
A restricted form of disjunction can easily be incorporated into the logics. Specifically, for each $\lambda,\lambda' \in \Lambda$ with $\lambda \ne \lambda'$, one can define a new modal operator $\langle \lambda \rangle\_ \sqcup \langle \lambda' \rangle\_$, of arity equal to $\arity(\lambda) + \arity(\lambda')$, with associated predicate lifting $L : \Pred^{\arity(\lambda) + \arity(\lambda')} \to \Pred$ given by:
\begin{align*}
L_X(p_1,\ldots,p_{\arity(\lambda)}, p_1',\ldots,p_{\arity(\lambda')}')(f) & \!=\! \begin{cases}
p_1(x_1) \bullet \ldots \bullet p_{\arity(\lambda)}(x_{\arity(\lambda)}), & \!\!\!\text{if } f = \iota_\lambda(x_1,\ldots,x_{\arity{(\lambda)}}) \\ 
p_1'(y_1) \bullet \ldots \bullet p'_{\arity(\lambda')}(y_{\arity(\lambda')}), & \!\!\!\text{if } f = \iota_{\lambda'}(y_1,\ldots,y_{\arity{(\lambda')}}) \\ 
0, & \!\!\!\text{otherwise} \end{cases}\!.
\end{align*}
Thus, the formula $\langle \lambda \rangle (\phi_1,\ldots,\phi_{\arity(\lambda)}) \sqcup \langle \lambda' \rangle (\phi_1',\ldots,\phi'_{\arity(\lambda')})$ represents a disjunction which can be resolved using a one-step unfolding of the coalgebra structure (since each branch resulting from a one-step unfolding will match either $\langle \lambda \rangle$ or $\langle \lambda' \rangle$ or none of them). Furthermore, the definition of such \emph{guarded disjunctions} generalises to any choice of a finite subset of $\Lambda$. The qualitative interpretation of $\mu\LL_\Lambda$ over $F$-coalgebras also extends naturally to such guarded disjunctions. Then, in the presence of such enhanced modalities, \emph{deterministic} parity automata over $F$-structures\footnote{This includes automata over infinite words ($F = A \times \Id$) and also over infinite trees ($F = A \times \Id \times \Id$).} can be encoded as formulas of our logic. The reader is referred to \cite{CirsteaSH17}, wherein translations from $\mu\LL_\Lambda$-formulas to \emph{parity $\T_S \circ F$-automata} ($\T_S \circ F$-coalgebras equipped with a parity map) and back are described. It is worth noting that the syntax of the logics in loc.\,cit.~includes weighted sums; in the absence of weighted sums, the resulting automata are essentially deterministic ones\footnote{The resulting automata have \emph{at most} one transition labelled by any given $\lambda \in \Lambda$ from each state.}.
\end{rem}
\begin{asm}
In what follows, we assume that the additional modal operators described in Remark~\ref{rem-disj} are present in the logics. With this extended syntax, the \emph{qualitative} semantics of $\mu\LL_\Lambda$ (given by a variant of the semantics in Section~\ref{coalg-fixpoint-logics} which leaves out the cases for disjunctions and conjunctions of formulas) extends to the additional modal operators $\langle \lambda \rangle\_ \sqcup \langle \lambda' \rangle\_$, via predicate liftings whose components $L_X : (\Pow (X))^{\arity(\lambda) + \arity(\lambda')} \to \Pow (F X)$ are given by
\begin{align*}
L_X(Y_1,\ldots,Y_{\arity(\lambda)}, Y_1',\ldots,Y_{\arity(\lambda')}') ~=~ & \{\iota_\lambda(x_1,\ldots,x_{\arity{(\lambda)}}) \mid x_i \in Y_i \text{ for } i \in \{1,\ldots,\arity(\lambda) \}\,\} \,\cup \\
& \{\iota_{\lambda'}(x_1',\ldots,x'_{\arity{(\lambda')}}) \mid x_i' \in Y_i' \text{ for } i \in \{1,\ldots,\arity(\lambda') \}\,\}\,.
\end{align*}
(In the above, $\{0,1\}$-valued predicates on $X$ are identified with elements of $\Pow(X)$.) In other words, under the qualitative semantics, the formula $\langle \lambda \rangle\phi \sqcup \langle \lambda' \rangle\psi$ with $\lambda \ne \lambda'$ is equivalent to the disjunction $\langle \lambda \rangle\phi \vee \langle \lambda' \rangle\psi$.
\end{asm}
\begin{exa}
\label{exa:instances}
This example describes instantiations to concrete semirings of our quantitative logics. The characterisations provided below are a direct consequence of our main result (Theorem~\ref{thm-equiv}), and at this point only serve as intuitions for the semantics.
\begin{enumerate}
\item As noted in Remark~\ref{rem:nabla}, when $S = (\{0,1\},\vee,0,\wedge,1)$, the predicate liftings $\lsem \lambda \rsem$ with $\lambda \in \Lambda$ are closely related to the Nabla modality of coalgebraic logics for $F$-coalgebras \cite{Moss1999}. However, the use of the extension predicate lifting $\E_{\T_S}$ to abstract away branching results in the logic $\mu\LL_\Lambda$ being similar to (the existential variant of) the \emph{linear-time} $\mu$-calculus: a $\mu\LL_\Lambda$-formula holds in a state $s$ of a coalgebra with non-deterministic branching whenever a maximal trace (element of the final $F$-coalgebra) satisfying that formula in the qualitative sense can be exhibited from $s$. Yet, unlike the linear-time $\mu$-calculus, which has a path-based semantics and contains conjunction operators, our logics have a step-wise semantics which prevents conjunctions with the expected interpretation from being included.
\item For $S = ([0,1],+,0,*,1)$, $\mu\LL_\Lambda$-formulas measure the likelihood of satisfying a certain linear-time property. In spite of the absence of disjunctions or conjunctions from the logics, if one additionally assumes that $F$ specifies word-like behaviour (that is, $F = \coprod_{\lambda \in \Lambda} \Id^{\arity(\lambda)}$, with $\arity(\lambda) \in \{0,1\}$ for $\lambda \in \Lambda$), our logics have the same expressive power, when interpreted qualitatively over $F$-structures, as deterministic parity automata (see \cite{CirsteaSH17} for details). Thus, our quantitative logics match the expressive power of logics such as LTL or the linear-time $\mu$-calculus, when interpreted over \emph{probabilistic} models.
\item For $S = (\mathbb N^\infty,\min,\infty, +,0)$ or one of its bounded variants, $\mu\LL_\Lambda$-formulas measure the minimal cost needed to satisfy a certain linear-time property.
\end{enumerate} 
\end{exa}

\begin{exa}
\label{exa:express}
To illustrate the use of modalities incorporating guarded disjunctions, first let $F = A \times \Id \simeq \coprod\limits_{a \in A} \Id$ with $A$ finite, and define $\langle \overline a \rangle \phi ::= \bigsqcup_{b \in A \setminus \{a\}} \langle b \rangle \phi$ for $a \in A$. Then, the \emph{qualitative} properties stating that $a \in A$ appears (i) always, (ii) eventually, (iii) finitely often, and (iv) infinitely often in the unfolding of a state in an $F$-coalgebra are captured by the formulas (i) $\nu x.\langle a \rangle x$, (ii) $\mu x.(\langle a \rangle\top \sqcup \langle \overline{a} \rangle x)$, (iii) $\mu x.\nu y.(\langle a \rangle x \sqcup \langle \overline{a} \rangle y)$, and (iv) $\nu x.\mu y.(\langle a \rangle x \sqcup \langle \overline{a} \rangle y)$, respectively. Now let $F = \{*\} + A \times \Id \simeq \{*\} + \coprod\limits_{a \in A} \Id$ with $A$ finite, and define $\langle A\rangle \phi ::= \bigsqcup_{a \in A} \langle a\rangle \phi$. Then, the formula $\mu x.(\langle * \rangle \sqcup \langle A\rangle x)$ captures terminating behaviour. Finally, let $F = \{a,b\} \times \Id \times \Id$. In this case, the behaviour of states in $F$-coalgebras is \emph{tree-shaped}. Then, the formula $\mu x.(\langle a \rangle (\top,\top) \sqcup \langle b \rangle (x,x))$ captures the qualitative property that $a$ occurs eventually on \emph{every} branch of the unfolding of a state in an $F$-coalgebra. (Note that the property that $a$ occurs eventually on \emph{some} branch of the unfolding of a state in an $F$-coalgebra is \emph{not} expressible -- a non-guarded disjunction would be needed to make a choice between the left and the right branch.) All the above formulas can also be interpreted over coalgebras with non-deterministic, probabilistic or weighted branching, with the resulting semantics measuring the extent (possibility, likelihood and minimal cost, respectively) of exhibiting the corresponding qualitative property.
\end{exa}

\begin{exa}
Assume $F = \{*\} + A \times \Id$ and consider the formula $\mu x.(\langle a \rangle \top \sqcup \langle \overline{a}\rangle x)$, expressing the qualitative property "eventually $a$" (where $\langle \overline{a} \rangle$ is defined as in Example~\ref{exa:express}). Its semantics in the two coalgebras from Example~\ref{mu-nu-extent-example} is given by $(x \mapsto \frac{2}{5}, y \mapsto \frac{1}{10}, z \mapsto \frac{1}{5})$ and respectively $(x \mapsto 4, y \mapsto 4, z \mapsto 4)$:
\begin{enumerate}
\item For the first coalgebra, using the fact that the $\nu$-extent of state $y$ is $\frac{3}{5}$, it follows that $\lsem \langle a \rangle \top \rsem$ is given by $(x \mapsto \frac{1}{2} \times \frac{3}{5}= \frac{3}{10}, y \mapsto 0, z \mapsto 0)$. Then, $\lsem \mu x.(\langle a \rangle \top \sqcup \langle \overline{a} \rangle x)\rsem$ is obtained as the least solution of the following system of equations:
\begin{eqnarray*}
\begin{cases}
x = \frac{3}{10} + \frac{1}{2}z\\
y = \frac{1}{4}x\\
z = \frac{1}{4}x + \frac{1}{2}z
\end{cases}
\end{eqnarray*}
\item In the second coalgebra, $\lsem \langle a\rangle \top \rsem$ is given by $(x \mapsto 4, y \mapsto 0, z \mapsto 0)$. Then, $\lsem \mu x.(\langle a \rangle \top \sqcup \langle \overline{a} \rangle x)\rsem$ is the least (w.r.t.~$\sqsubseteq$, which is now $\ge$ on $\mathbb N^\infty$) solution of:
\begin{eqnarray*}
\begin{cases}
x = \min(4,1 + z)\\
y = x\\
z = \min(x,z)
\end{cases}
\end{eqnarray*}
\end{enumerate}
\end{exa}

We conclude this section by recalling an alternative, automata-theoretic characterisation of the semantics of $\mu\LL_\Lambda$, given in \cite{CirsteaSH17}, which will prove useful later on. A notion of \emph{quantitative parity automaton} is defined in \cite{CirsteaSH17} as a $\T_S \circ F$-coalgebra  $(A,\alpha)$ together with a function $\Omega : A \to \{1,2,\ldots\}$ with finite range, called a \emph{parity map}. A translation from $\mu\LL_\Lambda$-formulas to such parity automata is then provided in loc.\,cit., and an automata-theoretic characterisation of the semantics of $\mu\LL_\Lambda$ is also given. This characterisation makes use of a generalisation of the notions of $\mu$-/$\nu$-extent of a $\T_S \circ F$-coalgebra, defined as a \emph{nested} fixpoint which takes into account the parity map of a quantitative parity automaton. As is the case for existential LTL/the existential variant of the linear-time $\mu$-calculus, the semantics of a $\mu\LL_\Lambda$-formula in a given coalgebra is recovered as the (generalised) extent of a parity automaton obtained as the product (Definition~\ref{simple-prod-aut}) between the coalgebra in question and the automaton induced by the formula, with parities inherited from the automaton.

\begin{defiC}[\cite{CirsteaSH17}]
\label{extent}
Let $(A,\alpha,\Omega)$ be a quantitative parity automaton with $\ran(\Omega) \subseteq \{1,\ldots,n\}$, let $A_k = \{a \in A \mid \Omega(a) = k \}$, and let $\alpha_k = \alpha \circ \iota_k : A_k \to \T_S F A$ denote the restriction of $\alpha$ to $A_k$. The \emph{extent $\extent_\alpha = [\extent_1,\ldots,\extent_n] : A \to S$ of $(A,\alpha,\Omega)$} is the solution of the nested equational system\footnote{See e.g.~\cite{Niwinski2001} for a definition of nested equational systems and their solutions.}
\begin{align}
\label{eqn-extent}\begin{bmatrix}
u_1 & =_\mu & \mu_1 \circ \T_S (\bullet_F) \circ T_S F [u_1, \ldots,u_n] \circ \alpha_1\\
& \vdots \\
u_n & =_\eta & \mu_1 \circ \T_S (\bullet_F) \circ T_S F [u_1, \ldots, u_n] \circ \alpha_n
\end{bmatrix}
\end{align}
with $\eta = \mu$ ($\eta = \nu$) if $n$ is odd (resp.~even), with variables $u_k$ ranging over the lattice $(S^{A_k},\sqsubseteq)$ (and therefore $[u_1,\ldots,u_n] : A \to S$), and with the right-hand-sides of the equations pictured below:
\begin{align*}
\UseComputerModernTips\xymatrix{A_k \ar[r]^-{\alpha_k} & \T_S F A \ar[rrr]^-{\T_S F [u_1, \ldots,u_n]} & & & \T_S F S \ar[rr]^-{\T_S(\bullet_F)} & & \T_S S = \T_S^2 1 \ar[r]^-{\mu_1} & \T_S 1 = S}
\end{align*}
\end{defiC}
\begin{rem}
\label{rem:extents}
Given that $\extent_\alpha$ is defined as a \emph{nested} fixpoint whereas the $\extent^\nu_\alpha$ of Definition~\ref{extent-coalgebra} is defined as the solution of a similar equational system, but with all variables being $\nu$-variables, it follows immediately that $\extent_\alpha \sqsubseteq \extent^\nu_\alpha$.
\end{rem}
The main result of \cite{CirsteaSH17} is an automata-theoretic characterisation of the semantics of the logics $\mu\LL_\Lambda$. Below, $\Cl(\phi)$ is the \emph{closure} of a formula $\phi$, defined as usual for fixpoint logics, and the \emph{product} between a $\T_S \circ F$-coalgebra $(C,\gamma)$ and a quantitative parity automaton $(A,\alpha,\Omega)$ is another quantitative parity automaton obtained by endowing the product of Definition~\ref{simple-prod-aut} with parities inherited from $(A,\alpha,\Omega)$.

\begin{thmC}[\cite{CirsteaSH17}]
\label{thm:extent-char}
If $(\Cl(\phi),\beta,\Omega)$ with $\ran(\Omega) \subseteq \{1,\ldots,n\}$ is the parity automaton for a \emph{clean}\footnote{A formula $\phi \in \mu\LL_\Lambda$ is said to be \emph{clean} if no variable appears both free and bound, or is bound more than once, in $\phi$.} and \emph{strictly guarded}\footnote{A formula is said to be \emph{strictly guarded} if every fixpoint variable is immediately preceded by a modal operator.} formula $\phi \in \mu\LL_\Lambda$, $(C,\gamma)$ is a $\T_S \circ F$-coalgebra and $\extent_\alpha = [(\extent_{n})_{n \in \ran(\Omega)}] : A \to S$ is the extent of the product parity automaton $(A,\alpha,\Omega)$ of $(C,\gamma)$ and $(\Cl(\phi),\beta,\Omega)$, then $\lsem \phi \rsem_\gamma(c) = \extent_\alpha (c,\phi)$ for $c \in C$.
\end{thmC}

While the logics in \cite{CirsteaSH17} lack the propositional constant $\top$, the results in loc.\,cit.~can easily be generalised by including $\top$ in $\mu\LL_\Lambda$. This is because the formula $\top$ corresponds to a $\T_S \circ F$-coalgebra automaton with a single state (with even parity) and a $\lambda$-transition with weight $1 \in S$ for each $\lambda \in \Lambda$. The fact that $\Lambda$ may be infinite is not an issue, since taking a product with a $\T_S \circ F$-coalgebra will yield a coalgebra with finite branching. It is easy to see that each formula in $\mu\LL_\Lambda$ is equivalent to a strictly guarded one: since modal operators are the only non-nullary operators in $\mu\LL_\Lambda$, the only way a clean formula can fail to be strictly guarded is when it contains a sub-formula of the form $\eta x.y$. Such a sub-formula can be replaced by $y$ (if $x \ne y$) or $\bot$ (if $x = y$ and $\eta = \mu$) or $\top$ (if $x = y$ and $\eta = \nu$).

\section{Semiring-Valued Measures}
\label{semiring-measure-gen}

This section gives a generalisation of the notion of (real-valued) measure on a $\sigma$-algebra (see e.g.~\cite{Ash}) to measures valued into a partial semiring, and shows how standard measure extension results generalise to this setting.

Throughout this section, we fix a partial commutative semiring $(S,+,0,\bullet,1)$ satisfying Assumption~\ref{ass:cpo}. For such a semiring, we define a countable (partial) addition operation on $S$ as an extension of the binary semiring addition:
\begin{eqnarray}
\label{eq:csum}
\sum\limits_{i \in \omega} s_i := \sup\limits_{n \in \omega} (s_0 + \ldots + s_n)
\end{eqnarray}
If $S$ is partial, the above countable sum is defined iff all sums $s_0 + \ldots + s_n$ with $n \in \omega$ are defined. The definition exploits the fact that $s \sqsubseteq s + t$ for any $s,t \in S$ for which $s + t$ is defined, together with the existence of joins of increasing countable chains (see Assumption~\ref{ass:cpo}). 
\begin{rem}
\label{commutativity-rem}
It follows immediately from the definition of countable sums that $\sum\limits_{i \in \omega}s_i = \sum\limits_{i \in \omega}s_{f(i)}$ for any bijection $f : \omega \to \omega$. 
\end{rem}

\begin{rem}
\label{rem:distrib-countable-sums}
It also follows from the definition of countable sums, together with the distributivity of $\bullet$ over finite sums and the preservation of joins of increasing countable chains by $\bullet$ in each argument (Assumption~\ref{ass:cpo}), that $\bullet$ distributes over countable sums; that is, whenever $\sum\limits_{i \in \omega}s_i$ is defined, then so is $\sum\limits_{i \in \omega} (s \bullet s_i)$ and moreover, $\sum\limits_{i \in \omega} (s \bullet s_i) = s \bullet \sum\limits_{i \in \omega}s_i$. This is because when $\sum\limits_{i \in \omega}s_i$ is defined, all the partial sums $s_0 + \ldots + s_n$ are also defined, and therefore so is $s \bullet s_0 + \ldots + s \bullet s_n = s \bullet (s_0 + \ldots + s_n)$, for all $n \in \omega$; this, in turn, gives that $\sum\limits_{i \in \omega} (s \bullet s_i)$ is defined and equal to $\sup\limits_{n \in \omega} (s \bullet (s_0 + \ldots + s_n)) = s \bullet \sup\limits_{n \in \omega} (s_0 + \ldots + s_n) = s \bullet \sum\limits_{i \in \omega} s_i$.
\end{rem}

A further assumption on the semiring $(S,+,0,\bullet,1)$ is needed to develop the theory of $S$-valued measures. In particular, this is needed to prove countable sub-additivity of our semiring-valued measures. The results in Sections~\ref{path-based-sem} and \ref{equiv-sem} also rely on this assumption.

\begin{asm}
\label{ass:cl}
We assume that $\bullet$ preserves both suprema and infima in each argument, and that the following holds for all $A_i \subseteq S$ with $i \in \omega$:
\begin{eqnarray}
\label{asm-eq}
\sum\limits_{i \in \omega} \inf A_i = \inf \big\{\sum\limits_{i \in \omega} a_i \mid a_i \in A_i \text{ for } i \in \omega,\, \sum\limits_{i \in \omega} a_i  \text{ is defined\,} \big\}
\end{eqnarray}
whenever $\sum\limits_{i \in \omega} \inf A_i $ is defined.
\end{asm}
\begin{rem}
\label{rem:partial}
The following always holds when $(S,\sqsubseteq)$ is a complete lattice:
\begin{eqnarray*}
\sum\limits_{i \in \omega} \inf A_i \sqsubseteq \inf \big\{\sum\limits_{i \in \omega} a_i \mid a_i \in A_i \text{ for } i \in \omega,\, \sum\limits_{i \in \omega} a_i  \text{ is defined\,} \big\}
\end{eqnarray*}
provided that $\sum\limits_{i \in \omega} \inf A_i $ is defined. Assumption~\ref{ass:cl} strengthens this to an equality.
\end{rem}

\begin{rem}
We verify that our example semirings (see Example~\ref{example-semirings}) satisfy Assumption~\ref{ass:cl}.
\begin{enumerate}
\item We first consider the case $S = ([0,1], +, 0, *, 1)$. Clearly, the equality (\ref{asm-eq}) holds when $\sum\limits_{i \in \omega} \inf A_i  = 1$ -- one direction follows from Remark~\ref{rem:partial} and the other is immediate. When $\sum\limits_{i \in \omega} \inf A_i  < 1$, the fact that $1$ is a limit point can be used to construct $c_i \in [0,1]$ with $i \in \omega$ such that $\inf A_i < c_i$ for $i \in \omega$ and $\sum\limits_{i \in \omega} c_i$ is defined. For this, let $b_0 >0$ be such that $(\sum\limits_{i \in \omega}\inf A_i) < (\sum\limits_{i \in \omega}\inf A_i + b_0) < 1$. Now assuming $b_0,\ldots,b_i$ have been defined in such a way that $(\sum\limits_{i \in \omega}\inf A_i ~+ \sum\limits_{j \in \{0,\ldots,i\}} b_j) < 1$, let $b_{i+1} > 0$ be such that $(\sum\limits_{i \in \omega}\inf A_i + \sum\limits_{j \in \{0,\ldots,i\}} b_j) < (\sum\limits_{i \in \omega}\inf A_i + \sum\limits_{j \in \{0,\ldots,i+1\}} b_j) < 1$; finally, let $c_i = \inf A_i + b_i$ for $i \in \omega$. We now have that $\sum\limits_{i \in \omega}c_i = \sum\limits_{i \in \omega} \inf A_i + \sum\limits_{i \in \omega}b_i$ is defined since $\sum\limits_{i \in \omega} \inf A_i + \sum\limits_{j \in \{0,\ldots,i\}}b_j$ is defined for all $i \in \omega$. Then, since $\inf A_i < c_i$ for $i \in \omega$ and since $\le$ is a \emph{total} order on $[0,1]$, we have
\begin{eqnarray*}
\inf A_i = \inf \{a_i \in A_i \mid c_i \not\le a_i\} = \inf \{a_i \in A_i \mid a_i < c_i\} \text{ for }i \in \omega\,.
\end{eqnarray*}
The required inequality now follows from:
\begin{align*}
\sum\limits_{i \in \omega} \inf A_i  & = \sum\limits_{i \in \omega} \inf \{a_i \in A_i \mid a_i < c_i\} \\ & = \inf \{ \sum\limits_{i \in \omega} a_i \mid a_i \in A_i,\, a_i < c_i \text{ for } i \in \omega\,\} \\ & \sqsupseteq \inf \{\sum\limits_{i \in \omega} a_i \mid a_i \in A_i \text{ for } i \in \omega,\, \sum\limits_{i \in \omega} a_i \text{ is defined\,} \} 
\end{align*}
where the second equality above is a consequence of all the sums $\sum\limits_{i \in \omega} a_i$ with $a_i < c_i$ being defined in $[0,1]$ (since $\sum\limits_{i \in \omega} c_i$ is defined) and of the same equality holding over the positive reals, whereas the inequality follows from $\{ \sum\limits_{i \in \omega} a_i \mid a_i \in A_i,\, a_i < c_i \text{ for } i \in \omega\,\} \subseteq \{\sum\limits_{i \in \omega} a_i \mid a_i \in A_i \text{ for } i \in \omega,\, \sum\limits_{i \in \omega} a_i \text{ is defined\,} \}$.
\item Second, we consider the case $S = (\mathbb N^\infty,\min,\infty,+,0)$; in this case, $\sqsubseteq$ is the $\ge$ relation on $\mathbb N^\infty$ and $\sum\limits_{i \in \omega} a_i$ is given by $\inf\limits_{i \in \omega}a_i$. Then, the required inequality instantiates to:
\begin{eqnarray*}
\inf\limits_{i \in \omega} \sup A_i \le \sup \big\{\inf\limits_{i \in \omega} a_i \mid a_i \in A_i \text{ for } i \in \omega \big\}
\end{eqnarray*}
That this is true now follows by a relatively straighforward case analysis on whether the lhs equals $\infty$. A similar argument applies to bounded variants of the tropical semiring.
\item Finally, in the case when $S = (\{0,1\},\vee,0,\wedge,1)$, the proof is straightforward.
\end{enumerate}
\end{rem}

\begin{defi}
A collection of sets $\A \subseteq \Pow X$ is a \emph{$\sigma$-algebra} provided that:
\begin{enumerate}
\item $\emptyset \in \A$,
\item if $A \in \A$ then $X \setminus A \in \A$,
\item if $A_i \in \A$ for $i \in \omega$, then $\bigcup\limits_{i \in \omega}A_i \in \A$.
\end{enumerate}
\end{defi}
Thus, a $\sigma$-algebra is also closed under (finite and) countable intersections (as $\bigcap\limits_{i \in \omega}A_i = X \setminus \bigcup\limits_{i \in \omega}(X \setminus A_i)$).  Also, if $A,B \in \A$, then $B \setminus A = B \cap (X \setminus A) \in \A$.

The next definition generalises real-valued measures on $\sigma$-algebras to measures valued into a partial semiring $(S,+,0,\bullet,1)$.
\begin{defi}[$S$-valued measure]
\label{semiring-valued-measure}
An \emph{$S$-valued measure} on a $\sigma$-algebra $\A$ is a function $\mu : \A \to S$ such that:
\begin{itemize}
\item $\mu(\emptyset) = 0$,
\item if $A_i \in \A$ for $i \in \omega$ are pairwise disjoint, then $\sum\limits_{i \in \omega} \mu(A_i)$ is defined and moreover,
$\mu(\bigcup\limits_{i \in \omega}A_i) = \sum\limits_{i \in \omega} \mu(A_i)$.
\end{itemize}
\end{defi}
It follows immediately that any $S$-valued measure on $\A$ is monotone, that is, $\mu(A) \sqsubseteq \mu(B)$ whenever $A,B \in \A$ are such that $A \subseteq B$: in this case, $\mu(B) = \mu(A) + \mu(B \setminus A) \sqsupseteq \mu(A)$. Moreover, as for standard measures, the following holds.
\begin{prop}
\label{prop:countable-unions}
If $\mu : \A \to S$ is an $S$-valued measure on a $\sigma$-algebra $\A$ and $A_i \in \A$ with $i \in \omega$ are such that $A_i \subseteq A_{i+1}$ for $i \in \omega$, then $\mu(\bigcup\limits_{i \in \omega} A_i) = \sup_{i \in \omega} \mu(A_i)$.
\end{prop}
\begin{proof}
We have $\mu(\bigcup\limits_{i \in \omega}A_i) = \mu(A_0 \cup (A_1 \setminus A_0) \cup (A_2 \setminus A_1) \cup \ldots) = \mu(A_0) + \mu(A_1\setminus A_0) + \mu(A_2 \setminus A_1) + \ldots = \sup_{i \in \omega} (\mu(A_0) + \mu(A_1\setminus A_0) + \ldots + \mu(A_i \setminus A_{i-1})) = \sup_{i \in \omega} \mu(A_i)$.
\end{proof}
Next, we recall the notions of \emph{semi-ring}\footnote{The  terminology clash with the notion of semiring should not cause any problems, given the different contexts in which \emph{semirings} (like~$S$) and \emph{semi-rings} (of sets) are used.}, \emph{ring} and \emph{field of sets}, and show how $S$-valued measures on $\sigma$-algebras arise from certain $S$-valued functions on semi-rings.

\begin{defi}
\label{semi-ring}
A collection of sets $\S \subseteq \Pow X$ is a \emph{semi-ring} provided that:
\begin{enumerate}
\item $\emptyset \in \S$,
\item if $A,B \in \S$, then $A \cap B \in \S$,
\item if $A,B \in \S$, there exist $A_1,\ldots,A_n  \in \S$ pairwise disjoint such that $A \setminus B = A_1 \cup \ldots \cup A_n$.
\end{enumerate}
A non-empty collection of sets $\R \subseteq \Pow X$ is a \emph{ring} provided that:
\begin{enumerate}
\item if $A,B \in \R$, then $A \cup B \in \R$,
\item  if $A,B \in \R$, then $A \setminus B \in \R$.
\end{enumerate}
A ring $\R \subseteq \Pow X$ which contains $X$ is called a \emph{field}.
\end{defi}
Thus, every ring is a semi-ring. Moreover, given a semi-ring $\S \subseteq \Pow X$, the smallest ring containing $\S$ (defined as the intersection of \emph{all} rings containing $\S$) is obtained by closing $\S$ under finite unions of pairwise disjoint sets. That this yields a ring (i.e.~it is also closed under relative complement) follows easily from $(\bigcup \limits_{i = 1}^{n} A_i) \setminus (\bigcup \limits_{j = 1}^{m} B_j) = \bigcup \limits_{i = 1}^{n} \bigcap \limits_{j = 1}^{m} (A_i \setminus B_j)$ together with the distributivity of finite intersections over finite unions.

The next definition generalises the notion of measure on a ring (see e.g.~\cite{Ash}) to \emph{$S$-valued} measures.
\begin{defi}
\label{def:measure-ring}
An \emph{$S$-valued measure on a ring $\R$} is a function $\mu : \R \to S$  such that:
\begin{itemize}
\item $\mu(\emptyset) = 0$,
\item if $A_i \in \R$ for $i \in \omega$ are pairwise disjoint and $\bigcup\limits_{i \in \omega}A_i \in \R$, then $\sum\limits_{i \in \omega} \mu(A_i)$ is defined and moreover,
$\mu(\bigcup\limits_{i \in \omega}A_i) = \sum\limits_{i \in \omega} \mu(A_i)$.
\end{itemize}
\end{defi}

\begin{rem}
\label{rem:definedness}
We note that any $S$-valued measure on a ring is \emph{finitely} additive, that is, if $A_1,\ldots,A_n \in \R$ are pairwise disjoint, then $\sum\limits_{i=1}^n \mu(A_i)$ is defined and moreover,  $\mu(\bigcup\limits_{i =1}^nA_i) = \sum\limits_{i =1}^n \mu(A_i)$ -- this follows by extending the finite family $A_1,\ldots,A_n$ to a countable family $(A_i)_{i \in \omega}$, where $A_i = \emptyset$ for $i > n$.  
\end{rem}

The following is an immediate consequence of Definition~\ref{def:measure-ring}.
\begin{prop}
\label{prop:ascending-ring}
Let $\mu : \R \to S$ be an $S$-valued measure on a ring $\R$, and let $A_1 \subseteq A_2 \subseteq \ldots$ be such that $A_i \in \R$ for $i \in \omega$ and also $\bigcup\limits_{i\in \omega} A_i \in \R$. Then, $\sup_{i \in \omega}\mu(A_i) = \mu(\bigcup\limits_{i\in \omega} A_i)$.
\end{prop}
\begin{proof}
Immediate from the definition of an $S$-valued measure on a ring, by considering the pairwise disjoint family $(A_i \setminus (A_0 \cup \ldots \cup A_{i-1}))_{i \in \omega}$.
\end{proof}

The following generalisation of a standard measure theory result to $S$-valued measures can also be proved.
\begin{prop}
\label{prop:ext-semi-ring}
Let $\S\subseteq \Pow X$ be a semi-ring and let $\mu : \S \to S$ be such that:
\begin{enumerate}
\item $\mu(\emptyset) = 0$,
\item whenever $A_i \in \S$ for $i \in \{1,\ldots,n\}$ are pairwise disjoint, the sum $\sum\limits_{i \in \{1,\ldots,n\}}\mu(A_i)$ is defined,
\item whenever $A_i \in \S$ for $i \in \{1,\ldots,n\}$ are pairwise disjoint and moreover, $\bigcup\limits_{i \in \{1,\ldots,n\}} A_i \in \S$, $\mu(\bigcup\limits_{i \in \{1,\ldots,n\}} A_i) = \sum\limits_{i \in \{1,\ldots,n\}}\mu(A_i)$.
\end{enumerate}
Then $\mu$ has a unique extension to an $S$-valued measure on the smallest ring $\R$ containing $\S$.
\end{prop}
\begin{proof}
Each $A \in \R$ can be written as $\bigcup\limits_{i \in \{1,\ldots,n\}} A_i$ with $A_i \in \S$ for $i \in \{1,\ldots,n\}$ pairwise disjoint. We let $\mu(A) = \sum\limits_{i \in \{1,\ldots,n\}}\mu(A_i)$. To prove the independence of this definition on the choice of $(A_i)_{i \in \{1,\ldots,n\}}$, let $B_j \in \S$ for $j \in \{1,\ldots,m\}$ be pairwise disjoint and such that $A = \bigcup\limits_{j \in \{1,\ldots,m\}} B_j$. Then
{\small \begin{eqnarray*}
\sum\limits_{i \in \{1,\ldots,n\}}\mu(A_i) = \sum\limits_{i \in \{1,\ldots,n\}}\sum\limits_{j \in \{1,\ldots,m\}}\mu(A_i \cap B_j) = \sum\limits_{j \in \{1,\ldots,m\}}\sum\limits_{i \in \{1,\ldots,n\}}\mu(A_i \cap B_j) = \sum\limits_{j \in \{1,\ldots,m\}}\mu(B_j)
\end{eqnarray*}}
with all the sums above being defined. That $\mu$ as defined above is an $S$-valued measure on $\R$ is immediate. 
\end{proof}

\begin{defi}
A \emph{$\sigma$-ring} is a ring which is closed under countable unions.
\end{defi}
Thus, a $\sigma$-ring which contains $X$ is a $\sigma$-algebra.

Given a ring $\R$, we write $\S(\R)$ for the $\sigma$-ring generated by $\R$. We also write $\H(\R)$ for the smallest subset of $\Pow X$ which contains $\S(\R)$ and is closed under taking subsets. It is not difficult to show that $\H(\R)$ is itself a $\sigma$-ring (see e.g.~\cite[Section~5.1]{deBarra}). Moreover, if $\R$ is a field, then $\S(\R)$ is a $\sigma$-algebra and $\H(\R) = \Pow X$.
   
We show in what follows that every $S$-valued measure on a field extends to an $S$-valued measure on the induced $\sigma$-algebra. A generalisation of the standard notion of \emph{outer measure} (see e.g.~\cite[Definition~1.3.4]{Ash} or \cite[Section~17.3]{Royden}) will be used to derive this result.

\begin{defi}[$S$-valued outer measure]
\label{def:outer}
Let $\R$ be a ring. A function $\mu^* : \H(\R) \to S$ is an \emph{$S$-valued outer measure} provided that:
\begin{enumerate}
\item $\mu^*(\emptyset) = 0$,
\item\label{out2} if $A,B \in \H(\R)$ and $A \subseteq B$ then $\mu^*(A) \sqsubseteq \mu^*(B)$,
\item\label{out3} if $A_i \in \H(\R)$ for $i \in \omega$ are pairwise disjoint and $\sum\limits_{i \in \omega} \mu^*(A_i)$ is defined, then $\mu^*(\bigcup\limits_{i \in \omega}A_i) \sqsubseteq \sum\limits_{i \in \omega} \mu^*(A_i)$.
\end{enumerate}
\end{defi}

\begin{rem}
In order to deal with the case when $S$ is a \emph{partial} semiring (in which case, given a countable collection $(A_i \in \H(\R))_{i \in \omega}$, $\sum\limits_{i \in \omega} \mu^*(A_i)$ may not be defined, even if the $A_i$s are pairwise disjoint), the above definition differs slightly from the standard definition of an outer measure. However, when $S$ is a \emph{total} semiring, condition (\ref{out3}) of Definition~\ref{def:outer} generalises the standard one (see.e.g.~\cite[Def.\,1.3.4]{Ash})  -- note that requiring it for pairwise disjoint, countable collections of subsets means that it also holds for arbitrary countable collections of subsets. Also, when $S$ is the (partial) probabilistic semiring, condition (\ref{out3}) of Definition~\ref{def:outer} is equivalent to the standard one used in the definition of outer probability measures, since whenever for a countable collection $(A_i)_{i \in \omega}$ the sum $\sum\limits_{i \in \omega} \mu^*(A_i)$ is not defined (in $[0,1]$), the required inequality holds automatically in $\mathbb R_{\ge 0}$.
\end{rem}

The next definition also generalises a standard one (see e.g.~\cite[Section1.3]{Ash} or \cite[Section~17.3]{Royden}).

\begin{defi}
\label{def:outer-measurable}
Given an $S$-valued outer measure $\mu^* : \H(\R) \to S$, we call a set $E \in \H(\R)$ \emph{$\mu^*$-measurable} if for every $A \in \H(\R)$, the sum $\mu^*(A \cap E) + \mu^*(A \cap (X \setminus E))$ is defined and moreover,
\begin{eqnarray}
\label{eq:def}
\mu^*(A) = \mu^*(A \cap E) + \mu^*(A \cap (X \setminus E))\,.
\end{eqnarray}
\end{defi}

\begin{rem}
At first sight, the above definition also differs from the standard definition of measurability induced by an outer measure when the semiring $S$ is partial. However, note that in the case of the probabilistic semiring, the standard definition, which considers the rhs of (\ref{eq:def}) in Definition~\ref{def:outer-measurable} as an element of $\mathbb R_{\ge 0}$, forces the sum in the rhs to belong to $[0,1]$; thus, in this case the two definitions are equivalent.
\end{rem}
The next result now generalises a standard one from measure theory.

\begin{prop}
\label{prop:sigma-alg-M}
If $\mu^* : \H(\R) \to S$ is an $S$-valued outer measure, then:
\begin{enumerate}
\item The collection $\M$ of $\mu^*$-measurable subsets forms a $\sigma$-algebra.
\item The restriction of $\mu^*$ to $\M$ is an $S$-valued measure.
\end{enumerate}
\end{prop}
\begin{proof}
In spite of our slightly different definition of outer measure, the proof is exactly the same as the standard one (see e.g.~\cite[Section~17.4]{Royden}): one first shows that $\M$ is closed under finite unions of pairwise disjoint subsets and that $\mu^*$ is finitely additive on $\M$, and then extends this to countable unions of pairwise disjoint sets and countable additivity. The requirement of Definition~\ref{def:outer-measurable} that the sum in the rhs of (\ref{eq:def}) is defined means that all the sums appearing in the standard proof are defined. 
\end{proof}

We are now ready to show how an $S$-valued measure on a field extends to an $S$-valued measure on the induced $\sigma$-algebra. As in standard measure theory, this is done with the help of an outer measure.

\begin{defi}
\label{def:induced-outer-meas}
Let $\R$ be a ring and let $\mu : \R \to S$ be an $S$-valued measure on $\R$. The \emph{outer measure $\mu^* : \H(\R) \to S$ induced by $\mu$} is given by
\begin{eqnarray*}
\mu^*(A) = \inf \{\,\sum\limits_{n \in \omega} \mu(E_n) \mid (E_n \in \R)_{n \in \omega} \text{ pairwise disjoint},~ A \subseteq \bigcup\limits_{n \in \omega} E_n \}
\end{eqnarray*}
\end{defi}
Note that the pairwise disjointness of $(E_n)_{n \in \omega}$ together with $\R$ being a ring means that each sum $\sum\limits_{i=1}^n \mu(E_i)$ with $n \in \omega$ is defined (see Remark~\ref{rem:definedness}), and therefore so is $\sum\limits_{n \in \omega}\mu(E_n)$. Also, note that the infima needed in Definition~\ref{def:induced-outer-meas} exist as a result of $(S,\sqsubseteq)$ being a complete lattice (see Assumption~\ref{ass:cpo}).

\begin{prop}
\label{prop:outer}~
\begin{enumerate}
\item $\mu^*$ is an outer measure on $\H(\R)$.
\item $\mu^*(E) = \mu(E)$ for $E \in \R$.
\end{enumerate}
\end{prop}
\begin{proof}
For the first statement, the fact that $\mu^*(\emptyset) = 0$ is immediate, and so is the monotonicity of $\mu^*$. (For the latter, (\ref{c2}) of Remark~\ref{rem-semiring} is needed.) The proof of countable sub-additivity of $\mu^*$ rests on Assumption~\ref{ass:cl}. To show countable sub-additivity, fix pairwise disjoint sets $A_i \in \H(\R)$ with $i \in \omega$, such that $\sum\limits_{i \in \omega} \mu^*(A_i)$ is defined. To show $\mu^*(\bigcup\limits_{i \in \omega}A_i) \sqsubseteq \sum\limits_{i \in \omega} \mu^*(A_i)$, let $(E_n^i \in \R)_{n \in \omega}$ be pairwise disjoint and such that $A_i \subseteq \bigcup\limits_{n \in \omega}E_n^i$, for each $i \in \omega$. Now view $(E_n^i \in \R)_{n \in \omega,i \in \omega}$ as a single countable family $(F_k)_{k \in \omega}$, and consider the countable, \emph{pairwise disjoint} family $(F'_k \in \R)_{k \in \omega}$ given by $F'_k = F_k \setminus (F_0 \cup \ldots \cup F_{k-1})$ for $k \in \omega$. (Recall that $\R$ is a ring and therefore $F'_k \in \R$ for $k \in \omega$.) Then, $\bigcup\limits_{i \in \omega}A_i \subseteq \bigcup\limits_{k \in \omega}F_k = \bigcup\limits_{k \in \omega}F'_k$, and therefore, by definition of $\mu^*$, $\mu^*(\bigcup\limits_{i \in \omega}A_i) \sqsubseteq \sum\limits_{k \in \omega} \mu(F'_k)$. At the same time, monotonicity of $\mu$ gives $\sum\limits_{k \in \omega} \mu(F'_k) \sqsubseteq \sum\limits_{k \in \omega} \mu(F_k) = \sum\limits_{i \in \omega} \sum\limits_{n \in \omega} \mu(E_n^i)$ whenever $\sum\limits_{i \in \omega} \sum\limits_{n \in \omega} \mu(E_n^i)$ is defined. (For the latter equality, Remark~\ref{commutativity-rem} is used.) Hence, $\mu^*(\bigcup\limits_{i \in \omega}A_i) \sqsubseteq \sum\limits_{i \in \omega} \sum\limits_{n \in \omega} \mu(E_n^i)$ whenever the latter sum is defined. As this holds for every choice of $(E_n^i \in \R)_{n \in \omega}$ with $(E_n^i)_{n \in \omega}$ pairwise disjoint and $A_i \subseteq \bigcup\limits_{n \in \omega}E_n^i$ with $i \in \omega$ for which $\sum\limits_{i \in \omega}\sum\limits_{n \in \omega}\mu(E_n^i)$ is defined, we now have:
\begin{eqnarray*}
\mu^*(\bigcup\limits_{i \in \omega}A_i) \sqsubseteq \inf \{  \sum\limits_{i \in \omega} \sum\limits_{n \in \omega} \mu(E_n^i) \mid (E_n^i \in \R)_{n \in \omega} \text{ pairwise disjoint},\, A_i \subseteq \bigcup\limits_{n \in \omega}E_n^i \text{ for } i \in \omega,~\\
 \sum\limits_{i \in \omega}\sum\limits_{n \in \omega}\mu(E_n^i) \text{ is defined\,}\}
\end{eqnarray*}
The above together with Assumption~\ref{ass:cl}, gives:
\begin{eqnarray*}
\mu^*(\bigcup\limits_{i \in \omega}A_i) \sqsubseteq \sum\limits_{i \in \omega} \inf \{ \sum\limits_{n \in \omega} \mu(E_n^i) \mid (E_n^i \in \R)_{n \in \omega} \text{ pairwise disjoint},\, A_i \subseteq \bigcup\limits_{n \in \omega}E_n^i \}
\end{eqnarray*}
That is, $\mu^*(\bigcup\limits_{i \in \omega}A_i) \sqsubseteq \sum\limits_{i \in \omega} \mu^*(A_i)$.

For the second statement, the fact that $\{E\}$ is a cover for $E \in \R$ immediately gives $\mu^*(E) \sqsubseteq \mu(E)$. To show $\mu(E) \sqsubseteq \mu^*(E)$, we must show $\mu(E) \sqsubseteq \sum\limits_{n \in \omega} \mu(E_n)$ for each pairwise disjoint family $(E_n \in \R)_{n \in \omega}$ such that $E \subseteq \bigcup\limits_{n \in \omega} E_n$. Given such a family $(E_n)_{n \in \omega}$, the family $(E \cap (E_0 \cup \ldots \cup E_n))_{n \in \omega}$ satisfies $\bigcup\limits_{n \in \omega}(E \cap (E_0 \cup \ldots \cup E_n)) = E \cap (\bigcup\limits_{n \in \omega} E_n) = E \in \R$, and then Proposition~\ref{prop:ascending-ring} gives $\sup_{n \in \omega} \mu(E \cap (E_0 \cup \ldots \cup E_n)) = \mu(E)$. On the other hand, $\mu(E \cap (E_0 \cup \ldots \cup E_n)) \sqsubseteq \mu(E_0 \cup \ldots \cup E_n) = \sum\limits_{i \in \{1,\ldots,n\}} \mu(E_i) \sqsubseteq \sum\limits_{n \in \omega} \mu(E_n)$ for all $n \in \omega$. Putting these together we have $\mu(E) = \sup_{n \in \omega} \mu(E \cap (E_0 \cup \ldots \cup E_n)) \sqsubseteq \sum\limits_{n \in \omega} \mu(E_n)$ as required.
\end{proof}

We now proceed as in the standard case to show that if $\mu^*$ is as in Definition~\ref{def:induced-outer-meas}, then the elements of $\R$ are $\mu^*$-measurable. For this, we additionally assume that $\R$ is a field.

\begin{prop}
\label{prop:field}
Let $\R$ be a field, and let $\mu : \R \to S$ and $\mu^* : \H(\R) \to S$ be as in Definition~\ref{def:induced-outer-meas}. Then each $E \in \R$ is $\mu^*$-measurable.
\end{prop}
\begin{proof}
We must show that for each $A \in \H(\R)$, the sum $\mu^*(A \cap E) + \mu^*(A \cap (X \setminus E))$ is defined and moreover, $\mu^*(A) = \mu^*(A \cap E) + \mu^*(A \cap (X \setminus E))$.

Since $\{E\}$ is a finite (and therefore countable) cover for $A \cap E$, and $\{X \setminus E\}$ is a finite cover for $A \cap (X \setminus E)$, it immediately follows that $\mu^*(A \cap E) \sqsubseteq \mu(E)$ and $\mu^*(A \cap (X \setminus E)) \sqsubseteq \mu(X \setminus E)$.  Definedness of $\mu^*(A \cap E) + \mu^*(A \cap (X \setminus E))$ now follows by (\ref{c2}) of Remark~\ref{rem-semiring}. Moreover, since $\mu^*$ is an outer measure and $A \cap E) \cap (A \cap (X \setminus E)) = \emptyset$, we immediately get $\mu^*(A) \sqsubseteq \mu^*(A \cap E) + \mu^*(A \cap (X \setminus E))$. For the reverse direction, note that any cover $(E_i)_{i \in \omega}$ for $A$ consisting of pairwise disjoint elements of $\R$ yields a cover $(E_i \cap E)_{i \in \omega}$ for $A \cap E$ and a cover $(E_i \cap (X \setminus E))_{i \in \omega}$ for $A \cap (X \setminus E)$, each consisting of pairwise disjoint elements. Then,
\begin{align*}
\sum\limits_{i \in \omega}\mu(E_i) & = \sum\limits_{i \in \omega}(\mu(E_i \cap E) + \mu(E_i \cap (X \setminus E))) \tag{Remark~\ref{rem:definedness}}\\
& = \sum\limits_{i \in \omega} \mu(E_i \cap E) + \sum\limits_{i \in \omega} \mu(E_i \cap (X \setminus E)) \tag{Remark~\ref{commutativity-rem}} \\
& \sqsupseteq \mu^*(A \cap E) + \mu^*(A \cap (X \setminus E)) \\ \tag{definition of $\mu^*(A \cap E)$ and $\mu^*(A \cap (X \setminus E))$}
\end{align*}
The above together with the definition of $\mu^*(A)$ now gives $\mu^*(A) \sqsupseteq \mu^*(A \cap E) + \mu^*(A \cap (X \setminus E))$. This concludes the proof.
\end{proof}

We are now ready to formulate our generalised measure extension result.

\begin{thm}
\label{thm:measure-extension}
Let $\mu : \R \to S$ be a measure on a field $\R$. Then, $\mu$ extends to a measure on the $\sigma$-algebra generated by $\R$.
\end{thm}
\begin{proof}
By Proposition~\ref{prop:outer}, the $S$-valued function $\mu^* : \H(\R) \to S$ of Definition~\ref{def:induced-outer-meas} is an outer measure which extends $\mu$. Then, by Proposition~\ref{prop:sigma-alg-M}, the set $\M$ of $\mu^*$-measurable sets is a $\sigma$-algebra containing $\R$, and the restriction of $\mu^*$ to $\M$ is an $S$-valued measure. Finally, by Proposition~\ref{prop:field}, $\M$ contains $\R$, and therefore also the $\sigma$-algebra generated by $\R$. This makes $\mu^*$ an $S$-valued measure on the $\sigma$-algebra generated by $\R$.
\end{proof}

In the standard case, the uniqueness of an extension of a measure $\mu$ on a field $\R$ to a measure on the generated $\sigma$-algebra requires $\mu$ to be  \emph{$\sigma$-finite} (cf.~Carath\'eodory's extension theorem, see e.g.~\cite[Theorem~1.3.10]{Ash}). Under the $\sigma$-finiteness assumption, the induced measure additionally satisfies $\mu(\bigcap\limits_{i \in \omega}A_i) = \inf\limits_{i \in \omega} \mu(A_i)$ with $A_0 \supseteq A_1 \supseteq \ldots$ a decreasing countable family of measurable sets. (Recall from Proposition~\ref{prop:countable-unions} that $S$-valued measures on $\sigma$-algebras always satisfy 
$\mu(\bigcup\limits_{i \in \omega}A_i) = \sup_{i \in \omega} \mu(A_i)$
 whenever $A_0 \subseteq A_1 \subseteq \ldots$ is an increasing countable family of measurable sets.)  Since the $S$-valued measures used later to provide a path-based semantics for our logics do not satisfy the equality $\mu_j(\bigcap\limits_{i \in \omega}A_i) = \inf\limits_{i \in \omega} \mu_j(A_i)$ for all $A_0 \supseteq A_1 \supseteq \ldots$ with $A_i$ measurable for $i \in \omega$ (see Example~\ref{counter-ex}), we are not interested in a uniqueness result generalising the above-mentioned result.

\section{Path-Based Semantics for Quantitative Linear-Time Logics}
\label{path-based-sem}

This section provides an alternative, path-based semantics for the logic $\mu\LL_\Lambda$, akin to the path-based semantics of LTL and the linear-time $\mu$-calculus (interpreted over either non-deterministic or probabilistic transition systems).

We begin by defining a $\sigma$-algebra structure on the set of paths from a given state in a $\T_S \circ F$-coalgebra $(C,\gamma)$. We subsequently show how to define an $S$-valued measure (Definition~\ref{semiring-valued-measure}) on this $\sigma$-algebra.

Similarly to the semantics of probabilistic LTL (see e.g.~\cite[Section~10.3]{Baier}), our $\sigma$-algebra is induced by so-called \emph{cylinder sets}. We begin by associating to each path fragment $q$ in $(C,\gamma)$ a cylinder set.

\begin{defi}
Let $q \in I_C$ be a path fragment from $c$ in $(C,\gamma)$. Its associated \emph{cylinder set} $\Cyl(q)$ is given by
\begin{eqnarray*}
\Cyl(q) = \{p \in \Paths_c \mid q \in \pref(p)\}
\end{eqnarray*}
\end{defi}
We note that only paths whose transitions have non-zero weights in $(C,\gamma)$ are considered.  
\begin{prop}
\label{prop:semi-ring}
For $c \in C$, the set 
\begin{eqnarray*}
{\Sigma}_c := \{ \emptyset\} \cup \{\,\Cyl(q) \mid q \text{ is a path fragment from $c$ in $(C,\gamma)$}\,\}
\end{eqnarray*}
is a semi-ring (Definition~\ref{semi-ring}). Moreover, $\Paths_c \in \Sigma_c$.
\end{prop}
\begin{proof}
(1) of Definition~\ref{semi-ring} clearly holds. For (2), note that $\Cyl(q_1) \cap \Cyl(q_2)$ is either $\emptyset$, if $q_1$ and $q_2$ are incompatible (Definition~\ref{defi:pref}), or $\Cyl(q)$, if $q$ is the smallest path fragment from $c$ such that $q_1 \in \pref(q)$ and $q_2 \in \pref(q)$, otherwise. For (3) of Definition~\ref{semi-ring}, assuming $A = \Cyl(p)$ and $B = \Cyl(q)$ with $A \setminus B \ne \emptyset$, unfolding the coalgebra structure $\gamma$ a number of times equal to $\max(\depth(p),\depth(q))$ allows us to identify a \emph{finite} number of path fragments $q_1,\ldots,q_n$ from $c$ in $(C,\gamma)$, such that $\Cyl(q_i) \cap \Cyl(q_j) = \emptyset$ for $i,j \in \{1,\ldots,n\}$ with $i \ne j$, and such that $\Cyl(p) \setminus \Cyl(q) = \Cyl(q_1) \cup \ldots \cup \Cyl(q_n)$. Finally, $\Paths_c \in \Sigma_c$ follows from $\Paths_c = \Cyl(q_c)$, where $q_c$ is the path fragment given by $(c,\iota_1(*))$. 
\end{proof}

\begin{defi}
\label{defi-measurable-sets}
For a $\T_S \circ F$-coalgebra $(C,\gamma)$ and $c \in C$, the set $\M_c \subseteq \Pow \Paths_c$ of \emph{measurable sets of paths from $c$} is 
the smallest $\sigma$-algebra containing $\Sigma_c$ (see Section~\ref{semiring-measure-gen}).
\end{defi}
Equivalently, $\M_c$ is the smallest $\sigma$-algebra induced by the smallest ring $\R_c$ containing $\Sigma_c$. The elements of $\R_c$ are finite unions of pairwise disjoint elements of $\Sigma_c$. Moreover, by Proposition~\ref{prop:semi-ring}, $\R_c$ is a field. Then, by Theorem~\ref{thm:measure-extension}, any $S$-valued measure $\mu$ on $\R_c$ (and therefore any $S$-valued measure on $\Sigma_c$) extends to an $S$-valued measure $\mu^*$ on the induced $\sigma$-algebra, constructed as in Section~\ref{semiring-measure-gen} with the help of an outer-measure.

\begin{defi}[Induced measure]
\label{induced-measure}
Let $(C,\gamma)$ be a $\T_S \circ F$-coalgebra. For $c \in C$, let $\mu_\gamma  : \Sigma_c \to S$ be given by:
\begin{enumerate}
\item $\mu_\gamma(\emptyset)=0$,
\item For $q$ a path fragment from $c$ in $(C,\gamma)$, $\mu_\gamma(\Cyl(q))$ is defined by induction on the structure of $q$:
\begin{enumerate}
\item If $\iota^{-1}_C(q) = (c,\iota_1(*))$, then $\mu_\gamma(\Cyl(q)) = \extent^\nu_\gamma(c)$ (the $\nu$-extent of Definition~\ref{extent-coalgebra}),
\item If $\iota^{-1}_C(q) = (c,\iota_\lambda (q_1,\ldots,q_{\arity(\lambda)}))$ for some $\lambda \in \Lambda$ and $\pi_1(\iota^{-1}_C(q_i)) = c_i$ for $i \in \{1,\ldots,\arity(\lambda)\}$, then
\begin{eqnarray*}
\mu_\gamma(\Cyl(q)) = \gamma(c)(\iota_\lambda(c_1,\ldots,c_{\arity(\lambda)})) \bullet \mu_\gamma(\Cyl(q_1)) \bullet \ldots \bullet \mu_\gamma(\Cyl(q_{\arity(\lambda)}))
\end{eqnarray*}
\end{enumerate}
\end{enumerate}
Also, let $\mu_\gamma : \R_c \to S$ be given by $\mu_\gamma(\bigcup\limits_{i \in \{1,\ldots,n\}}C_i) = \sum_{i \in \{1,\ldots,n\}} \mu_\gamma(C_i)$ for each finite family $(C_i)_{i \in \{1,\ldots,n\}}$ with $C_i \in \Sigma_c$ and $C_i \cap C_j = \emptyset$ for $i \ne j$. The \emph{induced measure} $\mu_\gamma : \M_c \to S$ is the extension of $\mu_\gamma : \R_c \to S$ to an $S$-valued measure on $\M_c$, as given by the proof of Theorem~\ref{thm:measure-extension}.
\end{defi}

The correctness of the above definition (i.e.~the existence of $\mu_\gamma : \R_c \to S$) is guaranteed by the additivity of $\mu_\gamma : \Sigma_c \to S$:

\begin{prop} Let $\mu_\gamma : \Sigma_c \to S$ be as in Definiton~\ref{induced-measure}. Then, $\mu_\gamma$ satisfies the assumptions of Proposition~\ref{prop:ext-semi-ring}.
\end{prop} 
\begin{proof}
That the first two assumptions are satisfied is an easy exercise (with induction and use of (\ref{c3}) of Remark~\ref{rem-semiring} being needed for the second assumption). For the third assumption, let $q$ be a path fragment from $c$ in $(C,\gamma)$, and assume $\Cyl(q) = \Cyl(q_1) \cup \ldots \cup \Cyl(q_n)$, with $\Cyl(q_i) \cap \Cyl(q_j) = \emptyset$ for $i \ne j$. Since for $i \in \{1,\ldots,n\}$, $\Cyl(q_i) \subseteq \Cyl(q)$, an easy induction on $\depth(q)$ gives $q \in \pref(q_i)$. Moreover, since the cylinder sets $\Cyl(q_i)$ with $i \in \{1,\ldots,n\}$ are pairwise disjoint, we can assume w.l.o.g.~that the path fragments $q_1,\ldots,q_n$ are pairwise incompatible: if this is not the case for $q_i$ and $q_j$, one can always replace each of $q_i$ and $q_j$ with an equivalent set of path fragments of depth uniformly $\max(\depth(q_i),\depth(q_j))$, and subsequently remove the path fragments whose associated cylinder sets have measure $0$. The fact that $\mu_\gamma(\Cyl(q)) = \mu_\gamma(\Cyl(q_1)) + \ldots + \mu_\gamma(\Cyl(q_n))$ then follows by induction on $\max \{\depth(q_i) \mid i \in \{1,\ldots,n\} \,\} - \depth(q)$. In the base case, the path fragments $q_1,\ldots,q_n$ all extend $q$ with at most one step in each leaf of $q$. Then, since the cylinder sets associated to $q_1,\ldots,q_n$ cover $\Cyl(q)$, it follows that for all path fragments that extend $q$ and are not subsumed by one of $q_1,\ldots,q_n$, the associated cylinder sets are empty (and thus have measure $0$). This, together with the definition of $\nu$-extents, gives $\mu_\gamma(\Cyl(q)) = \mu_\gamma(\Cyl(q_1)) + \ldots + \mu_\gamma(\Cyl(q_n))$. The induction step follows by applying the induction hypothesis to the path fragments $q^1, \ldots,q^{\arity(\lambda)}$, where the tree associated to $q$ has root given by $(c,\lambda)$ and immediate sub-trees corresponding to $q^1, \ldots,q^{\arity(\lambda)}$.
\end{proof}
As a result, $\mu_\gamma : \Sigma_c \to S$ extends to a measure on the field $\R_c$ generated by $\Sigma_c$. We can now consider the $S$-valued measure $\mu_\gamma$ on the smallest $\sigma$-algebra containing $\Sigma_c$, given by (the restriction of) the induced outer measure $\mu_\gamma^*$. (Recall that we did not prove a uniqueness result for an $S$-valued measure extending $\mu_\gamma : \Sigma_c \to S$. At this point it is not clear to us whether such a result can be proved, and what additional assumptions on the semiring $S$ this might require.) To avoid confusion with least fixpoint formulas, we will always use the subscript $\gamma$ when referring to the resulting $S$-valued measure.

\begin{exa}
\label{ex:measure}
Using Example~\ref{exa:extent} and Definition~\ref{induced-measure}, we now obtain the following instances of the resulting measure $\mu_\gamma : \M_c \to S$:
\begin{enumerate}
\item For $S = (\{0,1\},\vee,0,\wedge,1)$ and $A \in \M_c$, we have $\mu_\gamma(A) = 1$ if and only if $A \ne \emptyset$.
\item For $S = ([0,1],+,0,*,1)$ and $A \in \M_c$, $\mu_\gamma(A)$ gives the probability of exhibiting a path in $A$ from state $c$. Note that in this case, under the assumption that the sum of the probabilities of outgoing transitions from each state equals $1$, the induced measure coincides with the measure used when interpreting LTL over \emph{probabilistic transition systems}\footnote{These are essentially discrete-time Markov chains, but without the initial distribution.} (see e.g.~\cite[Section~10.3]{Baier}).
\item For $S = (\mathbb N^\infty,\min,\infty, +,0)$ and $A \in \M_c$, $\mu_\gamma(A)$ gives the minimal weight of a path in $A$ from state $c$. 
\end{enumerate}
\end{exa}

The next example shows that the resulting measure $\mu_\gamma : \M_c \to S$ does \emph{not} satisfy
\begin{eqnarray*}
\mu_\gamma(\bigcap\limits_{n \in \omega} A_n) = \inf\limits_{n \in \omega} \mu_\gamma(A_n)
\end{eqnarray*}
whenever $A_n$ with $n \in \omega$ are measurable sets such that $A_0 \supseteq A_1 \supseteq \ldots$.

\begin{exa}
\label{counter-ex}
Let $S = (\{0,1\},\vee,0,\wedge,1)$, and consider the $\T_S \circ (\{a,b\} \times \Id)$-coalgebra depicted below, (where all transitions shown have weight $1$):
\begin{eqnarray*}
\UseComputerModernTips\xymatrix@-1pc{
c_0 \ar[r]^-{a} \ar[d]_-{b} & d_0^0 \ar@(ur,dr)[]^-{b}\\
c_1 \ar[r]^-{a} \ar[d]_-{b} & d_1^0  \ar[r]^-{a} & d_1^1 \ar@(ur,dr)[]^-{b}\\
c_2 \ar[r]^-{a} \ar[d]_-{b} & d_2^0  \ar[r]^-{a} & d_2^1 \ar[r]^-{a} & d_2^2 \ar@(ur,dr)[]^-{b}\\
\vdots}
\end{eqnarray*}
Now consider the following sets of paths from $c_0$:
\begin{eqnarray*}A_0 & = & \{
\UseComputerModernTips\xymatrix@-1pc{c_0 \ar[r]^-{a} & d_0^0 \ar@(ur,dr)[]^-{b}},\,
\UseComputerModernTips\xymatrix@-1pc{c_0 \ar[r]^-{b} & c_1 \ar[r]^-{a} & d_1^0  \ar[r]^-{a} & d_1^1 \ar@(ur,dr)[]^-{b}},\,
\UseComputerModernTips\xymatrix@-1pc{c_0 \ar[r]^-{b} & c_1 \ar[r]^-{b} & c_2 \ar[r]^-{a} & d_2^0  \ar[r]^-{a} & d_2^1 \ar[r]^-{a} & d_2^2 \ar@(ur,dr)[]^-{b}},\,
\ldots
\}\\
A_1 & = & \{
\UseComputerModernTips\xymatrix@-1pc{c_0 \ar[r]^-{b} & c_1 \ar[r]^-{a} & d_1^0  \ar[r]^-{a} & d_1^1 \ar@(ur,dr)[]^-{b}},\,
\UseComputerModernTips\xymatrix@-1pc{c_0 \ar[r]^-{b} & c_1 \ar[r]^-{b} & c_2 \ar[r]^-{a} & d_2^0  \ar[r]^-{a} & d_2^1 \ar[r]^-{a} & d_2^2 \ar@(ur,dr)[]^-{b}},\,
\ldots
\}\\
\ldots
\end{eqnarray*}
with $A_0 \supseteq A_1 \supseteq \ldots$. That is, for $i \in \omega$, $A_i$ consists of those paths from $c_0$ containing at least $i + 1$ $a$-labelled transitions. Then, $\bigcap\limits_{i \in \omega}A_i = \emptyset$ and so $\mu_\gamma(\bigcap\limits_{i \in \omega}A_i) = 0$, however, $\mu_\gamma(A_i) = 1$ for all $i \in \omega$, and therefore $\inf\limits_{i \in \omega} \mu_\gamma(A_i) = 1$. 
\end{exa}

The above equality holds, however, when each $A_n$ is a \emph{finite} union of cylinder sets.

\begin{prop}
\label{prop:inf}
Let $A_0 \supseteq A_1 \supseteq \ldots$, with each $A_n$ ($n \in \omega$) being given by a \emph{finite} union of cylinder sets; that is, $A_n = \bigcup \A_n$ with $\A_n$ consisting of finitely-many cylinder sets, for $n \in \omega$. Then 
\begin{eqnarray*}
\mu_\gamma(\bigcap\limits_{n \in \omega} A_n) = \inf\limits_{n \in \omega} \mu_\gamma(A_n)
\end{eqnarray*}
\end{prop}
\begin{proof}
Clearly, $\mu_\gamma(\bigcap\limits_{n \in \omega}A_n) \sqsubseteq \inf\limits_{n \in \omega} \mu_\gamma (A_n)$:  $\bigcap\limits_{n \in \omega}A_n \subseteq A_n$ gives $\mu_\gamma(\bigcap\limits_{n \in \omega}A_n) \sqsubseteq \mu_\gamma(A_n)$ for all $n \in \omega$, and therefore $\mu_\gamma(\bigcap\limits_{n \in \omega}A_n) \sqsubseteq \inf\limits_{n \in \omega} \mu_\gamma (A_n)$.

To show that $\mu_\gamma(\bigcap\limits_{n \in \omega}A_n) \sqsupseteq \inf\limits_{n \in \omega} \mu_\gamma (A_n)$ also holds, recall from Definition~\ref{def:induced-outer-meas} that $\mu_\gamma(\bigcap\limits_{n \in \omega}A_n) = \inf \{\,\sum\limits_{k \in \omega} \mu(E_k) \mid (E_k \in \R_c)_{k \in \omega} \text{ pairwise disjoint},~ \bigcap\limits_{n \in \omega}A_n \subseteq \bigcup\limits_{k \in \omega} E_k \}$, which, by the definition of $\R_c$, coincides with $\inf \{\,\sum\limits_{k \in \omega} \mu(B_k) \mid (B_k \in \Sigma_c)_{n \in \omega} \text{ pairwise disjoint},~ \bigcap\limits_{n \in \omega}A_n \subseteq \bigcup\limits_{k \in \omega} B_k \}$. We now fix a countable family $(B_k \in \Sigma_c)_{k \in \omega}$ of pairwise disjoint cylinder sets with $\bigcap\limits_{n \in \omega}A_n \subseteq \bigcup\limits_{k \in \omega} B_k$. We can assume w.l.o.g.~that all cylinder sets $B_k$ with $k \in \omega$ have uniform depth, that is, when viewed as trees, all the paths from the root to the leaves of the tree have the same length. (Each cylinder set which does not have uniform depth can be replaced by a \emph{finite} number of pairwise disjoint cylinder sets of uniform depth, with the measure of the union of these cylinder sets being equal to the measure of the original cylinder set.) We can similarly assume that each $\A_n$ consists of pairwise disjoint cylinder sets of uniform depth, and moreover, each cylinder set in $\A_{n+1}$ is a subset of some cylinder set in $\A_n$. (Overlaps between different cylinder sets belonging to some $\A_n$ can be removed by replacing the cylinder sets in question with finite disjoint unions of cylinder sets of larger depth. Also, whenever a cylinder set in $\A_{n+1}$ is not included in a \emph{single} cylinder set in $\A_{n}$, the former can be replaced by a finite number of pairwise disjoint cylinder sets of larger depth, in such a way that the property is established.) We now show that, for each cylinder set $C \in \A_0$, there exists $n_C \in \omega$ such that all cylinder sets in $\A_{n_C}$ which are included in $C$ are also included in $\bigcup\limits_{k \in \omega} B_k$. For, if this was not the case, there would exist a decreasing sequence of cylinder sets $C=C_0 \supseteq C_1 \supseteq \ldots$ with $C_n \in \A_n$ and $C_n \not\subseteq \bigcup\limits_{k \in \omega} B_k$ for $n \in \omega$ -- this follows by applying K\"onig's lemma to the graph whose nodes are cylinder sets belonging to some $\A_n$ and included in $C$, but not included in $\bigcup\limits_{k \in \omega} B_k$, and whose edges are determined by the subset relation between cylinder sets at successive levels. This, in turn, would yield an element $z \in Z$ such that $z \in \bigcap\limits_{n\in \omega} A_n$ (as $z \in C_n \subseteq A_n$ for all $n \in \omega$) but $z \not\in \bigcup\limits_{k \in \omega} B_k$ (as none of the cylinder sets $C_i$ is contained in $\bigcup\limits_{k \in \omega} B_k$). Now by finiteness of $\A_0$, there exists $m \in \omega$ such that $A_m \subseteq \bigcup\limits_{k \in \omega} B_k$ -- simply take $m = \max \{n_C \mid C \in \A_0\}$. We now have:
\begin{eqnarray*}
\inf\limits_{n \in \omega} \mu_\gamma(A_n) \sqsubseteq \mu_\gamma(A_m) \sqsubseteq \sum\limits_{k \in \omega} \mu_\gamma(B_k)
\end{eqnarray*}
As the above holds for \emph{any} choice of cylinder set cover $(B_k)_{k \in \omega}$ for $\bigcap\limits_{n \in \omega} A_n$, we now obtain:
\begin{eqnarray*}
\inf\limits_{n \in \omega} \mu_\gamma(A_n) \sqsubseteq \mu_\gamma(\bigcap\limits_{n \in \omega} A_n)
\end{eqnarray*}
This concludes the proof.
\end{proof}

Before defining the path-based semantics for the logic $\mu\LL_\Lambda$, we prove one more property of the $\sigma$-algebras $\M_c$ and of the associated measures $\mu_\gamma : \M_c \to S$ with $c \in C$. This will be needed to show that the path-based semantics of modal formulas (of the form $\langle \lambda \rangle(\varphi_1,\ldots,\varphi_{\arity(\lambda)})$) agrees with their step-wise semantics.

\begin{prop}
\label{prop:lemma}
Let $\lambda \in \Lambda$ and $c,c_1,\ldots,c_{\arity(\lambda)} \in C$ be such that $\gamma(c)(\iota_\lambda(c_1,\ldots,c_{\arity(\lambda)})) \ne 0$. 
If $A_i \in \M_{c_i}$ for $i \in \{1,\ldots,\arity(\lambda)\}$, then $(c,\lambda)(c_1,\ldots,c_{\arity(\lambda)})[A_1/c_1,\ldots,A_{\arity(\lambda)}/c_{\arity(\lambda)}] \in \M_c$. Moreover, the following holds:
{\small \begin{eqnarray*}
\mu_\gamma((c,\lambda)(c_1,\ldots,c_{\arity(\lambda)})[A_1/c_1,\ldots,A_{\arity(\lambda)}/c_{\arity(\lambda)}]) = \gamma(c)(\iota_\lambda(c_1,\ldots,c_{\arity(\lambda)})) \bullet \mu_\gamma(A_1) \bullet \ldots \bullet \mu_\gamma(A_{\arity(\lambda)})\,.
\end{eqnarray*}}
\end{prop}
\begin{proof}
The first statement is immediate in the case when one of $A_1,\ldots,A_{\arity(\lambda)}$ is the empty set. Assume therefore that this is not the case (and thus also none of $\Paths_{c_1},\ldots,\Paths_{c_{\arity(\lambda)}}$ is the empty set). The statement now follows from $(c,\lambda)(c_1,\ldots,c_{\arity(\lambda)})[A_1/c_1,\ldots,A_{\arity(\lambda)}/c_{\arity(\lambda)} ] = \bigcap\limits_{i \in \{1,\ldots,\arity(\lambda)\}} (c,\lambda)(c_1,\ldots,c_{\arity(\lambda)})[\Paths_{c_1}/c_1,\ldots,A_i/c_i,\ldots,\Paths_{c_{\arity(\lambda)}}/c_{\arity(\lambda)}]$, together with each of the sets of paths in the above intersection belonging to $\M_c$. To see why the latter is true, note that we can use the $\sigma$-algebra $\M_c \subseteq \Pow \Paths_c$ to derive a $\sigma$-algebra $\M'_{c_i} \subseteq \Pow \Paths_{c_i}$ for each $i \in \{1,\ldots,\arity(\lambda)\}$ by taking
\begin{eqnarray*}
\M'_{c_i} = \{A_i \subseteq \Paths_{c_i} \mid (c,\lambda)[\Paths_{c_1}/c_1,\ldots,A_i/c_i,\ldots,\Paths_{c_{\arity(\lambda)}}/c_{\arity(\lambda)}] \in \M_c\}\,.
\end{eqnarray*}
That this yields a $\sigma$-algebra follows from $\M_c$ being a $\sigma$-algebra: (i) $\emptyset \in \M'_{c_i}$ follows from $\emptyset \in \M_c$, (ii) closure of $\M'_{c_i}$ under complements follows from the fact that, for $A_i \in \M'_{c_i}$, both $(c,\lambda)[\Paths_{c_1}/c_1,\ldots,\Paths_{c_{\arity(\lambda)}}/c_{\arity(\lambda)}]$ and $(c,\lambda)[\Paths_{c_1}/c_1,\ldots,A_i/c_i,\ldots,\Paths_{c_{\arity(\lambda)}}/c_{\arity(\lambda)}]$ belong to $\M_c$, and therefore so does their set difference $(c,\lambda)[\Paths_{c_1}/c_1,\ldots,(\Paths_{c_i} \setminus A_i)/c_i,\ldots,\Paths_{c_{\arity(\lambda)}}/c_{\arity(\lambda)}]$, and (iii) closure of $\M'_{c_i}$ under countable unions follows from the closure of $\M_c$ under countable unions of sets of the form used in the definition of $\M'_{c_i}$. Now $\M'_{c_i}$ contains all cylinder sets in $\Sigma_{c_i}$ (as $\M_c$ contains all cylinder sets in $\Sigma_c$), and therefore $\M'_{c_i} \supseteq \M_{c_i}$ (as the latter is the \emph{smallest} such $\sigma$-algebra). Thus, $A_i \in \M_{c_i}$ implies $A_i \in \M'_{c_i}$ implies $(c,\lambda)[\Paths_{c_1}/c_1,\ldots,A_i/c_i,\ldots,\Paths_{c_{\arity(\lambda)}}/c_{\arity(\lambda)}] \in \M_c$ as required.

The proof of the second statement is less straightforward. We reduce proving the equality to proving two inequalities.
\begin{enumerate}
\item For the "$\sqsubseteq$" part, we show that the following holds for any choice of countable, disjoint cylinder set covers $\C_i$ for $A_i$, with $i \in \{1,\ldots,\arity(\lambda)\}$:
{\small 
\begin{eqnarray*}
 \mu_\gamma((c,\lambda)(c_1,\ldots,c_{\arity(\lambda)})[A_1/c_1,\ldots,A_{\arity(\lambda)}/c_{\arity(\lambda)}]) \sqsubseteq \gamma(c)(\iota_\lambda(c_1,\ldots,c_{\arity(\lambda)})) \bullet \mu_\gamma(\C_1) \bullet \ldots \bullet \mu_\gamma(\C_{\arity(\lambda)})\,.
\end{eqnarray*}}
For this, note that the covers $\C_i$ yield a countable, disjoint cylinder set cover $\C$ for $(c,\lambda)(c_1,\ldots,c_{\arity(\lambda)})[A_1/c_1,\ldots,A_{\arity(\lambda)}/c_{\arity(\lambda)}]$, with each $C \in \C$ given by $\Cyl(q)$ where $q$ is a path fragment with root $(c,\lambda)$ and sub-trees $q_1,\ldots,q_{\arity(\lambda)}$, with $\Cyl(q_i) \in \C_i$ for $i \in \{1,\ldots,\arity(\lambda)\}$. It then follows from the definition of $\mu_\gamma$ on cylinder sets (Definition~\ref{induced-measure}) together with the distributivity of $\bullet$ over countable sums (Remark~\ref{rem:distrib-countable-sums})  that $\mu_\gamma(\C) = \gamma(c)(\iota_\lambda(c_1,\ldots,c_{\arity(\lambda)})) \bullet \mu_\gamma(\C_1)  \bullet \ldots \bullet \mu_\gamma(\C_{\arity(\lambda)})$. This gives the above inequality. Moreover, since this inequality holds for an arbitrary choice of $\C_i$s, and since $\bullet$ preserves infima in each argument (see Assumption~\ref{ass:cl}), we obtain
{\small 
\begin{eqnarray*}
 \mu_\gamma((c,\lambda)(c_1,\ldots,c_{\arity(\lambda)})[A_1/c_1,\ldots,A_{\arity(\lambda)}/c_{\arity(\lambda)}]) \sqsubseteq \gamma(c)(\iota_\lambda(c_1,\ldots,c_{\arity(\lambda)})) \bullet \mu_\gamma(A_1) \bullet \ldots \bullet \mu_\gamma(A_{\arity(\lambda)})
\end{eqnarray*}}
as required.
\item For the "$\sqsupseteq$" part, we let $\C$ be a countable, disjoint cylinder set cover for the set $(c,\lambda)(c_1,\ldots,c_{\arity(\lambda)})[A_1/c_1,\ldots,A_{\arity(\lambda)}/c_{\arity(\lambda)}]$. The obvious way to proceed would be to construct cylinder set covers $\C_i$ for $A_i$, such that $\C \supseteq (c,\lambda)(c_1,\ldots,c_{\arity(\lambda)})[\C_1/c_1,\ldots,$\linebreak$\C_{\arity(\lambda)}/c_{\arity(\lambda)}]$. However, this turns out to be difficult, if at all possible. Instead, we will construct a decreasing sequence of covers $\bigcup \C^i_0 \supseteq \bigcup \C^i_1 \supseteq \ldots$ for each $A_i$, with $i \in \{1,\ldots,\arity(\lambda)\}$, additionally satisfying $\C \supseteq (c,\lambda)(c_1,\ldots,c_{\arity(\lambda)})[(\bigcap\limits_{n \in \omega} \bigcup \C^1_n)/c_1,\ldots,$\linebreak$(\bigcap\limits_{n \in \omega} \bigcup \C^{\arity(\lambda)}_n)/c_{\arity(\lambda)}]$, and use this inclusion to show the required inequality.

We assume w.l.o.g.~that $\C$ only contains cylinder sets of uniform depth, listed in increasing depth order, and that $\C$ is \emph{minimal}, in that no proper subset of $\C$ covers $(c,\lambda)(c_1,\ldots,c_{\arity(\lambda)})[A_1/c_1,\ldots,A_{\arity(\lambda)}/c_{\arity(\lambda)}]$. (Any cylinder set in $\C$ with non-uniform depth can be replaced by a \emph{finite} number of pairwise disjoint cylinder sets of uniform depth, with the measure of the union of these cylinder sets being equal to the measure of the original cylinder set. Also, cylinder sets (of uniform depth) that are not needed to cover $(c,\lambda)(c_1,\ldots,c_{\arity(\lambda)})[A_1/c_1,\ldots,A_{\arity(\lambda)}/c_{\arity(\lambda)}]$ can be removed from $\C$ one by one to obtain a minimal cover; each such removal can only decrease the value of $\mu_\gamma(\C)$.) Now assuming $\C$ consists of cylinder sets $B_j = \Cyl((c,\lambda)(q_1^j,\ldots,q_{\arity(\lambda)}^j))$ with $j \in \omega$, we have, for $j,k \in \omega$ with $j < k$ and for $i \in \{1,\ldots,\arity(\lambda)\}$, that either $\Cyl(q^j_i) \cap \Cyl(q^k_i) = \emptyset$ or $\Cyl(q^j_i) \supseteq \Cyl(q^k_i)$. A key observation then is that, if $B_j, B_k \in \C$ are such that $\Cyl(q^j_i) \supseteq \Cyl(q^k_i)$, then since $B_j \cap B_k = \emptyset$, and since each of $B_j, B_k$ is needed to cover the set $(c,\lambda)(c_1,\ldots,c_{\arity(\lambda)})[A_1/c_1,\ldots,A_{\arity(\lambda)}/c_{\arity(\lambda)}]$, the set $A_i$ can only contain a path in $\Cyl(q^j_i) \setminus \Cyl(q^k_i)$ if this is witnessed by some other $B_l$ such that $\Cyl(q^j_i) \supseteq \Cyl(q^l_i)$; in other words, the cylinder set $\Cyl(q^j_i)$ is not really needed to construct a cover of $A_i$, and considering instead $\Cyl(q^k_i)$ and any additional cylinder sets $\Cyl(q^l_i)$ as above is sufficient. Moreover, this applies to \emph{any} pair of cylinder sets $B_j, B_k \in \C$ as above. Then, a decreasing sequence of finite covers of cylinder sets for each $A_i$, with $i \in \{1,\ldots,\arity(\lambda)\}$, is obtained through the following steps:
\begin{enumerate}
\item The countable set of cylinder sets $\{\Cyl(q^j_i) \mid j \in \omega\}$ is partitioned into \emph{finite} collections $\C^i_n = \{ \Cyl(q^j_i) \mid j \in \omega,\,\depth(q^j_i) = n\}$, with $n \in \omega$.
\item The collections $\C^i_n$ with $n \in \omega$ are extended as follows:
\begin{itemize}
\item for $n = 0,1,\ldots$, if $\C^i_{n+1}$ does not contain \emph{any} subset of some cylinder set in $\C^i_n$, then that cylinder set is added to $\C^i_{n+1}$.
\item for $n = 0,1,\ldots$, if $\C^i_n$ does not contain any \emph{superset} of some cylinder set in $\C^i_{n+1}$, then that cylinder set is added to $\C^i_n$.
\end{itemize}
\end{enumerate}
For each $i \in \{1,\ldots,\arity(\lambda)\}$, the resulting collections $\C_i^n$ with $n \in \omega$ are still finite, and such that $\bigcup \C^i_0 \supseteq \bigcup \C^i_1 \supseteq \ldots$. Moreover, this decreasing sequence is such that $\bigcap\limits_{n \in \omega} \bigcup \C^i_n \supseteq A_i$ (because of the key observation above), and satisfies the hypothesis of Proposition~\ref{prop:inf} (which in turn gives $\mu_\gamma(\bigcap\limits_{n \in \omega} \bigcup \C^i_n) = \inf\limits_{n \in \omega} \mu_\gamma(\bigcup \C^i_n)$). Finally, our construction of the sets $\C^i_n$ with $n \in \omega$ and $i \in \{1,\ldots,\arity(\lambda)\}$ gives
\begin{eqnarray*}
(c,\lambda)(c_1,\ldots,c_{\arity(\lambda)})[(\bigcap\limits_{n \in \omega}\bigcup \C^1_n)/c_1,\ldots,(\bigcap\limits_{n \in \omega}\bigcup \C^{\arity(\lambda)}_n)/c_{\arity(\lambda)}] \subseteq \C\,. 
\end{eqnarray*}
We now have:
\begin{align*}
\mu_\gamma(\C) & \sqsupseteq \mu_\gamma((c,\lambda)(c_1,\ldots,c_{\arity(\lambda)})[(\bigcap\limits_{n \in \omega}\bigcup \C^1_n)/c_1,\ldots,(\bigcap\limits_{n \in \omega}\bigcup \C^{\arity(\lambda)}_n)/c_{\arity(\lambda)}])  \tag{monotonicity of $\mu_\gamma$}\\
& = \mu_\gamma(\bigcap\limits_{n \in \omega} (c,\lambda)(c_1,\ldots,c_{\arity(\lambda)})[\bigcup \C^1_n/c_1,\ldots,\bigcup \C^{\arity(\lambda)}_n/c_{\arity(\lambda)}]) \\
& = \inf\limits_{n \in \omega}\mu_\gamma((c,\lambda)(c_1,\ldots,c_{\arity(\lambda)})[\bigcup \C^1_n/c_1,\ldots,\bigcup \C^{\arity(\lambda)}_n/c_{\arity(\lambda)}]) \tag{*}\\
& = \inf\limits_{n \in \omega} \left(\gamma(c)(\iota_\lambda(c_1,\ldots,c_{\arity(\lambda)})) \bullet \mu_\gamma(\bigcup \C^1_n) \bullet \ldots \bullet \mu_\gamma(\bigcup \C^{\arity(\lambda)}_n) \right) \tag{**}\\
& = \gamma(c)(\iota_\lambda(c_1,\ldots,c_{\arity(\lambda)})) \bullet \inf\limits_{n \in \omega} (\mu_\gamma(\bigcup \C^1_n) \bullet \ldots \bullet \mu_\gamma(\bigcup \C^{\arity(\lambda)}_n))\\ \tag{Assumption~\ref{ass:cl}}\\
 & \sqsupseteq \gamma(c)(\iota_\lambda(c_1,\ldots,c_{\arity(\lambda)})) \bullet \inf\limits_{n \in \omega} \mu_\gamma(\bigcup \C^1_n) \bullet \ldots \bullet \inf\limits_{n \in \omega} \mu_\gamma(\bigcup \C^{\arity(\lambda)}_n) \\
 & = \gamma(c)(\iota_\lambda(c_1,\ldots,c_{\arity(\lambda)})) \bullet \mu_\gamma(\bigcap\limits_{n \in \omega}\bigcup \C^1_n) \bullet \ldots \bullet \mu_\gamma(\bigcap\limits_{n \in \omega}\bigcup \C^{\arity(\lambda)}_n)\\ \tag{Proposition~\ref{prop:inf}}\\
 &  \sqsupseteq \gamma(c)(\iota_\lambda(c_1,\ldots,c_{\arity(\lambda)})) \bullet  \mu_\gamma(A_1) \bullet \ldots \bullet \mu_\gamma(A_{\arity(\lambda)}) \tag{monotonicity of $\mu_\gamma$}
\end{align*}
In the above, (*) follows from the application of Proposition~\ref{prop:inf}, after noting that the sets $(c,\lambda)(c_1,\ldots,c_{\arity(\lambda)})[\bigcup \C^1_n/c_1,\ldots,\bigcup \C^{\arity(\lambda)}_n/c_{\arity(\lambda)}]$ with $n \in \omega$ satisfy the hypothesis of this result (by finiteness of each $\C_n^i$). Also, (**) follows from the definition of $\mu_\gamma$ on finite disjoint unions of cylinder sets and the distributivity of $\bullet$ over finite sums, after noting that each of $\C^i_n$ as well as $(c,\lambda)(c_1,\ldots,c_{\arity(\lambda)})[\bigcup \C^1_n/c_1,\ldots,\bigcup \C^{\arity(\lambda)}_n/c_{\arity(\lambda)}]$ can be written as finite disjoint unions of cylinder sets. Now as the above sequence of (in)equalities holds for any disjoint cylinder set cover $\C$ for $(c,\lambda)(c_1,\ldots,c_{\arity(\lambda)})[A_1/c_1,\ldots,A_{\arity(\lambda)}/c_{\arity(\lambda)}]$, we obtain
{\small 
\begin{eqnarray*}
\qquad \mu_\gamma((c,\lambda)(c_1,\ldots,c_{\arity(\lambda)})[A_1/c_1,\ldots,A_{\arity(\lambda)}/c_{\arity(\lambda)}]) \sqsupseteq \gamma(c)(\iota_\lambda(c_1,\ldots,c_{\arity(\lambda)})) \bullet \mu_\gamma(A_1) \bullet \ldots \bullet \mu_\gamma(A_{\arity(\lambda)})
\end{eqnarray*}}
as required. This concludes the proof of the "$\sqsupseteq$" part and the proof of the proposition.
\end{enumerate}
\end{proof}

We now observe that $(\Paths_C,\pi_2 \circ \zeta_C)$ is an $F$-coalgebra, and therefore there is a standard notion of satisfaction of formulas in $\mu\LL_\Lambda$ by states of this coalgebra, i.e.~by paths in $C$ (see Section~\ref{coalg-fixpoint-logics}). In particular, under this standard semantics, any state of an $F$-coalgebra satisfies the formula $\top$. We are finally ready to define a path-based semantics for the logic $\mu\LL_\Lambda$. 

\begin{defi}[Path-based semantics for $\mu\LL_\Lambda$]
\label{path-based-semantics}
Let $(C,\gamma)$ be a $\T_S \circ F$-coalgebra. The \emph{path-based semantics} $\llsem \phi \rrsem_\gamma \in S^C$ of a formula $\phi \in \mu\LL_\Lambda$ in $(C,\gamma)$ is given by $\llsem \phi \rrsem_\gamma(c) = \mu_\gamma(\Paths_c(\phi))$ for $c \in C$, where $\Paths_c(\phi) = \{ p \in \Paths_c \mid p \in \lsem \phi \rsem_{\pi_2 \circ \zeta_C}\}$.
\end{defi}

\begin{rem}
\label{rem-existing-semantics}
As expected, when $S = (\{0,1\}, \vee, 0, \wedge,1)$, the measure of Definition~\ref{induced-measure} associates a non-zero value to a measurable set of paths precisely when the set in question is non-empty. When $S = ([0,1],+,0,*,1)$, our semantics is similar to the probabilistic semantics of LTL, which also involves defining a $\sigma$-algebra structure on the set of maximal paths from a given state \cite[Chapter~10]{Baier}. However, differently from probabilistic systems, the definition of the $S$-valued measure induced by a $\T_S \circ F$-coalgebra has a coinductive flavour, given our use of $\nu$-extents to measure cylinder sets. When the sum of probabilities of outgoing transitions from each state in a $\T_S \circ F$-coalgebra is $1$, our definition of the induced measure is similar to the standard one, which only takes into account the probabilities of transitions that match a given path fragment $q$, as in this case the $\nu$-extent of each state is equal to $1$. However, when the above condition is not satisfied, or when $S$ is a \emph{different} (e.g.~the tropical) semiring, our definition takes into account not just the weights of transitions that match a path fragment $q$, but also the future linear-time behaviour ($\nu$-extent) of the states annotating the leaves of $q$. This is natural, given our emphasis on \emph{maximal} (rather than finite) traces.
\end{rem}

Definition~\ref{path-based-semantics} thus generalises existing path-based semantics for non-deterministic and probabilistic systems. The remainder of this section shows the correctness of this definition, by proving that the set $\Paths_c(\phi)$ is indeed measurable for each $\phi \in \mu\LL_\Lambda$ and $c \in C$. To this end, we associate sets of paths not just to formulas in $\mu\LL_\Lambda$, but also to formulas in $\mu\LL_\Lambda^{\V}$ with $\V$ a set of variables and $V : \V \to \{0,1\}^{\Paths_C}$ a valuation. Then we show that, under certain assumptions on $V$, those sets of paths are measurable. This, in turn, yields measurability of $\Paths_c(\phi)$. Our proof makes use of the concept of \emph{fixpoint nesting depth}.

\begin{defi}
The \emph{fixpoint nesting depth} of a formula $\phi \in \mu\LL_\Lambda^{\V}$, denoted $\fnd(\phi)$, is defined by induction on the structure of $\phi$:
\begin{itemize}
\item If $\phi = \bot$ or $\phi = \top$ or $\phi = x$ with $x \in \V$, then $\fnd(\phi) = 0$.
\item If $\phi = \langle \lambda \rangle(\phi_1,\ldots,\phi_{\arity(\lambda)})$, then $\fnd(\phi) = \max \{\fnd(\phi_i) \mid i \in I\}$. The definition in the case of modalities incorporating guarded disjunctions is similar.
\item If $\phi = \eta x.\phi_x$, then $\fnd(\phi) = \fnd(\phi_x)+1$.
\end{itemize}
\end{defi}
We now fix a $\T_S \circ F$-coalgebra $(C,\gamma)$. For $\phi \in \mu\LL_\Lambda^{\V}$, $c \in C$ and $V : \V \to \{0,1\}^{\Paths_C}$, we let $\Paths_c^{V}(\phi)  := \{ p \in \Paths_c \mid p \in \lsem \phi \rsem_{\pi_2 \circ \zeta_C}^V \}$.
For conciseness, in what follows we write $\zeta'_C$ for $\pi_2 \circ \zeta_C : \Paths_C \to F \Paths_C$.

\begin{prop}
\label{prop-meas}
Let $(C,\gamma)$ be a $\T_S \circ F$-coalgebra, and let $\phi \in \mu\LL_\Lambda^{\V}$. If $V : \V \to \{0,1\}^{\Paths_C}$ is such that the set $\Paths_c^V(x) = \{p \in \Paths_C \mid V(x)(p)=1 \}$ is measurable for all $c \in C$ and $x \in \V$, then the set $\Paths_c^V(\phi)$ is measurable for all $c \in C$.
\end{prop}
\begin{proof}
Induction on $\fnd(\phi)$.
\begin{enumerate}
\item \label{case1} $\fnd(\phi) = 0$. We prove the statement by structural induction on $\phi$:
\begin{enumerate}
\item $\phi = \bot$. In this case $\Paths_c^V(\phi) = \emptyset \in \M_c$.
\item $\phi = \top$. In this case $\Paths_c^V(\phi) = \Paths_c \in \M_c$.
\item $\phi = x \in \V$. In this case the statement follows from the assumption on $V$.
\item $\phi = \langle \lambda \rangle(\phi_1,\ldots,\phi_{\arity(\lambda)})$ with $\Paths_{c_i}^V(\phi_i)$ measurable for each $c_i \in C$ and each $i \in \{1,\ldots,\arity(\lambda)\}$. Then, $\Paths_c^V(\langle \lambda \rangle(\phi_1,\ldots,\phi_{\arity(\lambda)}))$ is given by the finite union
\begin{eqnarray*}
\bigcup\limits_{\gamma(c)(\iota_{\lambda}(c_1,\ldots,c_{\arity(\lambda)})) \ne 0} (c,\lambda)(c_1,\ldots,c_{\arity(\lambda)})[\Paths_{c_1}^V(\phi_1)/c_1,\ldots,\Paths^V_{c_{\arity(\lambda)}}(\phi_{\arity(\lambda)})/c_{\arity(\lambda)} ]
\end{eqnarray*}
That each of $(c,\lambda)(c_1,\ldots,c_{\arity(\lambda)})[\Paths_{c_1}^V(\phi_1)/c_1,\ldots,\Paths^V_{c_{\arity(\lambda)}}(\phi_{\arity(\lambda)})/c_{\arity(\lambda)} ]$ is measurable follows from the induction hypothesis by Proposition~\ref{prop:lemma}.
\end{enumerate}
\item \label{case2} $\fnd(\phi) > 0$. As for case (\ref{case1}), we prove the statement by structural induction on $\phi$, with four similar sub-cases treated in exactly the same way, but with two additional cases, considered below:
\begin{enumerate}
\item[(e)] $\phi = \mu x.\phi_x$. In this case, it follows from the continuity of the operator used to define $\lsem \mu x .\phi_x \rsem_{\zeta'_C}^V$ (see Remark~\ref{rem:poly}) that
\begin{eqnarray}
\label{equality1}
\Paths_c^V(\phi) = \bigcup\limits_{i \in \omega} \Paths_c^{V \cup \{x \mapsto \lsem \phi_x^{(i)} \rsem^V_{\zeta'_C}\}}(\phi_x)
\end{eqnarray}
where $\phi_x^{(0)} = \bot$ and $\phi_x^{(i+1)} = \phi_x[\phi_x^{(i)}/x]$ for $i \in \omega$\footnote{In the above, we assume that $x \not\in \V$. This can easily be guaranteed by performing some variable renamings whenever the formula in question is not clean.}. Note that, at this point, we cannot immediately use the induction hypothesis, since in general $\fnd(\phi_x^{(i)}) \not < \fnd(\phi)$ -- take, for instance, $\phi := \mu x.\nu y.(\langle a \rangle(x) \sqcup \langle b \rangle y)$, for which $\fnd(\phi_x^{(2)}) = 2$. We proceed in the following way: we prove by induction on $i \in \omega$ that $\lsem \phi_x^{(i)} \rsem^V_{\zeta'_C}$ is measurable and $\Paths_c^{V \cup \{x \mapsto \lsem \phi_x^{(i)} \rsem^V_{\zeta'_C}\}}(\phi_x)$ is measurable.
\begin{enumerate}
\item $i = 0$. In this case, $\phi_x^{(i)} = \bot$ and hence $\lsem \phi_x^{(i)} \rsem^V_{\zeta'_C} = \emptyset$ is measurable. Also, $\Paths_c^{V \cup \{x \mapsto \lsem \phi_x^{(i)} \rsem^V_{\zeta'_C}\}}(x) = \Paths_c^{V \cup \{x \mapsto \lambda z.0\}}(x) = \emptyset$ is measurable. Then, since $\fnd(\phi_x) = \fnd(\phi) -1$, and since $\Paths_c^{V  \cup \{x \mapsto \lsem \phi_x^{(i)} \rsem^V_{\zeta'_C}\}}(y)$ is measurable (as it coincides with $\Paths_c^V(y)$) for $y \in \V$, it follows by the induction hypothesis that $\Paths_c^{V \cup \{x \mapsto \lsem \phi_x^{(i)} \rsem^V_{\zeta'_C}\}}(\phi_x)$ is also measurable.
\item Assume that $\lsem \phi_x^{(i)} \rsem^V_{\zeta'_C}$ is measurable, and $\Paths_c^{V \cup \{x \mapsto \lsem \phi_x^{(i)} \rsem^V_{\zeta'_C}\}}(\phi_x)$ is measurable. As $\phi_x^{(i+1)} = \phi_x[\phi_x^{(i)}/x]$, these assumptions yield $\lsem \phi_x^{(i+1)} \rsem^V_{\zeta'_C}$ measurable, and therefore $\Paths_c^{V \cup \{x \mapsto \lsem \phi_x^{(i+1)} \rsem^V_{\zeta'_C}\}}(x)$ is measurable (since it coincides with $\Paths_c^{V \cup \{x \mapsto \lsem \phi_x^{(i)} \rsem^V_{\zeta'_C}\}}(\phi_x)$). Then, since $\fnd(\phi_x) = \fnd(\phi) -1$, it follows by the induction hypothesis as before that $\Paths_c^{V \cup \{x \mapsto \lsem \phi_x^{(i+1)} \rsem^V_{\zeta'_C}\}}(\phi_x)$ is itself measurable.
\end{enumerate}
Now since each of $\Paths_c^{V \cup \{x \mapsto \lsem \phi_x^{(i)} \rsem^V_{\zeta'_C}\}}(\phi_x)$, with $i \in \omega$, is measurable, and moreover, $\Paths_c^{V \cup \{x \mapsto \lsem \phi_x^{(i)} \rsem^V_{\zeta'_C}\}}(\phi_x) \subseteq \Paths_c^{V \cup \{x \mapsto \lsem \phi_x^{(i+1)} \rsem^V_{\zeta'_C}\}}(\phi_x)$ for $i \in \omega$, it follows that $\Paths_c^V(\phi) = \bigcup\limits_{i \in \omega} \Paths_c^{V \cup \{x \mapsto \lsem \phi_x^{(i)} \rsem^V_{\zeta'_C}\}}(\phi_x)$ is also measurable.
\item[(f)] $\phi = \nu x.\phi_x$. In this case, the co-continuity of the operator used to define $\lsem \nu x .\phi_x \rsem_{\zeta'_C}^V$ (see Remark~\ref{rem:poly}) gives
\begin{eqnarray}
\label{equality2}
\Paths_c^V(\phi) = \bigcap\limits_{i \in \omega} \Paths_c^{V \cup \{x \mapsto \lsem \phi_x^{(i)} \rsem^V_{\zeta'_C}\}}(\phi_x)
\end{eqnarray}
where $\phi_x^{(0)} = \top$ and $\phi_x^{(i+1)} = \phi_x[\phi_x^{(i)}/x]$ for $i \in \omega$. The conclusion now follows similarly to the previous case. \qedhere
\end{enumerate}
\end{enumerate}
\end{proof}
Clearly, for $\phi \in \mu\LL_\Lambda$ and $c \in C$, we have $\Paths_c(\phi) = \Paths_c^\emptyset(\phi)$, where $\emptyset : \emptyset \to \{0,1\}^{\Paths_C}$ is the unique (empty) valuation. Since by Proposition~\ref{prop-meas}, $\Paths_c^\emptyset(\phi)$ is measurable, so is $\Paths_c(\phi)$.

\section{Equivalence of the Two Semantics}
\label{equiv-sem}

This section shows that the step-wise semantics and the path-based semantics for $\mu\LL_\Lambda$ coincide.

\begin{thm}
\label{thm-equiv}
The step-wise semantics for $\mu\LL_\Lambda$ (Definition~\ref{step-wise-semantics}) and the path-based semantics for $\mu\LL_\Lambda$ (Definition~\ref{path-based-semantics}) coincide. That is, given a $\T_S \circ F$-coalgebra $(C,\gamma)$ and $\phi \in \mu\LL_\Lambda$, $\lsem \phi \rsem_\gamma = \llsem \phi \rrsem_\gamma$.
\end{thm}
\begin{proof}
We will prove $\lsem \phi \rsem_\gamma \sqsubseteq \llsem \phi \rrsem_\gamma$ and $\lsem \phi \rsem_\gamma \sqsupseteq \llsem \phi \rrsem_\gamma$.

\begin{itemize}
\item The proof of $\lsem \phi \rsem_\gamma \sqsubseteq \llsem \phi \rrsem_\gamma$ exploits the automata-theoretic characterisation of $\lsem \phi \rsem_\gamma$ as the extent of the product (see Definition~\ref{simple-prod-aut}) of $(C,\gamma)$ and $(\Cl(\phi),\beta,\Omega)$ (Theorem~\ref{thm:extent-char}), where $(\Cl(\phi),\beta,\Omega)$ is the parity $\T_S \circ F$-automaton induced by the formula $\phi$. This result assumes that $\phi$ is (clean and) strictly guarded, but as discussed in Section~\ref{fixpoint-logic}, any formula $\phi \in \mu\LL_\Lambda$ is equivalent to a (clean and) strictly guarded one. Now given $\phi \in \mu\LL_\Lambda$ and $c \in C$, the part of the product automaton relevant to the computation of the extent $\extent_{\gamma \otimes \beta}(c,\phi)$ consists of those states reachable from $(c,\phi)$. Moreover, the automaton $(\Cl(\phi),\beta,\Omega)$ is \emph{deterministic} (there is at most one transition labelled by any $\lambda \in \Lambda$ from any state). As a result, for a reachable state $(d,\psi)$ of the product automaton, the coalgebra structure is inherited from that of $d$ in $(C,\gamma)$ (with only transitions matching those of the formula automaton being considered), and the parity is given by $\Omega(\psi)$. An immediate consequence of this is that paths from $c$ in $(C,\gamma)$ satisfying $\phi$ are in one-to-one correspondence with paths from $(c,\phi)$ in $(C \times \Cl(\phi),\gamma \otimes \beta)$ satisfying $\phi$, with the associated transition weights also agreeing. As a result, any countable, disjoint cylinder set cover for $\Paths_{(c,\phi)}(\phi)$ in $(C \times \Cl(\phi),\gamma \otimes \beta)$ yields a countable, disjoint cylinder set cover for $\Paths_c(\phi)$ in $(C,\gamma)$ with the same measure, and conversely. This finally gives $\mu_\gamma(\Paths_c(\phi)) = \mu_{\gamma \otimes \beta}(\Paths_{(c,\phi)}(\phi))$. Using this observation together with Theorem~\ref{thm:extent-char} and the definition of $\llsem \phi \rrsem_\gamma$, we can reduce proving $\lsem \phi \rsem_\gamma \sqsubseteq \llsem \phi \rrsem_\gamma$ to proving that the following holds
\begin{eqnarray}
\label{eq:goal}
\extent_{\gamma \otimes \beta}(c,\phi) \sqsubseteq \mu_{\gamma \otimes \beta}(\C)
\end{eqnarray}
for any countable, disjoint cylinder set cover $\C$ for $\Paths_{(c,\phi)}(\phi)$ in $(C \times \Cl(\phi),\gamma \otimes \beta)$. The proof of (\ref{eq:goal}) makes use of the \emph{unfolding} of the automaton $(C \times \Cl(\phi),\gamma \otimes \beta,\Omega)$ with initial state $(c,\phi)$ into an infinite tree. For a pointed, parity $\T_S \circ F$-automaton $(D,\delta,\Omega,d_0)$, its \emph{unfolding} is the pointed, parity $\T_S \circ F$-automaton $(D',\delta',\Omega',d'_0)$ where $D'$ contains a copy $d'_0$ of the initial state $d_0$, and for each copy $d' \in D'$ of some $d \in D$ and each transition $\UseComputerModernTips\xymatrix@-1pc{d \ar[r]^-{w,i} & (d_1,\ldots,d_{j_i})}$ in $(D,\delta)$, $(D',\delta')$ contains (new) copies $d_1',\ldots,d_{j_i}'$ of $d_1,\ldots,d_{j_i}$ and a transition $\UseComputerModernTips\xymatrix@-1pc{d' \ar[r]^-{w,i} & (d_1',\ldots,d_{j_i}')}$. The states of $D'$ inherit their parities from the corresponding states in $D$: if $d' \in D'$ is a copy of $d \in D$, then $\Omega'(d') = \Omega(d)$.

We now write $(D,\delta,\Omega,d_0)$ for the unfolding of $(C \times \Cl(\phi),\gamma \otimes \beta,\Omega)$ at $(c,\phi)$, and fix a countable, disjoint cylinder set cover $\C$ for the set of paths $\Paths_{(c,\phi)}(\phi)$ in $(C \times \Cl(\phi),\gamma \otimes \beta)$. Using $\C$ and the automaton $(D,\delta,\Omega,d_0)$, we construct a countable collection of parity $\T_S \circ F$-automata: (i) one for each cylinder set $C \in \C$, denoted $(D,\delta,\Omega,d_0)\!\!\restriction\!C$, whose paths are those paths in $(D,\delta,\Omega,d_0)$ which are covered by $C$, (ii) and another one, denoted $(D,\delta,\Omega,d_0) \setminus \C$, whose paths are those paths in $(D,\delta,\Omega,d_0)$ not covered by any $C \in \C$. Concretely, for $C \in \C$, $(D,\delta,\Omega,d_0)\!\!\restriction\!C = (D^C,\delta^C,\Omega,d_0^C)$ is the sub-tree of $(D,\delta,d_0)$ which "matches" $C$; that is, assuming $C = \Cyl(q)$, $(D^C,\delta^C,\Omega,d_0^C)$ is obtained from $(D,\delta,\Omega,d_0)$ by removing, starting from the root and continuing for a number of steps equal to the depth of the path fragment $q$, transitions not covered by $q$ along with the entire sub-trees having those transitions as initial transitions. Also, $(D,\delta,\Omega,d_0) \setminus \C = (D^0,\delta^0,\Omega,d_0^0)$ retains precisely those transitions of $(D,\delta,\Omega,d_0)$ which can be extended to a path not covered by $\C$; note that paths not covered by $\C$ cannot satisfy $\phi$, as $\C$ covers all of $\Paths_{(c,\phi)}(\phi)$. Both $(D,\delta,\Omega,d_0)\!\!\restriction\!C$ and $(D,\delta,\Omega,d_0) \setminus \C$) inherit parities from $(D,\delta,\Omega,d_0)$. In other words, we have "separated" the unfolding $(D,\delta,\Omega,d_0)$ of $(C \times \Cl(\phi),\gamma \otimes \beta,\Omega)$ into a countable number of $\T_S \circ F$-automata, based on the given cover $\C$ for $\Paths_{(c,\phi)}(\phi)$.

The inequality (\ref{eq:goal}) now follows from:
\begin{eqnarray*}
\extent_{\gamma \otimes \beta}(c,\phi) ~=~ 
\extent_\delta(c,\phi) ~\sqsubseteq~
\sum\limits_{C \in \C} \mu_{\delta^C}(C) ~=~
\sum\limits_{C \in \C} \mu_{\gamma \otimes \beta}(C)
\end{eqnarray*}
where:
\begin{itemize}
\item The first equality follows from the extent being preserved by the unfolding. This is because extents are preserved by $\T_S \circ F$-coalgebra homomorphisms (an easy induction on the number of parities proves this), and in particular by the homomorphism sending each copy of a state to the original state.
\item The inequality above is a consequence of:
\begin{itemize}
\item $\extent_\delta(c,\phi) \sqsubseteq \extent_{\delta^0}((c,\phi)^0) + \sum\limits_{C \in \C} \extent^\nu_{\delta^C}((c,\phi)^C)$, which also follows easily by induction on the number of parities, using the definitions of extent and $\nu$-extent, and the "splitting" of $(D,\delta,d_0)$ into $(D^0,\delta^0,d_0^0)$ and $(D^C,\delta^C,d_0^C)$ with $C \in \C$,
\item $\extent^\nu_{\delta^C}((c,\phi)^C) =  \mu_{\delta^C}(C)$ for each $C \in \C$, which follows from the definition of $\mu_{\delta^C}$ on cylinder sets, using the fact that $C$ covers \emph{all} paths from $(c,\phi)^C$ in $(D^C,\delta^C,\Omega,d_0^C)$ and nothing else), 
\item the fact that the extent of an automaton with no \emph{accepting paths}\footnote{Here, an \emph{accepting path} is one which satisfies the parity condition; that is, all branches of the tree representing that path are such that the largest parity occurring infinitely often along the branch is even.} (in this case $(D^0,\delta^0,\Omega,d_0^0)$) is $0$. To see why this is true note that, whenever the extent of an automaton state is not $0$, one can construct an accepting path from that state by using the definition of extent, in particular, the relationship between the extent of a state and the extents of its successors. The proof of this is again by induction on the number of parities, and the only slight difficulty is proving that the constructed path is accepting; for this, one needs to record the tuples of ordinals (one ordinal per parity) at which the extent values of the states along the path are obtained, as per Theorem~\ref{thm:cousot}. The reason why $(D^0,\delta^0,\Omega,d_0^0)$ has no accepting paths is that any accepting path would satisfy $\phi$ (given the construction of $(D^0,\delta^0,\Omega,d_0^0)$ from $(C\times \Cl(\phi),\gamma \otimes \beta,\Omega)$), and would therefore be covered by $\C$.
\end{itemize}
\item The last equality is an immediate consequence of the definition of $(D^C,\delta^C,\Omega,d_0^C)$.
\end{itemize}
This concludes the proof of the inequality $\lsem \phi \rsem_\gamma \sqsubseteq \llsem \phi \rrsem_\gamma$.

\item We now prove the converse inequality. For $\phi \in \mu\LL_\Lambda^{\V}$ and $V : \V \to \{0,1\}^{\Paths_C}$ such that $\Paths_c^V(x)$ is measurable for each $c \in C$ and $x \in \V$, we let $\llsem \phi \rrsem_\gamma^V : C \to S$ be given by $\llsem \phi \rrsem_\gamma^V(c) = \mu_\gamma(\Paths_c^V(\phi))$ for $c \in C$. (Recall that, by Proposition~\ref{prop-meas}, $\Paths_c^V(\phi)$ is measurable.) We will prove a more general statement, namely that $\lsem \phi \rsem_\gamma^{\tilde V} \sqsupseteq \llsem \phi \rrsem_\gamma^V$ for each $\phi \in \mu\LL_\Lambda^{\V}$ and each $V : \V \to \{0,1\}^{\Paths_C}$ such that $\Paths_c^V(x)$ is measurable for each $c \in C$ and $x \in \V$, where $\tilde V : \V \to S^C$ is given by $\tilde V(x)(c) = \mu_\gamma(\Paths_c^V(x))$. For this, we again use induction on the fixpoint nesting depth of $\phi$.
\begin{enumerate}
\item \label{case1thm} $\fnd(\phi) = 0$. We prove the statement by structural induction on $\phi$:
\begin{enumerate}
\item $\phi = \bot$. Then:
\begin{align*}
\llsem \bot \rrsem_\gamma^V (c) & = \mu_\gamma(\Paths_c^V(\bot)) \tag{definition of $\llsem \_\rrsem_\gamma^V$}\\
& = \mu_\gamma(\emptyset) \tag{$\Paths_c^V(\bot) = \emptyset$}\\
& = 0 \tag{$\mu_\gamma$ is an $S$-valued measure}\\
& = \lsem \bot \rsem_\gamma^{\tilde V} (c) \tag{step-wise semantics of $\bot$}\\
\end{align*}
\pagebreak
\item $\phi = \top$. Then:
\begin{align*}
\llsem \top \rrsem_\gamma^V (c) & = \mu_\gamma(\Paths_c^V(\top)) \tag{definition of $\llsem \_\rrsem_\gamma^V$}\\
& = \mu_\gamma(\Paths_c) \tag{$\Paths_c^V(\top) = \Paths_c$}\\
& = \extent^\nu_\gamma(c) \tag{definition of $\mu_\gamma$ on the largest cylinder set}\\
& = \lsem \top \rsem_\gamma^{\tilde V} (c) \tag{step-wise semantics of $\top$}\\
\end{align*}
\item $\phi = x \in \V$. In this case we have:
\begin{align*}
\llsem x \rrsem_\gamma^V (c) & = \mu_\gamma(\Paths_c^V(x)) \tag{definition of $\llsem \_\rrsem_\gamma^V$}\\
& = \tilde V(x)(c) \tag{definition of $\tilde V$}\\
& = \lsem x \rsem_\gamma^{\tilde V} (c) \tag{step-wise semantics of $x$}\\
\end{align*}
\item $\phi = \langle \lambda \rangle(\phi_1,\ldots, \phi_{\arity(\lambda)})$. In this case we have:
\begin{align*}
& \phantom{{}={}} \llsem\langle \lambda \rangle(\phi_1,\ldots, \phi_{\arity(\lambda)})\rrsem_\gamma^V (c) \\
& = \mu_\gamma(\Paths_c^V(\langle \lambda \rangle(\phi_1,\ldots, \phi_{\arity(\lambda)}))) \tag{definition of $\llsem \_\rrsem_\gamma^V$} \\
& = \mu_\gamma(\bigcup\limits_{\gamma(c)(\iota_{\lambda}(c_1,\ldots,c_{\arity(\lambda)}))  \ne 0} (c,\lambda)(c_1,\ldots,c_{\arity(\lambda)})[\Paths_{c_i}^V(\phi_i)/c_i]) \tag{definition of $\Paths_c^V(\langle \lambda \rangle(\phi_1,\ldots, \phi_{\arity(\lambda)}))$} \\
& = \sum_{\gamma(c)(\iota_\lambda(c_1,\ldots,c_{\arity(\lambda)})) \ne 0} \mu_\gamma((c,\lambda)(c_1,\ldots,c_{\arity(\lambda)})[\Paths_{c_i}^V(\phi_i)/c_i])) \tag{pairwise disjointness of the sets $(c,\lambda)(c_1,\ldots,c_{\arity(\lambda)})[\Paths_{c_i}^V(\phi_i)/c_i]$} \\
& = \sum_{\gamma(c)(\iota_\lambda(c_1,\ldots,c_{\arity(\lambda)})) \ne 0} \gamma(c)(\iota_{\lambda}(c_1,\ldots,c_{\arity(\lambda)})) \bullet (\bullet_{i \in \{1,\ldots \arity(\lambda)\}} \mu_\gamma(\Paths_{c_i}^V(\phi_i))) \tag{Proposition~\ref{prop:lemma}} \\
& = \sum_{\gamma(c)(\iota_\lambda(c_1,\ldots,c_{\arity(\lambda)})) \ne 0} \gamma(c)(\iota_{\lambda}(c_1,\ldots,c_{\arity(\lambda)})) \bullet \llsem \varphi_1 \rrsem_\gamma^{V} \bullet \ldots \bullet \llsem \varphi_{\arity(\lambda)} \rrsem_\gamma^{V} \\
& \sqsubseteq \sum_{\gamma(c)(\iota_\lambda(c_1,\ldots,c_{\arity(\lambda)})) \ne 0} \gamma(c)(\iota_{\lambda}(c_1,\ldots,c_{\arity(\lambda)})) \bullet \lsem \varphi_1 \rsem_\gamma^{\tilde V} \bullet \ldots \bullet \lsem \varphi_{\arity(\lambda)} \rsem_\gamma^{\tilde V} \tag{induction hypothesis, monotonicity of $\bullet$ and $+$ in each argument} \\
& = \gamma^*(\E_{\T_S}(\lsem \lambda \rsem_C (\lsem \varphi_1 \rsem_\gamma^{\tilde V},\ldots,\lsem \varphi_{\arity(\lambda)} \rsem_\gamma^{\tilde V})))(c) \tag{definition of $\E_{\T_S}$ and $\lsem \lambda \rsem$} \\
& = \lsem \langle \lambda \rangle(\phi_1,\ldots, \phi_{\arity(\lambda)}) \rsem_\gamma^{\tilde V}(c) \tag{step-wise semantics of $\langle \lambda \rangle(\phi_1,\ldots, \phi_{\arity(\lambda)})$} \\
\end{align*}
\end{enumerate}
\item $\fnd(\phi) > 0$. Again, we prove the statement by structural induction on $\phi$, with similar sub-cases as in case (\ref{case1thm}), and with two additional cases:
\begin{enumerate}
\item[(e)] $\phi = \mu x .\phi_x \in \mu\LL_\Lambda^\V$. By the (outer) induction hypothesis, we have $\llsem \psi \rrsem_\gamma^V \sqsubseteq \lsem \psi \rsem_\gamma^{\tilde V}$ for each $\psi$ with $\fnd(\psi) < \fnd(\phi)$ and each $V : \V \cup \{x\} \to \{0,1\}^{\Paths_C}$ such that the set $\Paths_c^V(y)$ is measurable for each $c \in C$ and $y \in \V \cup \{x\}$, with $\tilde V : \V \cup \{x\} \to S^C$ given by $\tilde V(y)(c) = \mu_\gamma(\Paths_c^V(y))$. Then, for $V : \V \to \{0,1\}^{\Paths_C}$ such that $\Paths_c^V(y)$ is measurable for each $c \in C$ and $y \in \V$, and $\tilde V : \V \to S^C$ given by $\tilde V(y)(c) = \mu_\gamma(\Paths_c^V(y))$ for $y \in \V$ and $c \in C$, we have:
\begin{align*}
\llsem \mu x .\phi_x \rrsem_\gamma^V(c) &= \mu_\gamma(\Paths_c^V(\mu x.\phi_x)) \tag{definition of $\llsem \_\rrsem_\gamma^V$} \\
& = \mu_\gamma(\bigcup\limits_{i \in \omega} \Paths_c^{V \cup \{x \mapsto \lsem \phi_x^{(i)} \rsem^V_{\zeta'_C}\}}(\phi_x)) \tag{continuity of the operator used to define $\lsem \mu x .\phi_x \rsem_{\zeta'_C}^V$, see Remark~\ref{rem:poly}} \\
& = \sup_{i \in \omega} \mu_\gamma(\Paths_c^{V \cup \{x \mapsto \lsem \phi_x^{(i)} \rsem^V_{\zeta'_C}\}}(\phi_x) ) \tag{Proposition~\ref{prop:countable-unions}} \\
&\sqsubseteq \sup_{i \in \omega} \lsem \phi_x \rsem_\gamma^{\tilde V \cup \{x \mapsto \lambda c. \mu_\gamma(\Paths_c^{V \cup \{x \mapsto \lsem \phi_x^{(i)} \rsem^V_{\zeta'_C}\}}(x)) \}} (c) \tag{induction hypothesis, using $\fnd(\phi_x) < \fnd(\phi)$} \\
& = \sup_{i \in \omega} \lsem \phi_x \rsem_\gamma^{\tilde V \cup \{x \mapsto \lambda c. \mu_\gamma(\Paths_c^V(\phi_x^{(i)}))\}}(c) \\
\tag{$\Paths_c^{V \cup \{x \mapsto \lsem \phi_x^{(i)} \rsem^V_{\zeta'_C}\}}(x) = \Paths_c^V(\phi_x^{(i)})$} \\
& \sqsubseteq \sup_{i \in \omega} \lsem \phi_x \rsem_\gamma^{\tilde V \cup \{x \mapsto \lsem \phi_x^{(i)} \rsem_{\gamma}^{\tilde V} \}} (c) \\
& \sqsubseteq \lsem \mu x .\phi_x \rsem_\gamma^{\tilde V}(c) \tag*{\begin{minipage}[t]{7cm}(characterisation of $\lsem \mu x.\phi_x \rsem^{\tilde V}_\gamma$ as supremum of an increasing, ordinal-indexed chain, see Theorem~\ref{thm:cousot})\end{minipage}} \\
\end{align*}
For the first inequality above, note that, by the proof of Proposition~\ref{prop-meas} (the case when $\phi = \mu x .\phi_x$), $\lsem \phi_x^{(i)} \rsem^V_{\zeta'_C}$ is measurable. For the second inequality above, note that, as in the proof of Proposition~\ref{prop-meas}, we again cannot make direct use of the induction hypothesis, since in general $\fnd(\phi_x^{(i)}) < \fnd(\phi)$ does not hold. Thus, we need to appeal once more to an inductive proof, to show $\mu_\gamma(\Paths_c^V(\phi_x^{(i)})) \sqsubseteq \lsem \phi_x^{(i)} \rsem_{\gamma}^{\tilde V}(c)$ for $i = 0,1,\ldots$ and $c \in C$. This follows below.
\begin{enumerate}
\item $i = 0$. In this case, $\phi_x^{(i)} = \bot$ and the statement is immediate.
\item Assume that $\mu_\gamma(\Paths_c^V(\phi_x^{(i)})) \sqsubseteq \lsem \phi_x^{(i)} \rsem_{\gamma}^{\tilde V}(c)$. We then have:
\begin{align*}
& \phantom{{}={}} \mu_\gamma(\Paths_c^V(\phi_x^{(i+1)})) \\
& = \mu_\gamma(\Paths_c^V(\phi_x[\phi_x^{(i)}/x])) \tag{definition of $\phi_x^{(i+1)}$}\\
& = \mu_\gamma(\Paths_c^{V \cup \{x \mapsto \lsem \phi_x^{(i)} \rsem^V_{\zeta'_C}\}} (\phi_x)) \tag{proof of Proposition~\ref{prop-meas}}\\
& \sqsubseteq \lsem \phi_x \rsem_{\gamma}^{\reallywidetilde{V \cup \{ x \mapsto \lsem \phi_x^{(i)} \rsem^V_{\zeta'_C} \}}}(c) \tag{outer induction hypothesis, using $\fnd(\phi_x) < \fnd(\phi)$}\\
& = \lsem \phi_x \rsem_{\gamma}^{\tilde V \cup \{ x \mapsto \lambda c.\mu_\gamma(\Paths_c^V(\phi_x^{(i)}))}(c) \tag{definition of $\reallywidetilde{V \cup \{ x \mapsto \lsem \phi_x^{(i)} \rsem^V_{\zeta'_C} \}}$}\\
& \sqsubseteq \lsem \phi_x \rsem_{\gamma}^{\tilde V \cup \{ x \mapsto \lsem \phi_x^{(i)} \rsem_{\gamma}^{\tilde V}}(c) \tag{inner induction hypothesis, monotonicity}\\
& = \lsem \phi_x^{(i+1)} \rsem_{\gamma}^{\tilde V}(c) \tag{definition of $\phi_x^{(i+1)}$}\\
\end{align*}
\end{enumerate}
This concludes the inner induction proof, and the proof in the case $\phi = \mu x.\phi_x$.
\item[(f)] $\phi = \nu x .\phi_x$ with $\phi_x \in \mu\LL_\Lambda^\V$. The proof in this case is different from the case $\phi = \mu x .\phi_x$: although the third equality in that proof could be replaced by an inequality (as $\mu_\gamma(\bigcap\limits_{n \in \omega} A_n) \sqsubseteq \inf\limits_{n \in \omega} \mu_\gamma(A_n)$ holds whenever $A_n$ with $n \in \omega$ are measurable sets such that $A_0 \supseteq A_1 \supseteq \ldots$), the inequality required in the last step of the proof does not hold -- see e.g.~Example~\ref{counter-ex}, where for $\phi := \nu x.\mu y.(\langle a \rangle x \sqcup \langle b \rangle y)$, the sets $A_0,A_1,\ldots$ are precisely the sets of paths satisfying $\phi_x^{(1)}, \phi_x^{(2)},\ldots$, and where $\lsem \phi \rsem_\gamma(c_0) = \lsem \nu x.\phi_x \rsem_\gamma(c_0) = 0$ while $\inf\limits_{i \in \omega} \lsem \phi_x \rsem_\gamma^{\tilde V \cup \{x \mapsto \lsem \phi_x^{(i)} \rsem_{\gamma}^{\tilde V} \}} (c_0) = 1$. As before, we assume $\llsem \phi_x \rrsem_\gamma^V \sqsubseteq \lsem \phi_x \rsem_\gamma^{\tilde V}$ for each $V : \V \cup \{x\} \to \{0,1\}^{\Paths_C}$ such that the set $\Paths_c^V(y)$ is measurable for each $c \in C$ and $y \in \V \cup \{x\}$, with $\tilde V : \V \cup \{x\} \to S^C$ given by $\tilde V(y)(c) = \mu_\gamma(\Paths_c^V(y))$. We also fix $V : \V \to \{0,1\}^{\Paths_C}$ such that $\Paths_c^V(y)$ is measurable for each $c \in C$ and $y \in \V$, and let $\tilde V : \V \to S^C$ be given by $\tilde V(y)(c) = \mu_\gamma(\Paths_c^V(y))$ for $y \in \V$ and $c \in C$. To show $\llsem  \nu x.\phi_x \rrsem^V_\gamma \sqsubseteq \lsem \nu x.\phi_x \rsem^{\tilde V}_\gamma$, it suffices to show that $\llsem  \nu x.\phi_x \rrsem^V_\gamma$ (that is, $\mu_\gamma(\lsem \nu x.\phi_x \rsem^V_{\zeta'_C})$) is a post-fixpoint of the operator $\LOp^{\tilde V}_\phi : S^C \to S^C$ used in the definition of $\lsem \nu x.\phi_x \rsem^{\tilde V}_\gamma$. To this end, we write $\Op^V_\phi : \{0,1\}^{\Paths_C} \to \{0,1\}^{\Paths_C}$ for the operator used to define $\lsem \nu x.\phi_x \rsem^V_{\zeta'_C}$. We then have:
\begin{align*}
\LOp^{\tilde V}_\phi(\mu_\gamma(\lsem \nu x.\phi_x \rsem^V_{\zeta'_C})) & = \lsem \phi_x \rsem_\gamma^{\tilde V[\mu_\gamma(\lsem \nu x.\phi_x \rsem^V_{\zeta'_C})/x]} \tag{definition of $\LOp^{\tilde V}_\phi$, see Section~\ref{coalg-fixpoint-logics}}\\
& \sqsupseteq \llsem \phi_x \rrsem_\gamma^{V[\lsem \nu x.\phi_x \rsem^V_{\zeta'_C}/x]} \tag{induction hypothesis}\\
& = \mu_\gamma(\Paths^{V[\lsem \nu x.\phi_x \rsem^V_{\zeta'_C}/x]}(\phi_x)) \tag{Definition~\ref{path-based-semantics}}\\
& = \mu_\gamma(\lsem \phi_x \rsem_{\zeta'_C}^{V[\lsem \nu x.\phi_x \rsem^V_{\zeta'_C}/x]}) \tag{definition of $\lsem \phi_x \rsem_{\zeta'_C}$}\\
& = \mu_\gamma(\Op^V_\phi(\lsem \nu x.\phi_x \rsem^V_{\zeta'_C})) \tag{definition of $\Op^V_\phi$, see Section~\ref{coalg-fixpoint-logics}}\\
& = \mu_\gamma(\lsem \nu x.\phi_x \rsem^V_{\zeta'_C}) \tag{$\lsem \nu x.\phi_x \rsem^V_{\zeta'_C}$ is the greatest fixpoint of $\Op^V_\phi$}\\
\end{align*}
\end{enumerate}
\end{enumerate}
This concludes the proof of the inequality $\lsem \phi \rsem_\gamma \sqsupseteq \llsem \phi \rrsem_\gamma$ and also the proof of the theorem.\qedhere
\end{itemize}
\end{proof}

Theorem~\ref{thm-equiv} instantiates to all our example semirings (see Example~\ref{example-semirings}), thereby formalising the statements of Example~\ref{exa:instances}.

The next example describes the relationship between our logics on the one hand, and LTL and the linear-time $\mu$-calculus, interpreted over both non-deterministic and probabilistic transition systems, on the other.

\begin{exa}
\label{ex:ltl}
Let $\At$ denote a finite set of atomic propositions, and let $F : \Set \to \Set$ be given by $F = \Pow(\At) \times \Id \simeq \coprod\limits_{A \subseteq \At}\Id$. Then, finitely-branching, non-deterministic (probabilistic) transition systems can be viewed as $\T_S \circ F$-coalgebras, with $S$ the boolean (resp.~probabilistic) semiring: such transition systems are in one-to-one correspondence with $\Pow(\At) \times \T_S$-coalgebras, and these can in turn be viewed as $\T_S \circ F$-coalgebras by post-composing the coalgebra map with the strength map of the monad $\T_S$. Concretely, if $(C,\gamma)$ is the $\Pow(\At) \times \T_S$-coalgebra associated to a non-deterministic (probabilistic) transition system, then $(C,\st_{\Pow(\At),C} \circ \gamma)$ is a $\T_S \circ F$-coalgebra:
\begin{align*}
\UseComputerModernTips\xymatrix{
C \ar[r]^-{\gamma} & \Pow(\At) \times \T_S C \ar[rr]^-{\st_{\Pow(\At),C}} & & \T_S(\Pow(\At) \times C)}
\end{align*}
We immediately observe that, with the above choice of functor $F$, paths through a transition system (given by infinite sequences of states, with each pair of successive states belonging to the transition relation) are in one-to-one correspondence with paths through the associated $\T_S \circ F$-coalgebra (which additionally record the atomic propositions that hold in each state along the path). In what follows we will not distinguish between the two. Also, for this choice of $F$, both the qualitative logic of Section~\ref{coalg-fixpoint-logics}, interpreted over $F$-coalgebras, and the quantitative logic of Definition~\ref{syntax}, interpreted over $\T_S \circ F$-coalgebras, employ modal operators of the form $\langle A \rangle$ with $A \subseteq \At$; while the use of subsets of atomic propositions as modal operators may seem surprising at first, the encoding of LTL into B\"uchi automata also uses automata over the alphabet $\Pow(\At)$. For the qualitative logic associated to the above choice of $F$, given an $F$-coalgebra $(C,\gamma)$, the interpretation of the modal operators is as follows:
\begin{align*}
c \in \lsem \langle A \rangle \phi \rsem_\gamma^V ~\text{ iff }~ \pi_1(\gamma(c)) = A \text{ and } \pi_2(\gamma(c)) \in \lsem \phi \rsem_\gamma^V
\end{align*}
for $c \in C$ and $V : \V \to \{0,1\}^C$ a valuation. As a result, the LTL next operator $\next \_$ is recovered as $\bigsqcup\limits_{A \subseteq \At} \langle A \rangle\_$, whereas atomic propositions $a \in \At$ are recovered as $\bigsqcup\limits_{A \subseteq \At, A \ni a} \langle A \rangle \top$. Now recall that standard propositional operators are \emph{not} present in our logics. As a result, a direct encoding of LTL into our qualitative logic for $F$-coalgebras, defined by induction on the structure of formulas, is not possible. That said, the presence of guarded disjunctions means that our qualitative logic for $F$-coalgebras \emph{can} encode all of LTL: to each LTL formula $\phi$ one can associate a non-deterministic B\"uchi automaton $B_\phi$ over the alphabet $\Pow(\At)$, and equivalently a \emph{deterministic} parity automaton $A_\phi$ over the same alphabet, in such a way that $\phi$ holds on an infinite path $p$ iff $B_\phi$/$A_\phi$ accepts the infinite word over $\Pow(\At)$ induced by $p$\footnote{This infinite word collects the sets of atomic propositions that hold in the states along $p$.}. Now using results in \cite{CirsteaSH17}, the automaton $A_\phi$ can be associated an equivalent $\mu\LL_{\Pow(\At)}$-formula. As a result, although our logics lack conjunctions and arbitrary disjunctions, they are at least as expressive as LTL. (They are in fact \emph{more} expressive, given that the expressiveness of deterministic parity automata goes beyond that of LTL.) A similar argument shows that our qualitative logics, when instantiated to the above functor $F$, match precisely the expressiveness of the linear-time $\mu$-calculus. It is also worth recalling from Example~\ref{exa:express} that typical linear-time properties can easily be encoded in our logics, in spite of the absence of conjunctions and arbitrary disjunctions.

We now show that existential LTL and the existential variant of the linear-time $\mu$-calculus, when interpreted over either non-deterministic or probabilistic transition systems, can be recovered as instances of our quantitative logics.
\begin{enumerate}
\item In the case of non-deterministic transition systems, we make the standard assumption that each state has at least one successor. A consequence of this is that the $\nu$-extent of each state is $1$. Then, for a non-deterministic transition system $T$ viewed as a $\T_S \circ F$-coalgebra $(C,\gamma)$ (with $S = (\{0,1\},\vee,0,\wedge,1)$), a state $c \in C$ and an existential LTL/linear-time $\mu$-calculus formula $\phi$ with semantics $\lsem \phi \rsem_T : C \to \{0,1\}$, writing $e(\phi)$ for the encoding of $\phi$ into our qualitative logic and using the observation that the interpretations of $\phi$ and $e(\phi)$ agree over infinite paths, we have:
\begin{align*}
\mu_\gamma (\{p \in \Paths_c \mid p \in \lsem e(\phi) \rsem_{\zeta'_C}\}) = 1 \quad \text{iff} \quad \Paths_c \cap  \lsem e(\phi) \rsem_{\zeta'_C} \ne \emptyset \\
\quad \text{iff} \quad \lsem \phi \rsem_{\Paths_c(T)} \ne \emptyset \quad \text{iff} \quad \lsem \phi \rsem_T(c) = 1
\end{align*}
where $\Paths_c(T)$ is the set of infinite computation paths from $c$ in the transition system $T$.
Thus, we recover the existential LTL semantics $\lsem \phi \rsem_T : C \to \{0,1\}$ of $\phi$. The same argument applies to the linear-time $\mu$-calculus. 
\item In the case of probabilistic transition systems, we assume that the sum of the probabilities of outgoing transitions from each state equals $1$. As a result, the $\nu$-extent of each state is again $1$. Then, for a probabilistic transition system $T$ viewed as a $\T_S \circ F$-coalgebra $(C,\gamma)$ (with $S = ([0,1],+,0,*,1)$), a state $c \in C$ and an existential LTL/linear-time $\mu$-calculus formula $\phi$ with semantics $\lsem \phi \rsem_T : C \to [0,1]$, again writing $e(\phi)$ for the encoding of $\phi$ into our qualitative logic and using the observation that the interpretations of $\phi$ and $e(\phi)$ agree over infinite paths, we have:
\begin{align*}
\mu_\gamma (\{p \in \Paths_c \mid p \in \lsem e(\phi) \rsem_{\zeta'_C}) = \mu_\gamma (\lsem \phi \rsem_{\Paths_c(T)}) = \lsem \phi \rsem_T(c)
\end{align*}
and thus we again recover the standard semantics $\lsem \phi \rsem_T : C \to [0,1]$ of $\phi$ over probabilistic transition systems. (Note that, by Example~\ref{ex:measure}, the definition of $\mu_\gamma$ agrees with the definition of the measure used in the semantics of probabilistic LTL.)
\end{enumerate}
\end{exa}

\section{Expressiveness of \texorpdfstring{$\mu\LL_\Lambda$}{μL\_Λ}}
\label{sem-char}

We now define (non-symmetric) semantic and logical distances between states of $\T_S \circ F$-coalgebras, with the latter induced by the logic $\mu\LL_\Lambda$, and show that the two coincide.

We begin by defining a binary operation $\oslash : S \times S \to S$ by
\begin{eqnarray*}
t \oslash s = \sup \{ u \mid u \bullet s \sqsubseteq t \}\,.
\end{eqnarray*}
The operation $\varoslash$ is thus a kind of inverse to the semiring multiplication. Assumption~\ref{ass:cpo} ensures that $\varoslash$ is well defined.

\begin{rem}
\label{rem:top}
As a result of the above definition, $t \oslash s = \top$ if and only if $s \sqsubseteq t$. Moreover, $(t \oslash s) \bullet s \sqsubseteq t$. The latter follows from
\begin{align*}
(t \oslash s) \bullet s & = \sup \{ u \mid u \bullet s \sqsubseteq t \} \bullet s \tag{definition of $\oslash$}\\
& = \sup \{ u  \bullet s \mid u \bullet s \sqsubseteq t \} \tag{$\bullet$ preserves suprema in the first argument, by Assumption~\ref{ass:cl}}\\
& \sqsubseteq t\\
\end{align*}
\end{rem}

\begin{rem}
\label{rem:adjoints}
The operation $\_ \varoslash s : (S,\sqsubseteq) \to (S,\sqsubseteq)$ is in fact a right adjoint to $\_ \bullet s : (S,\sqsubseteq) \to$ $(S,\sqsubseteq)$; that is, 
$u \bullet s \sqsubseteq t$ if and only if $u \sqsubseteq t \varoslash s $. The "only if" direction is immediate, whereas the "if" direction follows from the monotonicity of $\bullet$ in the second argument together with Remark~\ref{rem:top}.
\end{rem}

\begin{exa}~
\begin{enumerate}
\item When $S =  (\mathbb N^\infty,\min,\infty,+,0)$, $\varoslash : S \times S \to S$ instantiates to
$\varominus : \mathbb N^\infty \times \mathbb N^\infty \to \mathbb N^\infty$ given by
\begin{eqnarray*}
n \varominus m = \begin{cases}
0, ~\text{ if } m \ge n,\\
n - m,~\text{ otherwise}.
\end{cases}
\end{eqnarray*}
The above definition also applies to the bounded variants of the tropical semiring.
\item Similarly, when $S = ([0,1],+,0,*, 1)$, $\varoslash : S \times S \to S$ instantiates to $\varoslash : [0,1] \times [0,1] \to [0,1]$ given by
\begin{eqnarray*}
p \varoslash q = \begin{cases}
1, ~\text{ if } q \le p,\\
p/q,~\text{ otherwise}.
\end{cases}
\end{eqnarray*}
\item When $S = (\{0,1\},\vee,0,\wedge,1)$, $\varoslash: S \times S \to S$ can be defined using implication:
\begin{eqnarray*}
t \varoslash s = (s \rightarrow t)\,.
\end{eqnarray*}
\end{enumerate}
\end{exa}

The next two properties of $\varoslash$ will be useful later.

\begin{prop}
\label{prop:oslash0}
If $r, s, t \in S$ are such that $r \sqsubseteq s$, then $t \varoslash r \sqsupseteq t \varoslash s$.
\end{prop}
\begin{proof}
Immediate from the definition of $\varoslash$ and the monotonicity of $\bullet$ in the second argument.
\end{proof}

\begin{prop}
\label{prop:triangle}
The following holds: $(a \oslash b)\bullet (b \oslash c) \sqsubseteq a \oslash c$ for $a,b,c \in S$.
\end{prop}
\begin{proof}
The statement follows from:
\begin{align*}
(a \oslash b)\bullet (b \oslash c) & \sqsubseteq \sup \{ u \mid u \bullet b \sqsubseteq a \} \bullet \sup \{v \mid v \bullet c \sqsubseteq b \} \tag{definition of $\oslash$}\\
& = \sup \{ u \bullet v \mid u \bullet b \sqsubseteq a,  v \bullet c \sqsubseteq b \} \tag{Assumption~\ref{ass:cl}} \\
& \sqsubseteq \sup \{ u \bullet v \mid u \bullet v \bullet c \sqsubseteq a \} \tag{$u \bullet b \sqsubseteq a$ and $v \bullet c \sqsubseteq b$ imply $u \bullet v \bullet c \sqsubseteq a$}\\
& \sqsubseteq \sup \{ t \mid t \bullet c \sqsubseteq a \} \\
& \sqsubseteq a \oslash c \tag{definition of $\oslash$}
\end{align*}
\end{proof}

\begin{cor}
\label{cor:triangle}
Let $a_i,b_i,c_i \in S$ for $i \in I$. Then the following holds:
\begin{eqnarray*}
\inf\limits_{i \in I}(a_i \oslash b_i) \bullet \inf\limits_{i \in I}(b_i \oslash c_i) \sqsubseteq \inf\limits_{i \in I}(a_i \oslash c_i) 
\end{eqnarray*}
\end{cor}
\begin{proof}
We have:
\begin{eqnarray*}
\inf\limits_{i \in I}(a_i \oslash b_i) \bullet \inf\limits_{i \in I}(b_i \oslash c_i) \sqsubseteq (a_i \oslash b_i) \bullet (b_i \oslash c_i) \sqsubseteq a_i \oslash c_i 
\end{eqnarray*}
for all $i \in I$, and therefore:
\begin{eqnarray*}
\inf\limits_{i \in I}(a_i \oslash b_i) \bullet \inf\limits_{i \in I}(b_i \oslash c_i) & \sqsubseteq & \inf\limits_{i \in I} (a_i \oslash c_i)
\end{eqnarray*}
as stated.
\end{proof}

The following property of $\varoslash$ will also be needed later on.

\begin{prop}
\label{prop:oslash}
The operation $\varoslash : S \times S \to S$ satisfies $(\inf\limits_{i \in I}t_i) \varoslash s = \inf\limits_{i \in I}(t_i \varoslash s)$ for $s_i, t_i \in S$ with $i \in I$.
\end{prop}
\begin{proof}
Immediate from Remark~\ref{rem:adjoints}, using the fact that right adjoints preserve limits.
\end{proof}
 
The notion of \emph{partial trace}, defined next, plays a key role in our development.

\begin{defi}[Partial trace]
A \emph{partial ($F$-)trace} is an element of the initial $(\{*\} + F)$-algebra $(B,\beta)$.
\end{defi}
The notion of partial trace thus covers both (completed) finite traces and incomplete traces. It is similar to the notion of path fragment (Definition~\ref{path-fragm-def}), except that a partial trace does not record coalgebra states. (While we could have defined partial traces as an instance of Definition~\ref{path-fragm-def}, we believe a direct definition is clearer.) We write $\epsilon = \beta(\iota_1(*))$ for the empty partial trace.

The notion of \emph{partial trace behaviour} defined next exploits the fact that $\beta : \{*\} + F B \to B$ is an isomorphism.

\begin{defi}[Partial trace behaviour]
\label{par-trace-beh}
The \emph{partial trace behaviour} of a $\T_S \circ F$-coalgebra $(C,\gamma)$, denoted $\ptr_\gamma : C \times B \to S$, is the greatest fixpoint of the operator on $\Rel_{C,B}$ given by
\begin{eqnarray*}
\UseComputerModernTips\xymatrix@+0.5pc{
\Rel_{C, B} \ar[r]^-{\Rel(F)} & \Rel_{F C,F B} \ar[r]^-{\E_{\T_S}} & \Rel_{\T_S F C, F B}  \ar[r]^-{(\gamma \times \id_{F B})^*} & \Rel_{C, FB} \ar[r]^-{[\extent^\nu_\gamma,\_] \circ i} & \Rel_{C,\{*\}+F B} \ar[r]^-{(\id_{C} \times \beta^{-1})^*} & \Rel_{C,B}}
\end{eqnarray*}
where $\extent^\nu_\gamma$ is the $\nu$-extent of Definition~\ref{extent-coalgebra}, and where $i : C \times (\{*\} + F B) \to C + C \times F B$ is the isomorphism arising from the distributivity of products over coproducts in $\Set$.
\end{defi}
Thus, similarly to the notions of maximal trace behaviour (Definition~\ref{def-max-trace-beh}) and finite trace behaviour (Definition~\ref{def-finite-trace-beh}), the notion of partial trace behaviour assigns values in $S$ to pairs consisting of a state in a $\T_S \circ F$-coalgebra and a partial $F$-trace. The operator in Definition~\ref{par-trace-beh} is similar to the one used to define finite trace behaviour (Definition~\ref{def-finite-trace-beh}), while taking into account that $B$ also contains incomplete traces. In particular, $\ptr_\gamma$ assigns the value $\extent^\nu_\gamma(c)$ to the pair $(c,\epsilon)$.

We now observe that partial traces directly correspond to certain formulas in our logic, namely to those fixpoint-free formulas which only contain modalities of the form $\langle \lambda \rangle$ with $\lambda \in \Lambda$. We let $\LL_\Lambda^1 \subseteq \LL_\Lambda$ consist of all such formulas. 

\begin{rem}
\label{partial-traces-modal-formulas}
Partial traces $b \in B$ are in one-to-one correspondence with modal formulas $\phi_b \in \LL_\Lambda^1$ and moreover, $\ptr_\gamma(c,b) = \lsem \phi_b \rsem_\gamma(c)$ for all states $c$ of a $\T_S \circ F$-coalgebra $(C,\gamma)$. In particular, the empty partial trace $\epsilon$ corresponds to the modal formula $\top$. This can be proved by an easy induction over $B$, using the specific shape of the functor $F$.
\end{rem}

We are now ready to define both a semantic distance and a logical distance between states of $\T_S \circ F$-coalgebras.

\begin{defi}[Linear-time distance]
\label{def:lt}
The \emph{linear-time distance} $\d_{\gamma,\delta} : C \times D \to S$ from states of a $\T_S \circ F$-coalgebra $(C,\gamma)$ to states of a $\T_S \circ F$-coalgebra $(D,\delta)$ is given by
\begin{eqnarray*}
\d_{\gamma,\delta}(c,d) = \inf\limits_{b \in B} (\ptr_\delta(d,b) \oslash \ptr_\gamma(c,b))
\end{eqnarray*}
for $c \in C$ and $d \in D$.
\end{defi}
Intuitively, the distance from $c$ to $d$ measures how much worse the partial trace behaviour of state $d$ is compared to the partial trace behaviour of state $c$. Using Remark~\ref{rem:top}, we have that $\d_{\gamma,\delta}(c,d) = \top$ if and only if $\ptr_\gamma(c,b) \sqsubseteq \ptr_\delta(d,b)$ for all $b \in B$. Informally, $\d_{\gamma,\delta}(c,d) = \top$ if and only if the partial trace behaviour of $d$ is "better" than the partial trace behaviour of $c$, in that the quantity $d$ associates to a partial trace is above the quantity $c$ associates to that same partial trace, for every partial trace in $B$.

\begin{exa}
When $S = (\{0,1\},\vee,0,\wedge,1)$, $\d_{\gamma,\delta}(c,d) = 1 = \top$ if and only if state $d$ can exhibit all partial traces which state $c$ can exhibit. Also, when $S = (\mathbb N^\infty,\min,\infty,+,0)$, $\d_{\gamma,\delta}(c,d)$ is the maximum difference, taken across all partial traces, between the cost of exhibiting a partial trace from state $d$ in $(D,\delta)$ and the cost of exhibiting the same partial trace from state $c$ in $(C,\gamma)$. In particular, if $d$ has a lower cost than $c$ for each partial trace, $\d_{\gamma,\delta}(c,d) = 0 = \top$. 
\end{exa}

\begin{defi}[Logical distance]
\label{def:logi}
Given a set $L \subseteq \mu\LL_\Lambda$ of formulas, the \emph{logical distance} $\d^L_{\gamma,\delta} : C \times D \to S$ from a state $c \in C$ of a $\T_S \circ F$-coalgebra $(C,\gamma)$ to a state $d \in D$ of a $\T_S \circ F$-coalgebra $(D,\delta)$ is given by
\begin{eqnarray*}
\d^L_{\gamma,\delta}(c,d) = \inf\limits_{\phi \in L} (\llsem \phi \rrsem_\delta(d) \oslash \llsem \phi \rrsem_\gamma(c)) \,.
\end{eqnarray*}
\end{defi}

Similarly to the linear-time distance, we have that $\d_{\gamma,\delta}^L(c,d) = \top$ if and only if $\llsem \phi \rrsem_\gamma(c) \sqsubseteq \llsem \phi \rrsem_\delta(d)$ for all $\phi \in L$.

The next result states that both distances satisfy abstract versions of the reflexivity axiom (when restricted to a single coalgebra) and of the triangle inequality required of any pseudometric. The fact that the inequalities in (\ref{i3}) and (\ref{i4}) of Proposition~\ref{i34} are reversed compared to the standard triangle inequality stems from the fact that a distance of $\top \in S$ in our setting corresponds to a distance of $0 \in \mathbb R_{\ge 0}$ in the standard setting. We note also that our distances are \emph{not} symmetric when restricted to a single $\T_S \circ F$-coalgebra.

\begin{prop}
\label{i34}
Let $(C,\gamma)$, $(D,\delta)$ and $(E,\eta)$ be $\T_S \circ F$-coalgebras, and let $L \subseteq \mu\LL_\Lambda$. Then the following hold:
\begin{enumerate}
\item $\d_{\gamma,\gamma}(c,c) = \top$,
\item $\d_{\gamma,\gamma}^L(c,c) = \top$,
\item\label{i3} $\d_{\gamma,\delta}(c,d) \bullet \d_{\delta,\eta}(d,e) \sqsubseteq \d_{\gamma,\eta}(c,e)$,
\item\label{i4} $\d_{\gamma,\delta}^L(c,d) \bullet \d_{\delta,\eta}^L(d,e) \sqsubseteq \d_{\gamma,\eta}^L(c,e)$,
\end{enumerate}
for all $c \in C$, $d \in D$ and $e \in E$.
\end{prop}
\begin{proof}
The statements follow directly from Definitions~\ref{def:lt} and respectively \ref{def:logi}, together with Corollary~\ref{cor:triangle}.
\end{proof}

\begin{rem}
Lawvere \cite{Lawvere} observed that pseudometric spaces are the same as categories enriched in the quantale $(\mathbb N^\infty,\ge)$, with $+$ as tensor. Our use of semiring-valued distances is similar in spirit to loc.\,cit.~-- in the above quantale, $0$ is the top element and a distance of $\top$ corresponds to two points being equivalent. Indeed, our assumptions on the universe $S$ of quantities make $(S,\sqsubseteq)$ a quantale (with $\bullet$ as tensor), although we choose not to expand on or exploit this observation in the present paper.  
\end{rem}

An immediate consequence of Remark~\ref{partial-traces-modal-formulas} is that logical distance w.r.t.~modal formulas in $\LL_\Lambda^1$ coincides with the linear-time distance $\d_{\gamma,\delta}$. This immediately gives $\d_{\gamma,\delta}^{\mu\LL_\Lambda} \sqsubseteq \d_{\gamma,\delta}^{\LL_\Lambda^1} = \d_{\gamma,\delta}$. The remainder of this section shows that logical distance w.r.t.~the \emph{full} logic $\mu\LL_\Lambda$ also coincides with $\d_{\gamma,\delta}$. The next two lemmas, formalising some properties of $\varoslash$, will prove useful for this.
\begin{lem}
\label{lem:inf}
Let $s_i, t_i \in S$ for $i \in I$. Then the following holds:
\begin{eqnarray*}
(\inf_{i \in I}t_i) \varoslash (\inf_{i \in I}s_i) \sqsupseteq \inf_{i \in I}(t_i \varoslash s_i)
\end{eqnarray*}
\end{lem}
\begin{proof}
We have
\begin{align*}
\inf_{i \in I}(t_i \varoslash s_i) & \sqsubseteq \inf_{i \in I}(t_i \varoslash (\inf_{j \in I}s_j))  \tag{Proposition~\ref{prop:oslash0}}\\
& = (\inf_{i \in I}t_i) \varoslash (\inf_{j \in I}s_j) \tag{Proposition~\ref{prop:oslash}}\\
& = (\inf_{i \in I}t_i) \varoslash (\inf_{i \in I}s_i) \tag*{\qedhere}
\end{align*}
\end{proof}

\begin{lem}
\label{lem:sum-oslash}
The following holds for $s_i, t_i \in S$ with $i \in \omega$:
\begin{eqnarray*}
(\sum\limits_{i \in \omega}t_i) \oslash (\sum\limits_{i \in \omega}s_i) \sqsupseteq \inf\limits_{i \in \omega} (t_i \oslash s_i)
\end{eqnarray*}
\end{lem}
\begin{proof}
Let $u = \inf\limits_{i \in \omega} (t_i \oslash s_i)$. Then, $u \sqsubseteq t_i \oslash s_i$ for $i \in \omega$, and therefore, by Remark~\ref{rem:top} and the monotonicity of $\bullet$ in the second argument, $u \bullet s_i \sqsubseteq t_i$ for all $i \in \omega$. We then have:
\begin{align*}
u \bullet \sum\limits_{i \in \omega}s_i & = \sum\limits_{i \in \omega} (u \bullet s_i) \tag{Remark~\ref{rem:distrib-countable-sums}}\\
& \sqsubseteq \sum\limits_{i \in \omega} t_i \tag{see above}\\
\end{align*}
and therefore, using the definition of $\oslash$, $u \sqsubseteq (\sum\limits_{i \in \omega}t_i) \oslash (\sum\limits_{i \in \omega}s_i)$ as required.
\end{proof}

We will also use the following two lemmas to rephrase the definition of  $\mu_\gamma : \M_c \to S$ with $c \in C$ (see Definitions~\ref{induced-measure} and \ref{def:induced-outer-meas}) on measurable sets of paths of the form $\Paths_c(\phi)$ with $\phi \in \mu\LL_\Lambda$, by replacing cylinder set covers with covers whose elements are of the form $\Paths_c(\psi)$ with $\psi \in \LL^1_\Lambda$. To this end, we say that $\psi_2 \in  \LL^1_\Lambda$ \emph{extends} $\psi_1 \in \LL^1_\Lambda$ when the partial trace associated to $\psi_1$ (see Remark~\ref{partial-traces-modal-formulas}) is a prefix of that associated to $\psi_2$. Also, we say that $\psi_1, \psi_2 \in \LL^1_\Lambda$ are \emph{disjoint} when their associated partial traces are incompatible (i.e. they are not both prefixes of some other partial trace). Thus, two disjoint formulas in $\LL^1_\Lambda$ cannot both hold on a path. Disjointness then extends to pairs of formulas in $\{\bot\} \cup \LL_\Lambda^1$, with $\bot$ assumed to be disjoint from any formula in $\{\bot\} \cup \LL_\Lambda^1$, including itself.

The first lemma states that the semantics, over paths in $(C,\gamma)$, of any formula in $\LL^1_\Lambda$ can be expressed in terms of cylinder sets.

\begin{lem}
\label{lemma:form}
Let $(C,\gamma)$ be a $\T_S\circ F$-coalgebra, and let $\psi \in \LL^1_\Lambda$. Then, for $c \in C$, $\Paths_c(\psi)$ can be written as a finite disjoint union of pairwise disjoint cylinder sets.
\end{lem}
\begin{proof}[Proof (Sketch)]
Structural induction on $\psi$, using the fact that there are finitely-many transitions from each state in $(C,\gamma)$, and thus also finitely many transitions which match a given modal operator $\langle \lambda \rangle$ with $\lambda \in \Lambda$.
\end{proof}

The second lemma shows that the semantics, over paths in $(C,\gamma)$, of any fixpoint formula $\phi \in \mu\LL_\Lambda$ can be expressed in terms of the semantics of formulas in $\LL^1_\Lambda$.

\begin{lem}
\label{lemma:form0}
Let $(C,\gamma)$ be a $\T_S \circ F$-coalgebra, and let $\phi \in \mu\LL_\Lambda^{\V}$. Moreover, let $V : \V \to \{0,1\}^{\Paths_C}$ be such that for each $x \in V$, $V(x)$ is of the form
\begin{eqnarray}
\label{eq:snd}
\bigcap \limits_{i \in \omega} (\bigcup_{j \in \omega} \lsem \phi^i_j \rsem_{\zeta'_C})
\end{eqnarray}
with the formulas $\phi^i_j \in \{\bot\} \cup \LL^1_\Lambda$ for $i,j \in \omega$ subject to the following conditions:
\begin{itemize}
\item for each $i \in \omega$, the formulas $(\phi^i_j)_{j \in \omega}$ are pairwise disjoint,
\item for each $i \in \omega$, the family $(\phi^{i+1}_j)_{j \in \omega}$ \emph{refines} the family $(\phi^i_j)_{j \in \omega}$; that is, for each $j \in \omega$, there exists $k \in \omega$ such that $\phi^{i+1}_j$ extends $\phi^i_k$.
\end{itemize}
Then, $\lsem \phi \rsem_{\zeta'_C}$ can also be written in the above form.
\end{lem}
\begin{proof}[Proof (Sketch)]
The proof is by induction of the fixpoint nesting depth of $\phi$.
\begin{enumerate}
\item\label{znd} $\fnd(\phi) = 0$. We prove the statement by structural induction on $\phi$.
\begin{enumerate}
\item For $\phi = \top$, the statement is immediate -- take $\phi^i_0 = \top$ and $\phi^i_{j+1} = \bot$ for $i,j \in \omega$.
\item For $\phi = x \in \V$, the statement follows from the assumption on $V$.
\item For $\phi = \langle \lambda \rangle(\phi_1,\ldots,\phi_{\arity(\lambda)})$ and $i \in \omega$, assume that each $\lsem \phi_k \rsem_{\zeta'_C}$ with $k\!\in\!\{1,\ldots,\arity(\lambda)\}$ can be written as $\bigcap \limits_{i \in \omega} (\bigcup_{j \in \omega} \lsem (\phi_k)^i_j \rsem_{\zeta'_C})$. For $i \in \omega$, let $\{\phi^i_j \mid j \in \omega\}$ contain all formulas of the form $\langle \lambda \rangle((\phi_1)^i_{j_1},\ldots,(\phi_{\arity(\lambda)})^i_{j_{\arity(\lambda)}}))$ with $j_k \in \omega$ for $k \in \{1,\ldots,\arity(\lambda)\}$. Clearly, for each $i \in \omega$, the formulas $(\phi^i_j)_{j \in \omega}$ satisfy the two conditions in the statement of the lemma, given that each family $((\phi_k)^i_j)_{j \in \omega}$ with $k \in \{1,\ldots,\arity(\lambda)\}$ does so. The fact that (\ref{eq:snd}) holds can also be easily proved.
\item The case of modalities incorporating disjunctions (this also covers the case $\phi = \bot$) is treated similarly.
\end{enumerate}
\item $\fnd(\phi) > 0$. As for case (\ref{znd}), the statement follows by structural induction on $\phi$, with four similar sub-cases treated in the same way, but with two additional cases, considered below:
\begin{enumerate}
\item[(e)] For $\phi = \mu x.\phi_x$, just like in the proof of Proposition~\ref{prop-meas}, we cannot directly apply the induction hypothesis. Instead, we use the equality (\ref{equality1}) to reduce showing that $\lsem \phi \rsem_{\zeta'_C}$ can be written in the required form to showing that each of $\lsem \phi_x \rsem_{\zeta'_C}^{V \cup \{x \mapsto \lsem \phi_x^{(n)} \rsem_{\zeta'_C}^V\}}$, with $n \in \omega$, can be written in this form. Once this is shown (by an easy induction on $n$), the resulting families $(\phi_x^{n})^i_j$ with $i,j \in \omega$ and $n \in \omega$ can be used to obtain the required $(\phi^i_j)_{j \in \omega,i\in \omega}$. For this, the observation that, for fixed $i \in \omega$ and $n \in \omega$, the family $(\phi_x^{n+1})^i_j$ with $j \in \omega$ contains\footnote{once formulas containing $\bot$ have been identified with $\bot$} the family $(\phi_x^{n})^i_j$ with $j \in \omega$, is used; this, in turn, follows from $\phi$ not containing unguarded occurrences of any fixpoint variables.
\item[(f)] The proof in the case when $\phi = \nu x.\phi_x$ is similar and uses the equality (\ref{equality2}) to reduce showing that $\lsem \phi \rsem_{\zeta'_C}$ can be written in the required form to showing that each of $\lsem \phi_x \rsem_{\zeta'_C}^{V \cup \{x \mapsto \lsem \phi_x^{(n)} \rsem_{\zeta'_C}^V\}}$, with $n \in \omega$, can be written in this form. In this case, the observation that, for fixed $i \in \omega$ and $n \in \omega$, the family $(\phi_x^{n+1})^i_j$ with $j \in \omega$ refines the family $(\phi_x^{n})^i_j$ with $j \in \omega$, is used; this, again, is a consequence of guardedness.\qedhere
\end{enumerate}
\end{enumerate}
\end{proof}

\begin{exa}
Let $F = \{a,b\} \times \Id$.
\begin{itemize}
\item For $\phi = \mu x.\nu y.(\langle a \rangle x \sqcup \langle b \rangle y)$, capturing the property that only finitely-many $a$-labelled transitions are present, $\lsem \phi \rsem_{\zeta'_C}$ is given by the intersection $\bigcap\limits_{i \in \omega}A_i$ where, for $i \in \omega$, $A_i = \bigcup\limits_{k \in \omega, n_1,\ldots,n_k \in \mathbb N\setminus \{0\}} \lsem \langle b \rangle^{n_1} \langle a \rangle \langle b \rangle^{n_2} \langle a \rangle \ldots \langle b \rangle^{n_k} \langle a \rangle \langle b \rangle^i \top \rsem_{\zeta'_C}$. In other words, for $i \in \omega$, the family $(\phi^i_j)_{j \in \omega}$ contains those formulas in $\LL_\Lambda^1$ whose associated partial trace contains at least $i$ \emph{consecutive} $b$-transitions.
\item For $\phi = \nu x.\mu y.(\langle a \rangle x \sqcup \langle b \rangle y)$, capturing the property that infinitely-many $a$-labelled transitions are present, $\lsem \phi \rsem_{\zeta'_C}$ is given by the intersection $\bigcap\limits_{i \in \omega}A_i$ where, for $i \in \omega$, $A_i = \bigcup\limits_{n_0,\ldots,n_{i-1} \in \mathbb N} \lsem \langle b \rangle^{n_0} \langle a \rangle \langle b \rangle^{n_1} \langle a \rangle \ldots \langle b \rangle^{n_{i-1}} \langle a \rangle \top \rsem_{\zeta'_C}$. In other words, for $i \in \omega$, the family $(\phi^i_j)_{j \in \omega}$ contains those formulas in $\LL_\Lambda^1$ whose associated partial trace contains at least $i$ $a$-labelled transitions.
\end{itemize}
\end{exa}

We now have the following result.

\begin{prop}
\label{prop:formula}
Let $(C,\gamma)$ be a $\T_S\circ F$-coalgebra, and let $c \in C$ and $\phi \in \mu\LL_\Lambda$. Then:
{\small \begin{eqnarray*}
\mu_\gamma(\Paths_c(\phi))= \inf \{\sum\limits_{i \in \omega} \mu_\gamma(\Paths_c(\phi_i)) \mid 
(\phi_i \in \LL_\Lambda^1)_{i \in \omega} \text{ pairwise disj.},\, \Paths_c(\phi) \subseteq \bigcup\limits_{i \in \omega}\Paths_c(\phi_i)\}
\end{eqnarray*}}
\end{prop}
\begin{proof}
Recall that
\begin{eqnarray*}
\mu_\gamma(\Paths_c(\phi))= \inf \{ \mu_\gamma(\C) \mid \C \text{ is a countable, disjoint cylinder set cover for $\Paths_c(\phi)$} \} \,.
\end{eqnarray*}
Clearly, the lhs in the statement of Proposition~\ref{prop:formula} is $\sqsubseteq$ the rhs -- by Lemma~\ref{lemma:form}, countable covers for $\Paths_c(\phi)$ made of pairwise disjoint sets of the form $\Paths_c(\psi)$ with $\psi \in \LL^1_\Lambda$ yield countable covers made of pairwise disjoint cylinder sets, with the same measure.

To show that the lhs is $\sqsupseteq$ the rhs, let $\C$ be a countable cover for $\Paths_c(\phi)$ consisting of pairwise disjoint cylinder sets. We will use $\C$ to construct a countable cover for $\Paths_c(\phi)$ consisting of pairwise disjoint sets of the form $\Paths_c(\psi) \subseteq \bigcup \C$ with $\psi \in \LL^1_\Lambda$. To this end, let $\Psi = \{ \psi \in \LL^1_\Lambda \mid \emptyset \ne \Paths_c(\psi) \subseteq \bigcup \C \}$, and let $\C_\Psi = \{ \Paths_c(\psi) \mid \psi \in \Psi\}$. Note that, since $(C,\gamma)$ is finitely branching, $\C_\Psi$ is countable. Clearly, $\bigcup \C_\Psi \subseteq \bigcup \C$. We will now show that $\C_\Psi$ covers $\Paths_c(\phi)$. Assume that this is not the case; that is, there exists $p \in \Paths_c(\phi)$ such that $p \not\in \Paths_c(\psi)$ for any $\psi \in \Psi$. By $(Z_C,\zeta_C)$ being a final $C \times F$-coalgebra, with $Z_C$ obtained as an $\omega^\op$-limit, there exist cylinder sets $C^p_i$ of depth uniformly $i$, with $i \in \omega$, such that $\{p\} = \bigcap \limits_{i \in \omega} C^p_i$. Remark~\ref{partial-traces-modal-formulas} now yields a formula $\psi_i = \psi_{C^p_i} \in \LL^1_\Lambda$ (which corresponds to the underlying partial trace of $C^p_i$) for each $i \in \omega$. We immediately obtain $\psi_i \not \in \Psi$ (since $p \in C^p_i \subseteq \Paths_c(\psi_i)$) for $i \in \omega$. We now use the fact that $\Paths_c(\phi) = \bigcap\limits_{i \in \omega} \bigcup\limits_{j \in \omega} \Paths_c(\phi^i_j)$, with $\phi^i_j \in \LL^1_\Lambda$ for $i,j \in \omega$ (see Lemma~\ref{lemma:form0}) to derive a contradiction. From this and $p \in \Paths_c(\phi)$ we obtain a decreasing sequence $\Paths_c(\phi^0_{j_0}) \supseteq \Paths_c(\phi^1_{j_1}) \supseteq \ldots$ with $p \in \bigcap\limits_{i \in \omega} \Paths_c(\phi^i_{j_i})$. Moreover, $\phi^i_{j_i} \not\in \Psi$ (since $p \in \Paths_c(\phi^i_{j_i})$) for any $i \in \omega$. As a result, $\Paths_c(\phi^i_{j_i}) \not\subseteq \bigcup \C$ for any $i \in \omega$. It then follows that $\bigcap\limits_{i \in \omega} \Paths_c(\phi^i_{j_i}) \not\subseteq \bigcup \C$ (see below). This, in turn, gives $p' \in (\bigcap\limits_{i \in \omega} \Paths_c(\phi^i_{j_i})) \setminus (\bigcup \C)$. Moreover, $p'$ is $F$-behaviourally equivalent to $p$ (given the shape of the $\phi^i_{j_i}$s and the fact that both $p$ and $p'$ belong to $ \Paths_c(\phi^i_{j_i})$ for $i \in \omega$), which now gives $p' \in \Paths_c(\phi)$ (as behaviourally equivalent states in $(Z_C,\zeta_C)$ satisfy the same formulas of $\mu\LL_\Lambda$). We have therefore derived a contradiction, since $p' \not\in \bigcup \C$ and therefore $p' \not\in \Paths_c(\phi)$. This concludes the proof that $\C_\Psi$ covers $\Paths_c(\phi)$. We have thus constructed a countable cover for $\Paths_c(\phi)$, with elements of the form $\Paths_c(\psi) \subseteq \bigcup \C$ with $\psi \in \LL^1_\Lambda$. The last step is to provide a disjoint sub-cover of $\C_\Psi$, which still covers $\Paths_c(\phi)$; this can be done by first replacing each $\Paths_c(\psi)$ for which $\psi \in \Psi$ is \emph{not} of uniform depth by a finite collection of sets of the form $\Paths_c(\psi)$ with $\psi$ of uniform depth, and then removing from $\C_\Psi$ all sets $\Paths_c(\psi')$ for which there exists $\psi \in \Psi$ with $\Paths_c(\psi') \subsetneq \Paths_c(\psi)$. This transformation preserves the measure of $\C_\Psi$. Also, the resulting disjoint cover $\C_\Psi$ for $\Paths_c(\phi)$ has measure below that of $\C$ (as $\bigcup \C_\Psi \subseteq \bigcup \C$). As a result, the lhs in the statement of Proposition~\ref{prop:formula} is $\sqsupseteq$ the rhs.

It remains to prove the earlier claim that, if $\Paths_c(\phi^i_{j_i}) \not\subseteq \bigcup \C$ for $i \in \omega$, then $\bigcap\limits_{i \in \omega} \Paths(\phi^i_{j_i}) \not\subseteq \bigcup \C$. For this, note that by Lemma~\ref{lemma:form}, each $\Paths_c(\phi^i_{j_i})$ is a \emph{finite} union of pairwise disjoint cylinder sets (which we can assume w.l.o.g.~are of depth uniformly $i$ -- if not, simply replace them by finite unions of cylinder sets of depth uniformly $i$), at least one of which is not included in $\bigcup \C$. This allows us to construct a decreasing chain $C_0 \supseteq C_1 \supseteq \ldots$ of cylinder sets with $C_i \subseteq \Paths_c(\phi^i_{j_i})$ and $C_i \not\subseteq \bigcup \C$ for $i \in \omega$. Then, $\bigcap\limits_{i \in \omega} C_i \not\subseteq \bigcup \C$. For, if this was not the case, since $\bigcap\limits_{i \in \omega} C_i = \{p''\}$ with $p''\in Z_C$ (due to the shape of the $C_i$s), we would have $p'' \in \bigcup \C$, and therefore $p'' \in C$ for some $C \in \C$. But this would give $C_i \subseteq C$ for some $i \in \omega$, and therefore $C_i \subseteq \bigcup \C$, which contradicts the definition of the $C_i$s. This concludes the proof.
\end{proof}

We finally show that the logical distance w.r.t.~$\mu\LL_\Lambda$ and the linear-time distance coincide.
\begin{thm}
Let $(C,\gamma)$ and $(D,\delta)$ be two $\T_S \circ F$-coalgebras. Then, $\d_{\gamma,\delta} = \d_{\gamma,\delta}^{\mu\LL_\Lambda}$.
\end{thm}
\begin{proof}
We have already observed that $\d_{\gamma,\delta}^{\mu\LL_\Lambda} \sqsubseteq \d_{\gamma,\delta}$. It therefore suffices to prove that, for $c \in C$, $d \in D$, and $\phi \in \mu\LL_\Lambda$, $\llsem \phi \rrsem_\delta(d) \oslash \llsem \phi \rrsem_\gamma(c) \sqsupseteq \d_{\gamma,\delta}(c,d)$. This then gives $\d_{\gamma,\delta}^{\mu\LL_\Lambda} \sqsupseteq \d_{\gamma,\delta}$. We have:
\begin{align*}
& \phantom{{}={}} \llsem \phi \rrsem_\delta(d) \oslash \llsem \phi \rrsem_\gamma(c) \\
& = \mu_\gamma(\Paths_d(\phi))  \oslash \mu_\delta(\Paths_c(\phi)) \tag{definition of $\llsem \_ \rrsem$} \\
& =  \inf \{\sum\limits_{i \in \omega} \mu_\delta(\lsem \phi_i \rsem_{\zeta_D}) \mid 
(\phi_i \in \LL_\Lambda^1)_{i \in \omega} \text{ pairwise disjoint},\, \Paths_d(\phi) \subseteq \bigcup\limits_{i \in \omega}\lsem \phi_i \rsem_{\zeta_D}\}~\varoslash \\
& \phantom{{}={}} \inf \{\sum\limits_{i \in \omega} \mu_\gamma(\lsem \phi_i \rsem_{\zeta_C}) \mid 
(\phi_i \in \LL_\Lambda^1)_{i \in \omega} \text{ pairwise disjoint},\, \Paths_c(\phi) \subseteq \bigcup\limits_{i \in \omega}\lsem \phi_i \rsem_{\zeta_C}\}  \tag{Proposition~\ref{prop:formula}} \\
& = \inf \{\sum\limits_{i \in \omega} \mu_\delta(\lsem \phi_i \rsem_{\zeta_D}) \mid 
(\phi_i \in \LL_\Lambda^1)_{i \in \omega} \text{ pairwise disjoint},\, \\
& \phantom{{}={}} (\phi_i)_{i \in \omega} \text{ minimal s.t.~} \Paths_c(\phi) \subseteq \bigcup\limits_{i \in \omega}\lsem \phi_i \rsem_{\zeta_C} \text{ and } \Paths_d(\phi) \subseteq \bigcup\limits_{i \in \omega}\lsem \phi_i \rsem_{\zeta_D} \}  ~\varoslash \\
& \phantom{{}={}} \inf \{\sum\limits_{i \in \omega} \mu_\gamma(\lsem \phi_i \rsem_{\zeta_C}) \mid 
(\phi_i \in \LL_\Lambda^1)_{i \in \omega} \text{ pairwise disjoint},\, \\
& \phantom{{}={}} (\phi_i)_{i \in \omega} \text{ minimal s.t.~} \Paths_c(\phi) \subseteq \bigcup\limits_{i \in \omega}\lsem \phi_i \rsem_{\zeta_C} \text{ and } \Paths_d(\phi) \subseteq \bigcup\limits_{i \in \omega}\lsem \phi_i \rsem_{\zeta_D} \} \tag{*} \\
& \sqsupseteq \inf \{\sum\limits_{i \in \omega} \mu_\delta(\lsem \phi_i \rsem_{\zeta_D}) \oslash \sum\limits_{i \in \omega} \mu_\gamma(\lsem \phi_i \rsem_{\zeta_C}) \mid 
(\phi_i \in \LL_\Lambda^1)_{i \in \omega} \text{ pairwise disjoint},\, \\
& \phantom{{}={}} (\phi_i)_{i \in \omega} \text{ minimal s.t.~} \Paths_c(\phi) \subseteq \bigcup\limits_{i \in \omega}\lsem \phi_i \rsem_{\zeta_C} \text{ and } \Paths_d(\phi) \subseteq \bigcup\limits_{i \in \omega}\lsem \phi_i \rsem_{\zeta_D} \} \\
\tag{Lemma~\ref{lem:inf}} \\
& \sqsupseteq \inf\limits_{\phi \in \LL_\Lambda^1}  (\mu_\delta(\lsem \phi \rsem_{\zeta_D}) \oslash \mu_\gamma(\lsem \phi \rsem_{\zeta_C})) \tag{Lemma~\ref{lem:sum-oslash}} \\
& = \d_{\gamma,\delta}(c,d)
\end{align*}
The equality marked (*) follows from the shape of the $\phi_i$s, the minimality of the family $(\phi_i)_{i \in \omega}$ (w.r.t.~the conditions $\Paths_c(\phi) \subseteq \bigcup\limits_{i \in \omega}\lsem \phi_i \rsem_{\zeta_C}$ and $\Paths_d(\phi) \subseteq \bigcup\limits_{i \in \omega}\lsem \phi_i \rsem_{\zeta_D}$) and the observation that, if the pairwise disjoint family $(\phi_i)_{i \in \omega}$ is such that $\bigcup_{i \in \omega} \Paths_c(\phi_i) \supseteq \Paths_c(\phi)$, then $\sum\limits_{i \in \omega}\mu_\gamma(\Paths_c(\phi_i)) = \sum\limits_{j \in \omega} \mu_\gamma(\Paths_c(\psi_j))$ whenever the pairwise disjoint family $(\psi_j)_{j \in \omega}$ includes the family $(\phi_i)_{i \in \omega}$\footnote{That is, each $\phi_i$ with $i \in \omega$ is equal to some $\psi_j$ with $j \in \omega$.} -- this is because whenever $j \in \omega$ is such that $\psi_j \ne \phi_i$ for all $i \in \omega$, we have $\mu_\gamma(\Paths_c(\psi_j)) = 0$. This concludes the proof.
\end{proof} 

As a corollary, using Remark~\ref{rem:top}, we obtain the following logical characterisation of (a generalised version of) partial trace inclusion.

\begin{cor}
Let $(C,\gamma)$ and $(D,\delta)$ be two $\T_S \circ F$-coalgebras, and let $c \in C$ and $d \in D$. Then, $\ptr_\gamma(c,b) \sqsubseteq \ptr_\delta(d,b)$ for all $b \in B$ if and only if $\llsem \phi \rrsem_\gamma(c) \sqsubseteq \llsem \phi \rrsem_\delta(d)$ for all $\phi \in \mu\LL_\Lambda$.
\end{cor}

A similar result can be stated for partial trace equivalence. This provides a semantic characterisation of the logical equivalence between states induced by the logic $\mu\LL_\Lambda$.

\begin{cor}
Let $(C,\gamma)$ and $(D,\delta)$ be two $\T_S \circ F$-coalgebras, and let $c \in C$ and $d \in D$. Then, $\ptr_\gamma(c,b) = \ptr_\delta(d,b)$ for all $b \in B$ if and only if $\llsem \phi \rrsem_\gamma(c) = \llsem \phi \rrsem_\delta(d)$ for all $\phi \in \mu\LL_\Lambda$.
\end{cor}

\section{Conclusions and Future Work}
\label{conclusions}

We described a generic approach to defining linear-time fixpoint logics for state-based systems that incorporate branching behaviour, and proved the equivalence of the step-wise semantics for these logics with an alternative path-based semantics, akin to those employed by existing linear-time logics. This required generalising standard measure-theory concepts and results to measures valued in a partial semiring. Our approach is uniform in this semiring, can recover standard path-based semantics for existing logics interpreted over non-deterministic and probabilistic systems (see Example~\ref{ex:ltl}), and also instantiates to new settings (resource-aware systems and tree-shaped linear behaviours). The equivalence of the two semantics motivates the use of the term "linear-time" to describe our logics. We also introduced a semantic notion of linear-time distance between states of coalgebras with branching, and showed that this coincides with a logical distance induced by our logics.

Future work will explore generalising these results to coalgebras of type $F \circ \T_S$ (as considered e.g.~in \cite{JACOBS2015}) and beyond. We will also study extensions of our logics that incorporate a notion of \emph{offset} (in the case of the tropical semiring, this would instantiate to resource gains), as considered in \cite{Cirstea19}, as well as extensions to coalgebraic types which incorporate multiple types of branching.

\bibliography{lmcs.bib}

\end{document}